\documentclass[a4paper,USenglish]{lipics-v2021}

\usepackage{cleveref}
\usepackage[font=small,labelfont=bf,justification=centering]{caption}

\usepackage{libertine}
\usepackage{datetime}
\usepackage{url}
\usepackage{paralist}
\usepackage{multicol}

\theoremstyle{definition}

\theoremstyle{definition}

\newcolumntype{P}[1]{>{\centering\arraybackslash}p{#1}}
\newcolumntype{M}[1]{>{\centering\arraybackslash}m{#1}}
\usepackage{arydshln} 

\crefname{section}{\S}{\S\S}
\Crefname{section}{\S}{\S\S}
\crefformat{section}{\S#2#1#3}
\Crefname{assumption}{Assumption}{Assumptions}
\Crefname{invariant}{Invariant}{Invariants}
\Crefname{observation}{Observation}{Observations}

\usepackage{algorithm}
\usepackage[noend]{algpseudocode}
\usepackage{algpseudocode}

\newfloat{module}{htbp}{loa}
\floatname{module}{Module}
\crefname{module}{module}{modules}
\Crefname{module}{Module}{Modules}

\usepackage{xifthen}
\usepackage{stmaryrd}

\usepackage{xspace}
\usepackage{framed}

\usepackage{bold-extra} 

\usepackage{listings}
\newcommand{\strongVal}{\text{strong unanimity}\xspace}

\newcommand{\name}{\textsc{ErrorFreeExt}\xspace}
\newcommand{\slowext}{\textsc{SlowExt}\xspace}
\newcommand{\nameopt}{\textsc{HashExt}\xspace}

\newcommand{\reducecool}{\textsc{Reduce-COOL}\xspace}

\definecolor{lightgray}{gray}{0.90}
\renewenvironment{leftbar}[1][\hsize]
{%
\MakeFramed{\hsize#1\advance\hsize-\width\FrameRestore}%
}
{\endMakeFramed}

\algsetblockdefx[Class]{Class}{EndClass}{}{}[1]{\textbf{class} {#1}:}{\phantom}
\algtext*{EndClass}
\algsetblockdefx[Interface]{Interface}{EndInterface}{}{}{\textbf{Interface:}}{\phantom}
\algtext*{EndInterface}
\algsetblockdefx[Requests]{Requests}{EndRequests}{}{}{\textbf{Requests:}}{\phantom}
\algtext*{EndRequests}
\algsetblockdefx[Logging]{Logging}{EndLogging}{}{}{\textbf{Logging:}}{\phantom}
\algtext*{EndLogging}
\algsetblockdefx[Algorithm]{Algorithm}{EndAlgorithm}{}{}{\textbf{Algorithm:}}{\phantom}
\algtext*{EndAlgorithm}
\algsetblockdefx[Constants]{Constants}{EndConstants}{}{}{\textbf{Constants:}}{\phantom}
\algtext*{EndConstants}
\algsetblockdefx[Upon]{Upon}{EndUpon}{}{}[1]{\textbf{upon} $#1$}{\textbf{end upon}}
\algsetblockdefx[UponReceipt]{UponReceipt}{EndUponReceipt}{}{}[2]{\textbf{upon receipt of} $#1$ \textbf{ from} $#2$}{\textbf{end upon receipt}}
\algsetblockdefx[Structure]{Structure}{EndStructure}{}{}[1]{$#1 = \{$}{$\}$}

\algsetblockdefx[Procedure]{Procedure}{EndProcedure}{}{}[2]{\textbf{procedure} {#1}($#2$)}{\phantom}
\algtext*{EndProcedure}
\algsetblockdefx[Function]{Function}{EndFunction}{}{}[2]{\textbf{function} {#1}($#2$)}{\phantom}
\algtext*{EndFunction}

\algsetblockdefx[Leader]{Leader}{EndLeader}{}{}{\textbf{as} a leader}{\phantom}
\algtext*{EndLeader}
\algsetblockdefx[Replica]{Replica}{EndReplica}{}{}{\textbf{as} a replica}{\phantom}
\algtext*{EndReplica}

\algnewcommand{\BlueComment}[1]{\textcolor{blue}{\hfill\(\triangleright\) #1}}
\algnewcommand{\LineComment}[1]{\textcolor{blue}{ \(\triangleright\) #1}}

\crefname{lstlisting}{listing}{listings}
\Crefname{lstlisting}{Listing}{Listings}

\crefname{code}{line}{lines}
\Crefname{code}{Line}{Lines}

\definecolor{mygreen}{rgb}{0.254,0.572,0.294}
\definecolor{mygray}{rgb}{0.5,0.5,0.5}
\definecolor{myorange}{rgb}{1,0.35,0}
\definecolor{mymauve}{rgb}{0.58,0,0.82}
\definecolor{myblue}{rgb}{0.2,0.4,0.6}

\definecolor{rakos4orange}{RGB}{255,165,0}
\definecolor{rakos4blue}{RGB}{14,48,173}
\definecolor{rakos4lblue}{RGB}{92,172,238}
\definecolor{rakos4dgray}{RGB}{77,77,77}
\definecolor{plainred}{RGB}{211,63,63}
\definecolor{plainorange}{RGB}{221,105,41}

\lstdefinelanguage{Golang}%
  {morekeywords=[1]{package,import,struct,defer,panic,%
     recover,select,var,const,iota, class},%
   morekeywords=[2]{string,uint,uint8,uint16,uint32,uint64,int,int8,int16,%
     int32,int64,bool,float32,float64,complex64,complex128,byte,rune,uintptr,%
     error,interface,node},%
   morekeywords=[3]{map,slice,make,new,nil,len,cap,copy,close,%
     delete,append,real,imag,complex,chan,},%
   morekeywords=[4]{break,continue,goto,switch,case,fallthrough,%
    default,},%
   morekeywords=[5]{Println,Printf,Error,Send},%
   sensitive=true,%
   morecomment=[l]{//},%
   morecomment=[s]{/*}{*/},%
   morestring=[b]",%
   morestring=[s]{`}{`},%
   }

\title{Efficient Signature-Free Validated Agreement}

\hideLIPIcs  

\titlerunning{Efficient Signature-Free Validated Agreement}

\author{Pierre Civit}{Ecole Polytechnique Fédérale de Lausanne (EPFL), Switzerland}{}{}{}
\author{Muhammad Ayaz Dzulfikar}{NUS Singapore, Singapore}{}{}{}
\author{Seth Gilbert}{NUS Singapore, Singapore}{}{}{}
\author{Rachid Guerraoui}{Ecole Polytechnique Fédérale de Lausanne (EPFL), Switzerland}{}{}{}
\author{Jovan Komatovic}{Ecole Polytechnique Fédérale de Lausanne (EPFL), Switzerland}{}{}{}
\author{Manuel Vidigueira}{Ecole Polytechnique Fédérale de Lausanne (EPFL), Switzerland}{}{}{}
\author{Igor Zablotchi}{Mysten Labs, Switzerland}{}{}{}

\nolinenumbers 

\authorrunning{Civit, Dzulfikar, Gilbert, Guerraoui, Komatovic, Vidigueira, and Zablotchi} 

\Copyright{Pierre Civit, Muhammad Ayaz Dzulfikar, Seth Gilbert, Rachid Guerraoui, Jovan Komatovic, Manuel Vidigueira, and Igor Zablotchi} 



\begin{document}

\maketitle

\begin{abstract}
Byzantine agreement enables $n$ processes to agree on a common $L$-bit value, despite up to $t > 0$ arbitrary failures. 
A long line of work has been dedicated to improving the bit complexity of Byzantine agreement in synchrony. 
This has culminated in COOL, an error-free (deterministically secure against a computationally unbounded adversary) solution that achieves $O( nL + n^2 \log n )$ worst-case bit complexity (which is optimal for $L \geq n \log n$ according to the Dolev-Reischuk lower bound).
COOL satisfies \strongVal: if all correct processes propose the same value, only that value can be decided.
Whenever correct processes do not agree \emph{a priori} (there is no unanimity), they may decide a default value $\bot$ from COOL.

Strong unanimity is, however, not sufficient for today's state machine replication (SMR) and blockchain protocols.
These systems value progress and require a decided value to always be \emph{valid} (according to a predetermined predicate), excluding default decisions (such as $\bot$) even in cases where there is no unanimity a priori.
\emph{Validated Byzantine agreement} satisfies this property (called \emph{external validity}).
Yet, the best error-free (or even signature-free) validated agreement solutions achieve only $O(n^2L)$ bit complexity, a far cry from the $\Omega(nL+n^2)$ Dolev-Reischuk lower bound.
Is it possible to bridge this complexity gap?

We answer the question affirmatively.
Namely, we present two new synchronous algorithms for validated Byzantine agreement, \nameopt and \name, with different trade-offs.
Both algorithms are (1) signature-free, (2) optimally resilient (tolerate up to $t < n / 3$ failures), and (3) early-stopping (terminate in $O(f+1)$ rounds, where $f \leq t$ denotes the actual number of failures).
On the one hand, \nameopt uses only hashes and achieves $O(nL + n^3\kappa)$ bit complexity, which is optimal for $L \geq n^2\kappa$ (where $\kappa$ is the size of a hash).
On the other hand, \name is error-free, using no cryptography whatsoever, and achieves $O\big((nL + n^2)\log n \big)$ bit complexity, which is near-optimal for any $L$.
\end{abstract}

\section{Introduction} \label{section:introduction}

Byzantine agreement~\cite{Lamport1982} is arguably the most important problem of distributed computing.
It lies at the heart of state machine replication (SMR)~\cite{adya2002farsite,CL02,kotla2004high,abd2005fault,amir2006scaling,kotla2007zyzzyva,veronese2011efficient,malkhi2019flexible,momose2021multi} and blockchain systems~\cite{luu2015scp,buchman2016tendermint,solida,chen2016algorand,abraham2016solidus,CGL18,correia2019byzantine}.
Additionally, Byzantine agreement plays an essential role in cryptographic protocols such as multi-party computation~\cite{goldreich1987play,BGW88,Keller2023,Beerliova-Trubiniova2007,Garg2019,Chandran2015}.

Byzantine agreement operates among $n$ processes, out of which up to $t > 0$ can be corrupted by the adversary.
A corrupted process is said to be \emph{faulty} and can behave arbitrarily; a non-faulty process is said to be \emph{correct} and follows the prescribed protocol.
Let $\mathsf{Value}$ denote the set of $L$-bit values.  
(As this paper is concerned with multi-valued Byzantine agreement, we set no restrictions on the cardinality of the $\mathsf{Value}$ set.)
During the agreement protocol, each process \emph{proposes} exactly one value, and eventually the protocol outputs a single \emph{decision}, as per the following interface:
\begin{compactitem}
    \item \textbf{request} $\mathsf{propose}(v \in \mathsf{Value}):$ a process proposes an $L$-bit value $v$.

    \item \textbf{indication} $\mathsf{decide}(v' \in \mathsf{Value})$: a process decides an $L$-bit value $v'$.
\end{compactitem}
Intuitively, Byzantine agreement ensures that all correct processes agree on the same \emph{admissible} value.
(We formally define the properties of Byzantine agreement in the later part of this section.)

\smallskip
\noindent \textbf{Practical notion of value-admissibility.}
A critical question in designing practical Byzantine agreement algorithms is which values should be considered admissible.  Traditionally, Byzantine agreement algorithms treated the proposals of correct processes as admissible.
Consequently, they have focused on properties like \emph{\strongVal}~\cite{abraham2019asymptotically,Chen2021,Nayak2020}: if every correct process proposes the same value $v$, then $v$ is the only possible decision. 
Notice that in such cases, if even one correct process proposes a value different from the (same) value held by all other $n-1$ processes, it is perfectly legal to decide some default ``null op'' value (e.g., $\bot$); it is also perfectly legal to decide a value that is ``nonsense'' from the perspective of the underlying application.
Thus, unless all correct processes agree \emph{a priori}, Byzantine agreement algorithms with \strongVal are not guaranteed to make any ``real'' progress.

Many modern applications may require a stronger requirement: even if correct processes propose different values, the resulting decision should still adhere to some \emph{validity} test, ensuring that the decision is not ``wasted''.
Such a condition is usually called \emph{external validity}~\cite{Cachin2001,lamport2019concurrency,LL0W20,spiegelman2020search,abraham2019asymptotically, ZZZDHWL22,GelashviliKSSX22,LL022,song2024flexbft}: any decided value must be valid according to a predetermined logical predicate.
We underline that the external validity property is prevalent in today's blockchain systems. 
Indeed, as long as a produced block is valid (e.g., no double-spending), the block can safely be added to the chain (irrespectively of who produced it).\footnote{Let us underline that real-world blockchain systems might be concerned with \emph{fairness}, thus making the question of ``who produced a block'' important. However, this work does not focus on fairness (or any similar topic \cite{GLTZ24,K0GJ20}).}

\smallskip
\noindent \textbf{Synchronous validated agreement.}
We study \emph{validated agreement}, a variant of the Byzantine agreement problem satisfying the external validity property, in the standard synchronous setting.
Formally, let $\mathsf{valid}: \mathsf{Value} \to \{ \mathit{true}, \mathit{false} \}$ be any predetermined predicate.
Importantly, correct processes propose valid values.
The following properties are guaranteed by validated agreement:
\begin{compactitem}
    \item \emph{Agreement:} No two correct processes decide different values.

    \item \emph{Integrity:} No correct process decides more than once.

    \item \emph{Termination:} All correct processes eventually decide.

    \item \emph{Strong unanimity:} If all correct processes propose the same value $v$, then no correct process decides any value $v' \neq v$.

    \item \emph{External validity:} If a correct process decides a value $v$, then $\mathsf{valid}(v) = \mathit{true}$.
\end{compactitem}
We underline that validated agreement algorithms usually do not satisfy \strongVal (but only external validity).
Additionally, we emphasize that obtaining an agreement algorithm $\mathcal{A}^{\star}$ that satisfies \emph{both} strong unanimity and external validity is straightforward given (1) an agreement algorithm $\mathcal{A}_1$ satisfying only \strongVal, and (2) an agreement algorithm $\mathcal{A}_2$ satisfying only external validity.
Indeed, to obtain $\mathcal{A}^{\star}$, processes run $\mathcal{A}_1$ and $\mathcal{A}_2$ in parallel.
Then, processes decide (1) the value of $\mathcal{A}_1$ if that value is valid, or (2) the value of $\mathcal{A}_2$ otherwise.

\smallskip
\noindent \textbf{Complexity of synchronous validated agreement.}
There exist two dominant worst-case complexity metrics when analyzing any synchronous validated agreement algorithm: (1) \emph{the bit complexity}, the total number of bits correct processes send, and (2) \emph{the round complexity}, the number of synchronous rounds it takes for all correct processes to decide (and halt).
The lower bound on the bit complexity of validated agreement is $\Omega(nL + n^2)$: (1) the ``$nL$'' term comes from the fact that each correct process needs to receive the decided value, and (2) the ``$n^2$'' term comes from the seminal Dolev-Reischuk bound~\cite{dolev1985bounds} stating that even agreeing on a single bit requires $\Omega(n^2)$ exchanged bits.
We emphasize that the $\Omega(nL + n^2)$ lower bound holds even in \emph{failure-free} executions in the signature-free world (with signatures, the bound does not hold \cite{spiegelman2020search}).
The lower bound on the round complexity is $\Omega(f+1)$~\cite{dolev1990early}, where $f \leq t$ denotes the \emph{actual} number of failures.
If an algorithm achieves $O(f+1)$ round complexity, it is said that the algorithm is \emph{early-stopping}.\footnote{We consider only \emph{asymptotic} early-stopping (as in~\cite{Lenzen2022}) instead of \emph{strict} early stopping (as in~\cite{dolev1990early}) that requires termination in exactly $f + 2$ rounds.}

\smallskip
\noindent \textbf{State-of-the-art.}
The most efficient known validated agreement algorithm is \textsc{Ada-Dare}~\cite{ada_dare_to_appear_podc24}.
\textsc{Ada-Dare} achieves $O(nL + n^2 \kappa)$ bit complexity (optimal for $L > n\kappa$), where $\kappa$ denotes a security parameter.
However, \textsc{Ada-Dare} internally utilizes threshold signatures~\cite{Shoup00}.
(We emphasize that if $t < n / 3$, some partially synchronous authenticated algorithms~\cite{YPAKT22,Camenisch2022} can trivially be adapted to achieve $O(nL + n^2\kappa)$ bit complexity in synchrony; \textsc{Ada-Dare} tolerates up to $t < n / 2$ failures.)
Perhaps surprisingly, the best \emph{signature-free} validated agreement algorithms~\cite{Lenzen2022,berman1992bit,CoanW92,Chen2021} still achieve only $O(n^2L)$ bit complexity, a far cry from the $\Omega(nL + n^2)$ lower bound.
  
The fact that no efficient signature-free validated agreement is known becomes even more surprising when considering that optimal signature-free algorithms exist for the ``traditional'' Byzantine agreement problem.
COOL~\cite{Chen2021} is a Byzantine agreement algorithm satisfying (only) \strongVal while exchanging $O(nL + n^2 \log n)$ bits.
Although it was not the goal of the COOL algorithm, COOL can trivially achieve early-stopping (by internally utilizing an early-stopping binary agreement such as~\cite{Lenzen2022}).
In addition, COOL is optimally resilient (tolerates up to $t < n / 3$ failures).
Importantly, COOL uses no cryptography whatsoever: we say that COOL is \emph{error-free} as it is deterministically secure against a computationally unbounded adversary.

Is there a fundamental complexity gap between external validity and \strongVal in the signature-free world?
Can signature-free validated agreement be solved efficiently in synchrony?
These are the questions we study in this paper.

\subsection{Contributions}

In this paper, we present the first validated agreement algorithms achieving $o(n^2L)$ bit complexity \emph{without} signatures:
\begin{compactitem}
    \item First, we introduce \nameopt, a hash-based algorithm that exchanges $O(nL + n^3\kappa)$ bits (optimal for $L \geq n^2\kappa$), where $\kappa$ denotes the size of a hash.

    \item Second, we provide \name, an error-free (i.e., cryptography-free) solution that achieves $O\big( (nL + n^2) \log n \big)$ bit complexity and is thus nearly-optimal.
\end{compactitem}
Importantly, both \nameopt and \name are (1) optimally resilient (tolerate up to $t < n / 3$ failures), and (2) early-stopping (terminate in $O(f + 1)$ synchronous rounds).
A comparison of our new algorithms with the state-of-the-art can be found in \Cref{fig:complexities_summary_short}.

\begin{table}[th]
\centering
\footnotesize
\begin{tabular}{ |P{2.7cm}|P{1.5cm}|P{3.5cm}|P{1.5cm}|P{2.1cm}|  }
 \hline
 Protocol & Validity & Bit complexity & Resilience & Cryptography \\
 \hline
 \hline
 COOL~\cite{Chen2021,Lenzen2022}                 & S & $O(nL + n^2 \log n)$ & $n > 3t$ & None \\
 Parallel COOL \cite{Chen2021,Lenzen2022}                 & IC $\to$ (S + E) & $O(n^2L + n^3 \log n)$ & $n > 3t$ & None \\

$\textsc{Ada-Dare}_{ic}$~\cite{ada_dare_to_appear_podc24}                 & IC $\to$ (S + E) & $O(n^2L + n^2\kappa)$ & $n > 2t$ & Threshold Sign. \\

 $\textsc{Ada-Dare}_{su}$~\cite{ada_dare_to_appear_podc24}                 & S + E & $O(nL + n^2\kappa)$ & $n > 2t$ & Threshold Sign. \\

\textbf{\nameopt}          & S + E      & $O(nL + n^3 \kappa)$ & $n > 3t$ & Hash \\
 
 \textbf{\name}          & S + E      & $O\big( (nL + n^2)\log n \big)$ & $n > 3t$ & None \\
 \hline
 \hline
 Lower bound~\cite{dolev1985bounds,WBAlowerBound} & Any & $\Omega(nL + n^2)$ & $t \in \Omega(n)$ & Any \\
 \hline
\end{tabular}
    \caption{Performance of deterministic synchronous agreement algorithms with $L$-bit values and $\kappa$-bit security parameter.
    S stands for ``\strongVal'', E stands for ``external validity'', and IC stands for ``interactive consistency'' (where processes agree on the proposals of all processes).
    (There exists a trivial reduction from IC to S + E, where each correct process decides the most represented valid value in the decided vector. Hence, we write that IC implies S + E.)
    All considered algorithms are early-stopping, except for $\textsc{Ada-Dare}_{ic}$ and $\textsc{Ada-Dare}_{su}$ (whose goal was not early-stopping).
    Randomized algorithms are discussed in \Cref{section:related_work_extended}.
    }
\label{fig:complexities_summary_short}
\end{table}

\subsection{Overview \& Technical Challenges} \label{subsection:technical_challenges}

\noindent \textbf{Why is efficient validated agreement hard?}
To solve the validated agreement problem (i.e., to satisfy external validity), a decided value must be valid.
Therefore, a validated agreement algorithm needs to ensure that it is operating on (or converging to) a valid value.
If the value (in its entirety) is attached to every message, satisfying external validity is (relatively) simple: each message can be individually validated and invalid messages can be ignored.
Unfortunately, attaching an $L$-bit value to each message is inherently expensive, yielding a sub-optimal bit complexity of $\Omega(n^2L)$.

To avoid attaching an $L$-bit value to each message, the most efficient solutions to validated agreement (designed for arbitrary-sized values) involve coding techniques, where an $L$-bit value is split into $n$ different shares of $O(\frac{L}{n} + \log n)$ size.  The goal is to (somehow) reach agreement on a valid value using $O(n^2)$ messages of $O(\frac{L}{n} + \log n)$ bits, for a total of $O(nL + n^2\log n)$ exchanged bits.
However, this ``coding-based'' design introduces a new challenge.
How can a process that only holds one share (or constantly many shares) know that the corresponding value is valid?
For example, to check if a split value $v$ is valid, correct processes might attempt to reconstruct it, expending $O(nL+n^2\log n)$ bits in the process (as reconstruction is expensive).  
Since there may be (in the worst case) up to $t \in \Omega(n)$ invalid values (from as many faulty processes), this reconstruction process might have to be repeated many times before a valid value is found, resulting in (say) sub-optimal $O(n^2L+n^3\log n)$ total communication.

\smallskip
\noindent \textbf{Overview of \nameopt.}
To overview \nameopt's design, we first revisit how efficient signature-based validated agreement is solved (see, e.g.,~\cite{ada_dare_to_appear_podc24}).
In the signature-based paradigm, efficient validated agreement algorithms adopt the following approach: 
(1) First, each process disseminates its value (using coding techniques) and obtains a \emph{proof of retriveability} (PoR). A PoR is a cryptographic object containing a digest (of a value) and proving that (i) the pre-image of the digest can be retrieved by all correct processes, and (ii) the pre-image of the digest is valid.
(2) Second, processes agree on a single PoR.
(3) Third, processes retrieve a value corresponding to the agreed-upon PoR.
Importantly, each PoR must be ``self-certifying'': once a correct process obtains an alleged PoR, the process must be able to determine if the PoR is valid to be sure that if this PoR gets decided in the second step, a valid value can be retrieved.
That is why PoRs are usually implemented using signatures: if a PoR contains a \emph{signature-based certificate}, processes can be confident in its validity.
Due to this ``self-certifying'' nature of PoRs, it seems challenging to adapt them to the signature-free world.

To design a hash-based validated agreement algorithm \nameopt, we (roughly) follow the aforementioned three-step approach with one fundamental difference: \nameopt utilizes \emph{implicit} (``non-self-certyfing'') PoRs.
Given any observed digest $d$, a correct process executing \nameopt can determine if (1) the pre-image $v$ of digest $d$ can be retrieved, and (2) $v$ is valid.
There is no \emph{proof} that the valid pre-image can be retrieved -- only the protocol design ensures this guarantee.

\smallskip
\noindent \textbf{Overview of \name.}
To implement \name, our error-free (cryptography-free) near-optimal solution, we rely on a recursive structure -- carefully adapting to long values the recursive design proposed by~\cite{berman1992bit,CoanW92,Lenzen2022,Momose2021} that is only concerned with constant-sized values. 
At each recursive iteration with $n$ processes, processes are statically partitioned into two halves that run the algorithm among $n / 2$ processes.
Moreover, each recursive iteration exhibits ``additional work'' through the \emph{graded consensus}~\cite{AW23,Abraham2022} primitive.
Intuitively, the graded consensus primitive reconciles decisions made by two distinct halves to ensure that all processes agree on a unique valid value.
Due to the recursive nature of \name, its bit complexity depends on the complexity of graded consensus.
To obtain \name's near-optimal $O\big( (nL + n^2) \log n \big)$ bit complexity, we observe that a graded consensus algorithm with $O(nL + n^2 \log n)$ bits can be derived from the ``reducing'' technique introduced by the previously mentioned COOL~\cite{Chen2021} protocol.\footnote{A similar observation has recently been made for (balanced) synchronous \emph{gradecast}, a sender-oriented counterpart to graded consensus~\cite{AsharovChandramouli24}.}

\smallskip
\noindent \textbf{Roadmap.}
We define the system model and introduce some preliminaries in \Cref{section:preliminaries}.
We present \nameopt in \Cref{section:optimal}, whereas \name is introduced in \Cref{section:general_framework}.
We discuss related work in \Cref{section:related_work_extended}.
Finally, we conclude in \Cref{section:conclusion}.
Omitted pseudocode and detailed proofs are relegated to the optional appendix.
\section{System Model \& Preliminaries} \label{section:preliminaries}


\subsection{System Model} \label{subsection:system_model}

\noindent\textbf{Processes.}
We consider a static set $\Pi = \{p_1, p_2, ..., p_n\}$ of $n$ processes, where each process acts as a deterministic state machine.
Our \nameopt (resp., \name) algorithm implements validated agreement against a computationally bounded (resp., unbounded) adversary that can corrupt up to $t < n / 3$ processes at any time during an execution.
(We underline that no signature-free agreement algorithm can tolerate $n / 3$ or more failures~\cite{LSP82}, disregarding the restricted-resource model~\cite{Garay2020} that allows for a higher corruption threshold.)
A corrupted process is said to be \emph{faulty}; a non-faulty process is said to be \emph{correct}.
We denote by $f \leq t$ the actual number of faulty processes; we emphasize that $f$ is not known.

\smallskip
\noindent \textbf{Stopping.} Each correct process can invoke a special $\mathsf{stop}$ request while executing any protocol. 
Once a correct process stops executing a protocol, it ceases taking any steps (e.g., sending and receiving messages).

\smallskip
\noindent \textbf{Communication network.}
Processes communicate by exchanging messages over an authenticated point-to-point network.
The communication network is reliable: if a correct process sends a message to a correct process, the message is eventually received.

\smallskip
\noindent\textbf{Synchrony.}
We assume the standard synchronous environment in which the computation unfolds in synchronous $\delta$-long rounds, where $\delta$ denotes the known upper bound on message delays.
In each round $1, 2, ... \in \mathbb{N}$, each process (1) performs (deterministic) local computations, (2) sends (possibly different) messages to (a subset of) the other processes, and (3) receives the messages sent to it by the end of the round. 

\subsection{Complexity Measures} \label{subsection:complexity_measures}
Let $\mathsf{Agreement}$ be any synchronous validated agreement algorithm, and let $\mathcal{E}(\mathsf{Agreement})$ denote the set of $\mathsf{Agreement}$'s executions.
Let $\alpha \in \mathcal{E}(\mathsf{Agreement})$ be any execution.
The bit complexity of $\alpha$ is the number of bits correct processes collectively send throughout $\alpha$.
The \emph{bit complexity} of $\mathsf{Agreement}$ is then defined as 
\[
\max_{\alpha \in \mathcal{E}(\mathsf{Agreement})}\bigg\{\text{the bit complexity of } \alpha\bigg\}.
\]
Similarly, the latency complexity of $\alpha$ is the time it takes for all correct processes to decide and stop in $\alpha$.
The \emph{latency complexity} of $\mathsf{Agreement}$ is then defined as 
\[
\max_{\alpha \in \mathcal{E}(\mathsf{Agreement})}\bigg\{\text{the latency complexity of } \alpha\bigg\}.
\]
We say that $\mathsf{Agreement}$ satisfies \emph{early stopping} if and only if the latency complexity of $\mathsf{Agreement}$ belongs to $O\big( (f + 1)\delta \big)$.
Note that the maximum number of rounds $\mathsf{Agreement}$ requires to decide -- the \emph{round complexity} of $\mathsf{Agreement}$ -- is equal to the latency complexity of $\mathsf{Agreement}$ divided by $\delta$.
Throughout the paper, we use the latency and round complexity interchangeably.

\subsection{Building Blocks} \label{subsection:preliminaries}
This subsection overviews building blocks utilized in both \nameopt and \name.

\smallskip
\noindent \textbf{Reed-Solomon codes.}
\nameopt and \name rely on Reed-Solomon (RS) codes~\cite{reed1960}.
We use $\mathsf{RSEnc}$ and $\mathsf{RSDec}$ to denote RS' encoding and decoding algorithms.
In brief, $\mathsf{RSEnc}(M, m, k)$ takes as input a message $M$ consisting of $k$ symbols, treats it as a polynomial of degree $k - 1$, and outputs $m$ evaluations of the corresponding polynomial.
Similarly, $\mathsf{RSDec}(k, r, T)$ takes as input a set of symbols $T$ (some of the symbols might be incorrect) and outputs a degree $k - 1$ polynomial (i.e., $k$ symbols) by correcting up to $r$ errors (incorrect symbols) in $T$.
Note that $\mathsf{RSDec}$ can correct up to $r$ errors in $T$ and output the original message given that $|T| \geq k + 2r$~\cite{macwilliams1977theory}.
Importantly, the bit-size of an RS symbol obtained by the $\mathsf{RSEnc}(M, m, k)$ algorithm is $O(\frac{|M|}{k} + \log m ) $, where $|M|$ denotes the bit-size of the message $M$.

\smallskip
\noindent \textbf{Graded consensus.} 
Both \nameopt and \name make extensive use of the graded consensus primitive~\cite{AW23,Abraham2022} (also known as Adopt-Commit~\cite{delporte2021weakest}), whose formal specification is given in \Cref{mod:graded-consensus}.
In brief, graded consensus allows processes to propose their input value from the $\mathsf{GC\_Value}$ set and decide on some value from the $\mathsf{GC\_Value}$ set with some binary grade.
The graded consensus primitive ensures agreement among correct processes only if some correct process decides a value with (higher) grade $1$.
If no correct process decides with grade $1$, graded consensus allows correct processes to disagree.
(Thus, graded consensus is a weaker problem than validated agreement.) 
\nameopt employs the graded consensus primitive on hash values ($\mathsf{GC\_Value} \equiv \text{the set of all hash values}$).
On the other hand, \name utilizes graded consensus on values proposed to validated agreement ($\mathsf{GC\_Value} \equiv \mathsf{Value}$).

\begin{module}[ht]
\caption{Graded consensus}
\label{mod:graded-consensus}
\footnotesize
\begin{algorithmic}[1]

\Statex \textbf{Events:}
\begin{compactitem}
    \item \emph{request} $\mathsf{propose}(v \in \mathsf{GC\_Value})$: a process proposes a value $v \in \mathsf{GC\_Value}$.
    \item \emph{indication} $\mathsf{decide}(v' \in \mathsf{GC\_Value}, g' \in \{0, 1\})$: a process decides a value $v' \in \mathsf{GC\_Value}$ with a grade $g'$.
\end{compactitem}

\medskip 
\Statex \textbf{Assumed behavior:} 
\begin{compactitem}
    \item Every correct process proposes exactly once.
    \item All correct processes propose simultaneously (i.e., in the same round).
    \textcolor{orange}{(We revisit this assumption for the graded consensus primitive employed in \name; see \Cref{subsection:name_pseudocode}.)}
\end{compactitem}

\medskip 
\Statex \textbf{Properties:}

\begin{compactitem}
    \item \emph{Strong unanimity:} If all correct processes propose the same value $v$ and a correct process decides a pair $(v', g')$, then $v' = v$ and $g' = 1$.

    \item \emph{Justification:} If a correct process decides a pair $(v', \cdot)$, then $v'$ was proposed by a correct process.

    
    \item \emph{Consistency:} If any correct process decides a pair $(v, 1)$, then no correct process decides any pair $(v' \neq v, \cdot)$.
    
    \item \emph{Integrity:} No correct process decides more than once.
    
    \item \emph{Termination:} All correct processes decide simultaneously (i.e., in the same round).
    \textcolor{orange}{(The ``simultaneous'' termination is revisited in the graded consensus primitive employed in \name; see \Cref{subsection:name_pseudocode}.)}
\end{compactitem}

\end{algorithmic}
\end{module}
\section{\nameopt: Optimal Early-Stopping Hash-Based Solution} \label{section:optimal}

In this section, we present \nameopt, our hash-based validated Byzantine agreement solution that achieves $O(nL + n^3 \kappa)$ bit complexity, which is optimal for $L \geq n^2 \kappa$ ($\kappa$ denotes the size of a hash value).
Additionally, \nameopt is (1) optimally resilient as it tolerates up to $t < n / 3$ faults, and (2) early-stopping as it terminates in $O\big( (f + 1) \delta \big)$ time (i.e., $O(f + 1)$ synchronous rounds). 

We start by introducing the building blocks of \nameopt (\Cref{subsection:nameopt_building_blocks}).
Then, we present \nameopt's pseudocode (\Cref{subsection:nameopt_pseudocode}).
Finally, we present a proof sketch of \nameopt's correctness and complexity (\Cref{subsection:nameopt_proof_sketch}).
We relegate a proof of \nameopt's correctness and complexity to \Cref{section:nameopt_proof}.

\subsection{Building Blocks} \label{subsection:nameopt_building_blocks}

\noindent \textbf{Digests.}
We assume a collision-resistant function $\mathsf{digest}: \mathsf{Value} \to \mathsf{Digest} \equiv \{0, 1\}^{\kappa}$, where $\kappa$ is a security parameter.
Concretely, the $\mathsf{digest}(v \in \mathsf{Value})$ function performs the following steps: (1) it encodes value $v$ into $n$ RS symbols $[m_1, m_2, ..., m_n] \gets \mathsf{RSEnc}(v, n, t + 1)$; (2) it aggregates $[m_1, m_2, ..., m_n]$ into an accumulation value $z_v$ using the Merkle-tree-based (i.e., hash-based) cryptographic accumulator~\cite{merkle-tree-crypto87} (see \Cref{section:cryptographic_preliminaries}); (3) it returns $z_v$.
Note that, as we employ hash-based Merkle trees, an accumulation value $z_v$ is a hash.
The formal definition of the $\mathsf{digest}(\cdot)$ function can be found in \Cref{section:cryptographic_preliminaries}. 

\smallskip
\noindent \textbf{Data dissemination.}
The formal specification of the data dissemination primitive is given in \Cref{mod:hashfin}.
Intuitively, the data dissemination primitive allows all correct processes to obtain the same value $v^{\star}$ assuming that (1) all correct processes a priori agree on the digest $d^{\star}$ of value $v^{\star}$ (even if processes do not know the pre-image $v^{\star}$ of $d^{\star}$ a priori), and (2) at least one correct process initially holds the pre-image $v^{\star}$.
We relegate the implementation of the data dissemination primitive to \Cref{subsection:hashfin_proof}.
In brief, the implementation heavily relies on Merkle-tree-based accumulators (see \Cref{section:cryptographic_preliminaries}) and it exchanges $O(nL + n^2 \kappa \log n)$ bits while terminating in $2\delta$ time.

\begin{module}
\caption{Data dissemination}
\label{mod:hashfin}
\footnotesize
\begin{algorithmic}[1]

\Statex \textbf{Events:}

\Statex 
\begin{compactitem}
    \item \emph{request} $\mathsf{input}(v \in \mathsf{Value} \cup \{\bot\}, d \in \mathsf{Digest})$: a process inputs a value $v$ (or $\bot$) and a digest $d$.

    \item \emph{request} $\mathsf{output}(v' \in \mathsf{Value})$: a process outputs a value $v'$.
\end{compactitem}

\medskip 
\Statex \textbf{Assumed behavior:} 
\begin{compactitem}
    \item All correct processes input a pair.
    We underline that correct processes might not input their values simultaneously (i.e., at the exact same round).

    \item No correct process stops unless it has previously output a value.

    \item There exists a value $v^{\star} \in \mathsf{Value}$ ($v^{\star} \neq \bot$) and a digest $d^{\star} = \mathsf{digest}(v^{\star})$ such that:
    \begin{compactitem}
        \item If any correct process inputs a pair $(v \in \mathsf{Value}, \cdot)$, then $v = v^{\star}$.

        \item If any correct process inputs a pair $(\cdot, d \in \mathsf{Digest})$, then $d = d^{\star}$.

        \item At least one correct process inputs a pair $(v^{\star}, d^{\star})$.
    \end{compactitem}

\end{compactitem}

\medskip 
\Statex \textbf{Properties:}
\Statex 
\begin{compactitem}
    \item \emph{Safety:} If any correct process outputs a value $v$, then $v = v^{\star}$.

    \item \emph{Liveness:} Let $\tau$ be the first time by which all correct processes have input a pair.
    Then, every correct process outputs a value by time $\tau + 2\delta$.

    \item \emph{Integrity:} No correct process outputs a value unless it has previously input a pair.
\end{compactitem}
\end{algorithmic}
\end{module}

\subsection{Pseudocode} \label{subsection:nameopt_pseudocode}

The pseudocode of \nameopt is given in \Cref{algorithm:optimal_trivial}.

\smallskip
\noindent \textbf{Key idea.}
The crucial idea behind \nameopt is to ensure that all correct processes agree on a digest $d^{\star}$ of a \emph{valid} value $v^{\star}$ such that at least \emph{one} correct process knows the pre-image $v^{\star}$ of $d^{\star}$.
To solve validated agreement, it then suffices to utilize the data dissemination primitive (see \Cref{mod:hashfin}): if (1) all correct processes input the same digest $d^{\star}$, and (2) at least one correct process inputs the pre-image $v^{\star}$ of $d^{\star}$, then all correct processes agree on the (valid) value $v^{\star}$.
Given that the data dissemination primitive exchanges $O(nL + n^2 \kappa \log n)$ bits and terminates in $2$ rounds, \nameopt dedicates $O(nL + n^3 \kappa)$ bits and $O(f + 1)$ rounds to agreeing on digest $d^{\star}$.

\begin{algorithm} [hp]
\caption{\nameopt: Pseudocode (for process $p_i$)}
\label{algorithm:optimal_trivial}
\begin{algorithmic} [1]
\footnotesize

\State \textbf{Uses:}
\State \hskip2em Graded consensus, \textbf{instances} $\mathcal{GC}_1[V]$, $\mathcal{GC}_2[V]$, for each view $V \in [1, t + 1]$  \BlueComment{bits: $O(n^2\kappa)$; rounds: $2$}
\State \hskip2em Data dissemination, \textbf{instance} $\mathcal{DD}$ \BlueComment{bits: $O( nL + n^2 \kappa \log n )$; rounds: $2$}

\medskip
\State \textbf{Local variables:}
\State \hskip2em $\mathsf{Value}$ $v_i \gets p_i$'s proposal
\State \hskip2em $\mathsf{Digest}$ $\mathit{locked}_i \gets \bot$ \BlueComment{locked digest} \label{line:variable_locked}
\State \hskip2em $\mathsf{Digest}$ $\mathit{vote}_i \gets \bot$ \BlueComment{digest to be voted for} \label{line:variable_vote}
\State \hskip2em $\mathsf{View}$ $\mathit{committeed\_view}_i \gets \mathit{\bot}$
\State \hskip2em $\mathsf{Map}( \mathsf{Digest} \to \mathsf{Value})$ $\mathit{known\_values}_i \gets \{\bot, \bot, ..., \bot\}$ \BlueComment{values corresponding to digests} \label{line:variable_known}
\State \hskip2em $\mathsf{Map}( \mathsf{View} \to \mathsf{Set}(\mathsf{Digest}) )$ $\mathit{accepted}_i \gets \{\emptyset, \emptyset, ..., \emptyset \}$ \BlueComment{accepted digests per view} \label{line:variable_accepted}

\medskip
\State \textbf{-- Task 1 --}
\State \textbf{for each} view $V \in [1, t + 1]$:

\smallskip
\State \hskip2em \textbf{if} $\mathit{committed\_view}_i \neq \bot$ and $\mathit{commited\_view}_i + 1 = V$: complete the view after $6$ synchronous rounds \label{line:execute_only_the_next_view}

\State \hskip2em \textbf{if} $\mathit{commited\_view}_i \neq \bot$ and $V > \mathit{committed\_view}_i + 1$: do not execute the view \label{line:no_executing_the_next_next_view}

\smallskip
\State \hskip2em \emph{Step 1 of view $V$:} \BlueComment{$2$ synchronous rounds}
\State \hskip4em Let $(d_1 \in \mathsf{Digest} \cup \{\bot\}, g_1 \in \{0, 1\}) \gets \mathcal{GC}_1[V].\mathsf{propose}(\mathit{locked}_i)$ \label{line:propose_gc1}

\smallskip
\State \hskip2em \emph{Step 2 of view $V$:} \BlueComment{$2$ synchronous round}
\State \hskip4em \textbf{if} $p_i = \mathsf{leader}(V)$:
\State \hskip6em \textbf{if} $d_1 \neq \bot$: \BlueComment{check if a non-$\bot$ digest is decided from $\mathcal{GC}_1[V]$}
\State \hskip8em \textbf{broadcast} $d_1$ \BlueComment{broadcast a non-$\bot$ digest decided from $\mathcal{GC}_1[V]$} \label{line:broadcast_d1}
\State \hskip6em \textbf{else:}
\State \hskip8em \textbf{broadcast} $v_i$ \BlueComment{broadcast the proposed value} \label{line:broadcast_vi}

\smallskip
\State \hskip4em \textbf{if} $d_1 \neq \bot$ and $g_1 = 1$: \label{line:check_support_1}
\State \hskip6em \textbf{broadcast} $\langle \textsc{support}, d_1 \rangle$ \label{line:broadcast_happy_1}
\State \hskip4em \textbf{else:}
\State \hskip6em \textbf{if} $d_l \in \mathsf{Digest}$ is received from $\mathsf{leader}(V)$ and a view $V' < V$ exists with $d_l \in \mathit{accepted}[V']$: \label{line:check_happy_2}
\State \hskip8em \textbf{broadcast} $\langle \textsc{support}, d_l \rangle$ \label{line:broadcast_happy_2}
\State \hskip6em \textbf{else if} $v_l \in \mathsf{Value}$ is received from $\mathsf{leader}(V)$ such that $\mathsf{valid}(v_l) = \mathit{true}$:  \label{line:check_validity_received_from_leader}
\State \hskip8em $\mathit{known\_values}[\mathsf{digest}(v_l)] \gets v_l$ \label{line:update_known_values}
\State \hskip8em \textbf{broadcast} $\langle \textsc{support}, \mathsf{digest}(v_l) \rangle$ \label{line:broadcast_happy_3}

\smallskip
\State \hskip2em \emph{Step 3 of view $V$:} \BlueComment{0 synchronous round (only local computation)}
\State \hskip4em \textbf{if} exists $d \in \mathsf{Digest}$ such that a $\langle \textsc{support}, d \rangle$ message is received from $t + 1$ processes:
\State \hskip6em $\mathit{accepted}_i[V] \gets \mathit{accepted}_i[V] \cup \{d\}$ \label{line:update_accepted}

\smallskip
\State \hskip4em \textbf{if} exists $d \in \mathsf{Digest}$ such that a $\langle \textsc{support}, d \rangle$ message is received from $2t + 1$ processes:
\State \hskip6em $\mathit{vote}_i \gets d$ \label{line:update_vote}
\State \hskip4em \textbf{else:}
\State \hskip6em $\mathit{vote}_i \gets \bot$ \label{line:update_vote_bot}

\smallskip 
\State \hskip2em \emph{Step 4 of view $V$:} \BlueComment{$2$ synchronous rounds}
\State \hskip4em Let $(d_2 \in \mathsf{Digest} \cup \{\bot\}, g_2 \in \{0, 1\}) \gets \mathcal{GC}_2[V].\mathsf{propose}(\mathit{vote}_i)$

\State \hskip4em \textbf{if} $d_2 \neq \bot$: \BlueComment{check if a non-$\bot$ digest is decided from $\mathcal{GC}_2[V]$}
\State \hskip6em $\mathit{locked}_i \gets d_2$ \BlueComment{digest $d_2$ is locked as some correct process might commit it} \label{line:update_locked}
\State \hskip6em \textbf{if} $g_2 = 1$ and $\mathit{committed\_view}_i = \bot$: \BlueComment{check if digest $d_2$ is decided with grade $1$}
\State \hskip8em $\mathit{committed\_view}_i \gets V$
\State \hskip8em \textbf{invoke} $\mathcal{DD}.\mathsf{input}(\mathit{known\_values}[d_2], d_2)$ \BlueComment{\textbf{commit digest $d_2$}} \label{line:commit_hash_value}

\medskip
\State \textbf{-- Task 2 --} \BlueComment{executed in a separate thread}
\State \textbf{upon} $\mathcal{DD}.\mathsf{output}(v' \in \mathsf{Value})$: \label{line:upon_output_optimal_2}
\State \hskip2em \textbf{trigger} $\mathsf{decide}(v')$ \label{line:decide_nonrecursive}
\State \hskip2em \textbf{wait for} view $\mathit{committed\_view}_i + 1$ to be completed (if not yet and if $\mathit{committed\_view}_i + 1 \leq t + 1$) \label{line:wait_for_completion}
\State \hskip2em \textbf{trigger} $\mathsf{stop}$ \BlueComment{process $p_i$ stops \nameopt} \label{line:halt_optimal_2}

\end{algorithmic}
\end{algorithm}

\smallskip
\noindent \textbf{Protocol description.}
\nameopt internally utilizes an instance $\mathcal{DD}$ of the data dissemination primitive.
We design \nameopt in a \emph{view-based} manner: \nameopt operates in (at most) $f + 1$ views, where each view $V$ has its leader $\mathsf{leader}(V) = p_V$.\footnote{\nameopt elects leaders in a round-robin fashion.}
Each view $V$ internally uses two instances $\mathcal{GC}_1[V]$ and $\mathcal{GC}_2[V]$ of the graded consensus primitive (see \Cref{mod:graded-consensus}) that operates on digests.

We say that a correct process $p_i$ \emph{commits} a digest $d$ in view $V$ if and only if $p_i$ invokes $\mathcal{DD}.\mathsf{input}(\cdot, d)$ in view $V$ (line~\ref{line:commit_hash_value}).
Each correct process $p_i$ maintains four important local variables: 
\begin{compactitem}
    \item $\mathit{locked}_i$ (line~\ref{line:variable_locked}): holds a digest (or $\bot$) on which $p_i$ is currently ``locked on''.

    \item $\mathit{vote}_i$ (line~\ref{line:variable_vote}): holds a digest (or $\bot$) currently supported by $p_i$.

   \item $\mathit{known\_values}_i[D]$, for every digest $D$ (line~\ref{line:variable_known}): holds the pre-image of digest $D$ observed by $p_i$.

    \item $\mathit{accepted}_i[V]$, for every view $V$ (line~\ref{line:variable_accepted}): holds the set of digest that are ``accepted'' in view $V$.
\end{compactitem}
Let $p_i$ be any correct process.
Each view $V$ operates in four steps:
\begin{compactenum}
    \item Process $p_i$ proposes $\mathit{locked}_i$ to $\mathcal{GC}_1[V]$ and decides a pair $(d_1, g_1)$ (line~\ref{line:propose_gc1}).
    Intuitively, if $d_1 \neq \bot$ and $g_1 = 1$, $p_i$ sticks with digest $d_1$ throughout the view as it is possible that some other correct process has previously committed digest $d_1$.
    (Hence, not sticking with digest $d_1$ in view $V$ might be dangerous as it could lead to a disagreement on committed digests.)
    
    \item Here, the leader of view $V$ (if correct) aims to enable all correct processes to commit a digest in view $V$.
    Specifically, the leader behaves in the following manner:
    \begin{compactitem}
        \item If it decided a non-$\bot$ digest from $\mathcal{GC}_1[V]$, then the leader disseminates the digest (line~\ref{line:broadcast_d1}).

        \item Otherwise, the leader disseminates its proposal (line~\ref{line:broadcast_vi}).
    \end{compactitem}
    Process $p_i$ behaves according to the following logic:
    \begin{compactitem}
        \item If $p_i$ decided a non-$\bot$ digest $d$ with grade $1$ from $\mathcal{GC}_1[V]$ ($d_1 = d \neq \bot$ and $g_1 = 1$; see the rule at line~\ref{line:check_support_1}), then $p_i$ supports digest $d$ by broadcasting a \textsc{support} message for $d$ (line~\ref{line:broadcast_happy_1}).

        \item If $p_i$ decided $\bot$ from $\mathcal{GC}_1[V]$, then $p_i$ supports a digest $d$ by broadcasting a \textsc{support} message for $d$ (line~\ref{line:broadcast_happy_2} or line~\ref{line:broadcast_happy_3}) if (1) it receives digest $d$ from the leader and $p_i$ accepted $d$ in any previous view (line~\ref{line:check_happy_2}), or (2) it receives a valid value $v$ from the leader such that $\mathsf{digest}(v) = d$ (line~\ref{line:check_validity_received_from_leader}).
        If the latter case applies, process $p_i$ ``observes'' the pre-image $v$ of digest $d$ (line~\ref{line:update_known_values}).
    \end{compactitem}

    \item Process $p_i$ accepts a digest $d$ in view $V$ if it receives a \textsc{support} message for $d$ from $t + 1$ processes (line~\ref{line:update_accepted}).
    Moreover, process $p_i$ updates its $\mathit{vote}_i$ variable to a digest $d$ if it receives a \textsc{support} message for $d$ from $2t + 1$ processes (line~\ref{line:update_vote}).
    Otherwise, process $p_i$ sets its $\mathit{vote}_i$ variable to $\bot$ (line~\ref{line:update_vote_bot}).
    Observe that if any correct process $p_j$ updates its $\mathit{vote}_j$ variable to a digest $d$, then every correct process $p_k$ accepts $d$ in view $V$.
    Indeed, as $p_j$ receives a \textsc{support} message for digest $d$ from at least $2t + 1$ processes out of which at least $t + 1$ are correct, it is guaranteed that $p_k$ receives a \textsc{support} message for $d$ from at least $t + 1$ processes.

    \item Process $p_i$ proposes $\mathit{vote}_i$ to $\mathcal{GC}_2[V]$ and decides a pair $(d_2, g_2)$.
    If $d_2 \neq \bot$, process $p_i$ updates its $\mathit{locked}_i$ variable to $d_2$ (line~\ref{line:update_locked}).
    Additionally, if $g_2 = 1$, then $p_i$ commits $d_2$ (line~\ref{line:commit_hash_value}).
    Importantly, if any correct process $p_j$ commits a digest $d \neq \bot$ in view $V$, \emph{every} correct process $p_k$ updates its $\mathit{locked}_k$ variable to $d$.
    Indeed, as $p_j$ commits $d$, it decides $(d \neq \bot, 1)$ from $\mathcal{GC}_2[V]$.
    The consistency property of $\mathcal{GC}_2[V]$ ensures that each correct process $p_k$ decides $d$ from $\mathcal{GC}_2[V]$.
\end{compactenum}
We emphasize that if process $p_i$ commits a digest in some view $V$, process $p_i$ does not execute any view greater than $V + 1$ (line~\ref{line:no_executing_the_next_next_view}).
Moreover, if $p_i$ commits in view $V < t + 1$, then process $p_i$ necessarily completes view $V + 1$ before stopping (line~\ref{line:wait_for_completion}).
Importantly, process $p_i$ completes view $V + 1$ after \emph{exactly} 6 rounds have elapsed.
Let us elaborate.
As some correct process $p_j \neq p_i$ might never enter view $V + 1$ (since it has committed in a view smaller than view $V$), it is possible that \emph{not all} correct processes participate in view $V + 1$.
This implies that utilized graded consensus instances might never complete, which further means that process $p_i$ can forever be stuck executing a graded consensus instance in view $V + 1$.
To avoid this scenario, process $p_i$ completes view $V + 1$ after 6 rounds irrespectively of which step of view $V + 1$ $p_i$ is in after 6 rounds.
Finally, once $p_i$ outputs a value $v'$ from $\mathcal{DD}$ (and completes the aforementioned ``next view''), $p_i$ decides $v'$ (line~\ref{line:decide_nonrecursive}) and stops executing \nameopt (line~\ref{line:halt_optimal_2}).

\subsection{Proof Sketch} \label{subsection:nameopt_proof_sketch}

This subsection provides a proof sketch of the following theorem:

\begin{theorem} \label{theorem:nameopt}
\nameopt (\Cref{algorithm:optimal_trivial}) is a hash-based early-stopping validated agreement algorithm with $O(nL + n^3 \kappa)$ bit complexity.   
\end{theorem}
Our proof sketch focuses on the crucial intermediate guarantees ensured by \nameopt.
Recall that a formal proof of \Cref{theorem:nameopt} is relegated to \Cref{section:nameopt_proof}. 

\smallskip
\noindent \textbf{Preventing disagreement on committed digests.}
First, we show that correct processes do not disagree on committed digests.
Let $\mathcal{V}$ denote the first view in which a correct process commits; let $d^{\star}$ be the committed digest.
No correct process commits any non-$d^{\star}$ digest in view $\mathcal{V}$ due to the consistency property of $\mathcal{GC}_2[\mathcal{V}]$: it is impossible for correct processes to decide different digests from $\mathcal{GC}_2[\mathcal{V}]$ with grade $1$.

If $\mathcal{V} < t + 1$, \nameopt prevents any non-$d^{\star}$ digest to be committed in any view greater than $\mathcal{V}$.
Specifically, \nameopt guarantees that all correct processes commit $d^{\star}$ (and no other digest) by the end of view $\mathcal{V} + 1$.
The consistency property of $\mathcal{GC}_2[\mathcal{V}]$ ensures that every correct process $p_i$ updates its $\mathit{locked}_i$ variable to $d^{\star}$ at the end of view $\mathcal{V}$.
Therefore, all correct processes propose $d^{\star}$ to $\mathcal{GC}_1[\mathcal{V} + 1]$, which implies that all correct processes decide $(d^{\star}, 1)$ from $\mathcal{GC}_1[\mathcal{V} + 1]$ (due to the strong unanimity property of $\mathcal{GC}_1[\mathcal{V} + 1]$).
Hence, all correct processes broadcast a \textsc{support} message for digest $d^{\star}$ (line~\ref{line:broadcast_happy_1}), which further implies that all correct processes propose $d^{\star}$ to $\mathcal{GC}_2[\mathcal{V} + 1]$.
Finally, the strong unanimity property of $\mathcal{GC}_2[\mathcal{V} + 1]$ ensures that all correct processes decide $(d^{\star}, 1)$ from $\mathcal{GC}_2[\mathcal{V} + 1]$ and thus commit $d^{\star}$ by the end of view $\mathcal{V} + 1$.

\smallskip
\noindent \textbf{Ensuring eventual agreement on the committed digest.}
Second, we prove that an agreement on the committed digest eventually occurs.
Concretely, we now show that \emph{all} correct processes commit a digest by the end of the first view whose leader is correct.
Let that view be denoted by $\mathcal{V}_l \in [1, f + 1]$ and let $p_{\mathcal{V}_l}$ be the leader of $\mathcal{V}_l$.
If any correct process commits a digest in any view smaller than $\mathcal{V}_l$, then all correct processes commit the same digest by the end of view $\mathcal{V}_l$ due to the argument from the previous paragraph.
Hence, suppose no correct process commits any digest in any view preceding view $\mathcal{V}_l$.
We distinguish two scenarios:
\begin{compactitem}
    \item Let $p_{\mathcal{V}_l}$ decide a digest $d \neq \bot$ from $\mathcal{GC}_1[\mathcal{V}_l]$.
    Crucially, the justification property of $\mathcal{GC}_1[\mathcal{V}_l]$ ensures that $d \neq \bot$ is proposed by some correct process $p_j$.
    Hence, the value of the $\mathit{locked}_j$ variable is $d$ at the beginning of view $\mathcal{V}_l$.
    Let $V' < \mathcal{V}_l$ denote the view in which $p_j$ updates $\mathit{locked}_j$ to $d$ upon deciding $d \neq \bot$ from $\mathcal{GC}_2[V']$. 
    Again, the justification property of $\mathcal{GC}_2[V']$ guarantees that a correct process proposed $d$ to $\mathcal{GC}_2[V']$ upon receiving $2t + 1$ \textsc{support} messages for $d$.
    As at least $t + 1$ such messages are received from correct processes, \emph{every} correct process accepts digest $d$ in view $V'$.

    In this case, process $p_{\mathcal{V}_l}$ broadcasts digest $d$ in Step 2.
    We show that all correct processes broadcast a \textsc{support} message for digest $d$.
    Consider any correct process $p_i$.
    We study two possible cases:
    \begin{compactitem}
        \item Let $p_i$ decide a non-$\bot$ digest $d'$ with grade $1$ from $\mathcal{GC}_1[\mathcal{V}_l]$.
        In this case, the consistency property of $\mathcal{GC}_1[\mathcal{V}_l]$ ensures that $d = d'$.
        Thus, process $p_i$ sends a \textsc{support} message for digest $d$ (line~\ref{line:broadcast_happy_1}).

        \item Let $p_i$ decide $\bot$ or with grade $0$ from $\mathcal{GC}_1[\mathcal{V}_l]$.
        In this case, process $p_i$ sends a \textsc{support} message for digest $d$ (line~\ref{line:broadcast_happy_2}) as (1) it receives $d$ from $p_{\mathcal{V}_l}$, and (2) it accepts $d$ in view $V' < \mathcal{V}_l$.
    \end{compactitem}

    \item Let $p_{\mathcal{V}_l}$ decide $\bot$ from $\mathcal{GC}_1[\mathcal{V}_l]$.
    Note that this implies that no correct process decides a non-$\bot$ digest with grade $1$ from $\mathcal{GC}_1[\mathcal{V}_l]$ (due to the consistency property of $\mathcal{GC}_1[\mathcal{V}_l]$).
    Hence, process $p_{\mathcal{V}_l}$ broadcasts its valid value $v$, which then implies that all correct processes send a \textsc{support} message for digest $d = \mathsf{digest}(v)$ (line~\ref{line:broadcast_happy_3}).
\end{compactitem}
Hence, there exists a digest $d$ for which all correct processes express their support in both cases.
Therefore, all correct processes propose $d$ to $\mathcal{GC}_2[\mathcal{V}_l]$.
Finally, the strong unanimity property ensures that all correct processes decide $(d, 1)$ from $\mathcal{GC}_2[\mathcal{V}_l]$ and thus commit digest $d$ in view $\mathcal{V}_l$.

\smallskip
\noindent \textbf{Ensuring that some correct process knows the valid pre-image of the committed digest.}
We show how \nameopt enables processes to ``obtain'' implicit PoRs (see \Cref{section:introduction}).
Let $d^{\star}$ denote the (unique) committed digest.
For $d^{\star}$ to be committed, there exists a correct process that sends a \textsc{support} message for $d^{\star}$ in a view in which $d^{\star}$ is committed (due to the justification property of $\mathcal{GC}_2[V]$, for every view $V$).
Therefore, it suffices to show that the first correct process to ever send a \textsc{support} message for $d^{\star}$ (or any other digest) does so at line~\ref{line:broadcast_happy_3} upon receiving valid value $v^{\star}$ with $\mathsf{digest}(v^{\star}) = d^{\star}$.
Let $p_i$ denote the first process to send a \textsc{support} message for digest $d^{\star}$ and let it do so in some view $V$.
We study if $p_i$ could have sent the message at lines~\ref{line:broadcast_happy_1} and~\ref{line:broadcast_happy_2}:
\begin{compactitem}
    \item Process $p_i$ could not have sent the \textsc{support} message at line~\ref{line:broadcast_happy_1} as this would imply that $p_i$ is not the first correct process to send the message for $d^{\star}$.
    The justification property of $\mathcal{GC}_1[V]$ ensures that some correct process $p_j$ has its $\mathit{locked}_j$ variable set to $d^{\star}$ at the beginning of view $V$.
    For process $p_j$ to update its $\mathit{locked}_j$ variable to $d^{\star}$ in some view $V' < V$, there must exist a correct process that sends a \textsc{support} message for $d^{\star}$ in view $V'$ (due to the justification property of $\mathcal{GC}_2[V']$).
    Therefore, $p_i$ cannot be the first correct process to send a \textsc{support} message for $d^{\star}$.

    \item Process $p_i$ could not have sent the \textsc{support} message at line~\ref{line:broadcast_happy_2} as this would also imply that $p_i$ is not the first correct process to send the message for $d^{\star}$.
    Indeed, for the message to be sent at line~\ref{line:broadcast_happy_2}, process $p_i$ accepts $d^{\star}$ in some view $V' < V$, which implies that at least one correct process sends a \textsc{support} message for $d^{\star}$ in view $V'$.
\end{compactitem}
Hence, $p_i$ must have sent the message at line~\ref{line:broadcast_happy_3}, which implies that $p_i$ knows the pre-image $v^{\star}$ of digest $d^{\star}$ and that $v^{\star}$ is valid (due to the check at line~\ref{line:check_validity_received_from_leader}).

\smallskip
\noindent \textbf{Correctness.}
The previous three intermediate results show that the preconditions of $\mathcal{DD}$ (see \Cref{mod:hashfin}) are satisfied, which implies that $\mathcal{DD}$ behaves according to its specification.
Hence, all correct processes decide the same valid value from \nameopt due to the properties of $\mathcal{DD}$.

\smallskip
\noindent \textbf{Complexity.}
Each view with a non-correct leader exchanges $O(n^2\kappa)$ bits.
Moreover, each view with a correct leader exchanges $O(nL + n^2 \kappa)$ bits.
As $\mathcal{DD}$ exchanges $O(nL + n^2 \kappa \log n)$ bits and it is ensured that only $O(1)$ views with correct leaders are executed, \nameopt exchanges $O(nL + n^2 \kappa) + n \cdot O(n^2 \kappa) + O(nL + n^2 \kappa \log n) = O(nL + n^3\kappa)$ bits.
(As we prove in \Cref{section:nameopt_proof}, a tight analysis shows that \nameopt's bit complexity is actually $O( nL + n^2(f + \log n)\kappa)$.)

As all correct processes start $\mathcal{DD}$ at the end of the first view with a correct leader (at the latest), all correct processes input to $\mathcal{DD}$ in $O(f+1)$ rounds (recall that each view has $6$ rounds).
Since $\mathcal{DD}$ guarantees agreement in $2$ rounds, all correct decide and stop in $O(f+1)$ rounds.

\smallskip
\noindent \textbf{On the lack of \strongVal.}
Note that \nameopt as presented in \Cref{algorithm:optimal_trivial} does not satisfy \strongVal.
Indeed, even if all correct processes propose the same value $v$, it is possible that correct processes agree on a value $v'$ proposed by a faulty leader.
However, as specified in \Cref{section:introduction}, it is trivial to modify \nameopt to obtain an early-stopping algorithm with both strong unanimity and external validity that exchanges $O(nL + n^3\kappa)$ bits.
Indeed, this can be done by running in parallel (1) the current (without \strongVal) implementation of \nameopt, and (2) the error-free  early-stopping COOL~\cite{Chen2021,Lenzen2022} protocol with only \strongVal.
\section{\name: Near-Optimal Early-Stopping Error-Free Solution} \label{section:general_framework}

This section presents \name, an error-free validated Byzantine agreement algorithm that achieves (1) $O\big( (nL + n^2) \log n \big)$ bit complexity, and (2) early stopping.
Recall that \name is also optimally resilient (tolerates up to $t < n / 3$ Byzantine processes).

We start by introducing \name's building blocks (\Cref{subsection:ext_building_blocks}).
To introduce \name's recursive structure, we first show how (a simplified version of) the recursive structure yields a near-optimal validated agreement without early-stopping -- \slowext (\Cref{subsection:ext_overview}).
Then, we overview \name (\Cref{subsection:name_pseudocode}) and give a proof sketch of its correctness and complexity (\Cref{subsection:name_proof_sketch}).
We relegate \name's full pseudocode and a formal proof to \Cref{section:name_detailed_proof}.

\subsection{Building Blocks} \label{subsection:ext_building_blocks}

We now overview the building blocks of \name.
Given \name's recursive structure, the specification of each building block explicitly states its participants (to increase the clarity).
Moreover, given that building blocks might be executed among an overly corrupted set of participants (due to the recursion), each building block explicitly states what properties are ensured given the level of corruption among its participants.

\smallskip
\noindent \textbf{Committee broadcast.}
The formal specification of the committee broadcast primitive is given in \Cref{mod:committee-broadcast}.
Committee broadcast is concerned with two sets of processes: (1) $\mathsf{Entire} \subseteq \Pi$, and (2) $\mathsf{Committee} \subseteq \mathsf{Entire}$.
Moreover, the primitive is associated with a validated Byzantine agreement algorithm $\mathcal{VA}$ to be executed among processes in $\mathsf{Committee}$.
Intuitively, the committee broadcast primitive ensures the following: (1) correct processes in $\mathsf{Committee}$ agree on the same value using the $\mathcal{VA}$ algorithm (given that $\mathsf{Committee}$ is not overly corrupted), and (2) correct processes in $\mathsf{Committee}$ disseminate the previously agreed-upon value to all processes in $\mathsf{Entire}$.
We underline that the totality property of committee broadcast (deliberately written in orange in \Cref{mod:committee-broadcast}) is important only for \name's early-stopping, i.e., it can be ignored for \slowext (in \Cref{subsection:ext_overview}).

\begin{module}[ht]
\caption{Committee broadcast $\langle \mathsf{Entire}, \mathsf{Committee}, \mathcal{VA} \rangle$}
\label{mod:committee-broadcast}
\footnotesize
\begin{algorithmic}[1]

\Statex \textbf{Participants:}
\begin{compactitem}
    \item $\mathsf{Entire} \subseteq \Pi$; let $x = |\mathsf{Entire}|$ and let $x'$ be the greatest integer smaller than $x / 3$.

    \item $\mathsf{Committee} \subseteq \mathsf{Entire}$; let $y = |\mathsf{Committee}|$, let $y'$ be the greatest integer smaller than $y / 3$ and let $f'$ be the actual number of faulty processes in $\mathsf{Committee}$.
\end{compactitem}

\medskip
\Statex \textbf{Utilized validated agreement among $\mathsf{Committee}$:}
\begin{compactitem}
    \item $\mathcal{VA}$; let $\mathcal{L}_{\mathcal{VA}}(y,f')$ denote the worst-case latency complexity of $\mathcal{VA}$ with up to $f'$ faulty processes and let $\mathcal{B}_{\mathcal{VA}}(y)$ denote the maximum number of bits any correct process sends while executing $\mathcal{VA}$ with up to $y'$ faulty processes.
    (We underline that $\mathcal{L}_{\mathcal{VA}}(y, f')$ is based on the non-known \emph{actual} number of failures, whereas $\mathcal{B}_{\mathcal{VA}}(y)$ is based on the known \emph{upper bound} on the number of failures.)
\end{compactitem}

\medskip
\Statex \textbf{Events:}
\begin{compactitem}
    \item \emph{request} $\mathsf{input}(v \in \mathsf{Value}, g \in \{0, 1\})$: a process inputs a pair $(v, g)$.

    \item \emph{indication} $\mathsf{output}(v' \in \mathsf{Value})$: a process outputs a value $v'$.
\end{compactitem}

\medskip
\Statex \textbf{Assumed behavior:}
\begin{compactitem}
    \item Every correct process inputs a pair.

    \item If a correct process inputs a pair $(v, \cdot)$, then $\mathsf{valid}(v) = \mathit{true}$.
    
    \item No correct process stops unless it has previously output a value.

    \item If any correct process inputs a pair $(v, 1)$, for any value $v$, then no correct process inputs a pair $(v' \neq v, \cdot)$.
\end{compactitem}

\medskip
\Statex \textbf{Properties ensured only if up to $x'$ processes in $\mathsf{Entire}$ are faulty:}
\begin{compactitem}
    \item \textcolor{orange}{\emph{Totality:} Let $\tau$ denote the first time at which a correct process outputs a value.
    Then, every correct process outputs a value by time $\tau + 2\delta$.}

    \item \emph{Stability:} If a correct process inputs a pair $(v, 1)$ and outputs a value $v'$, then $v' = v$.

    \item \emph{External validity:} If a correct process outputs a value $v$, then $\mathsf{valid}(v) = \mathit{true}$.

    \item \emph{Optimistic consensus:} If (1) there are up to $y'$ faulty processes in $\mathsf{Committee}$, and (2) all correct processes in $\mathsf{Entire}$ start within $2\delta$ time of each other, the following properties are satisfied:
    \begin{compactitem}
        \item \emph{Liveness:} Let $\tau$ be the first time by which all correct processes in $\mathsf{Committee}$ have input a pair.
        Then, every correct process outputs a value by time $\tau + 7\delta + \mathcal{L}_{\mathcal{VA}}(y,f')$.

        \item \emph{Agreement:} No two correct processes output different values.

        \item \emph{Strong unanimity:} If every correct process proposes a pair $(v, \cdot)$, for any value $v$, then no correct process outputs a value different from $v$.
    \end{compactitem}
\end{compactitem}

\medskip
\Statex \textbf{Properties ensured even if more than $x'$ processes in $\mathsf{Entire}$ are faulty:}
\begin{compactitem}
    \item \emph{Complexity:} Each correct process sends $O(L + x \log x) + \mathcal{B}_{\mathcal{VA}}$ bits. 
\end{compactitem}
\end{algorithmic}
\end{module}

\smallskip
\noindent \textbf{Finisher.}
The formal specification of the finisher primitive is given in \Cref{mod:finisher}.
Finisher is executed among a set $\mathsf{Entire} \subseteq \Pi$ of processes.
Each process inputs a pair $(v \in \mathsf{Value}, g \in \{0, 1\})$, where $v$ is a value and $g$ is a binary grade.
In brief, finisher ensures that all correct processes output the same value if all correct processes input the same value with grade 1 (the liveness property).
Moreover, finisher ensures totality: if any correct process outputs a value, then all correct processes output the same value.
We emphasize that the finisher primitive is introduced \emph{only} for achieving early-stopping in \name, i.e., it plays no role in \slowext.

\begin{module}[ht]
\caption{Finisher $\langle \mathsf{Entire} \rangle$}
\label{mod:finisher}
\footnotesize
\begin{algorithmic}[1]

\Statex \textbf{Participants:}
\begin{compactitem}
    \item $\mathsf{Entire} \subseteq \Pi$; let $x = |\mathsf{Entire}|$ and let $x'$ be the greatest integer smaller than $x / 3$.
\end{compactitem}

\medskip
\Statex \textbf{Events:}
\begin{compactitem}
    \item \emph{request} $\mathsf{input}(v \in \mathsf{Value}, g \in \{0, 1\})$: a process inputs a pair $(v, g)$.

    \item \emph{indication} $\mathsf{output}(v' \in \mathsf{Value})$: a process outputs a value $v'$.
\end{compactitem}

\medskip
\Statex \textbf{Assumed behavior:}
\begin{compactitem}
    \item All correct processes input a pair and they do so within $2\delta$ time of each other.

    \item No correct process stops unless it has previously output a value.

    \item If any correct process inputs a pair $(v, 1)$, for any value $v$, then no correct process inputs a pair $(v' \neq v, \cdot)$.
\end{compactitem}

\medskip
\Statex \textbf{Properties ensured only if up to $x'$ processes in $\mathsf{Entire}$ are faulty:}
\begin{compactitem}
    \item \emph{Preservation:} If a correct process $p_i$ outputs a value $v'$, then $p_i$ has previously input a pair $(v', \cdot)$.

    \item \emph{Agreement:} No two correct processes output different values.

    \item \emph{Justification:} If a correct process outputs a value, then a pair $(\cdot, 1)$ was input by a correct process.

    \item \emph{Liveness:} 
    Let all correct processes input a pair $(v, 1)$, for any value $v$.
    Let $\tau$ be the first time by which all correct processes have input.
    Then, all correct processes output value $v$ by time $\tau + \delta$.

    \item \emph{Totality:} Let $\tau$ be the first time at which a correct process outputs a value.
    Then, all correct processes output a value by time $\tau + 2\delta$.
\end{compactitem}

\medskip
\Statex \textbf{Properties ensured even if more than $x'$ processes in $\mathsf{Entire}$ are faulty:}
\begin{compactitem}
    \item \emph{Complexity:} Each correct process sends $O(x)$ bits.
\end{compactitem}
\end{algorithmic}
\end{module}

\subsection{\slowext: Achieving Near-Optimality Without Early-Stopping}
\label{subsection:ext_overview}
\noindent \textbf{Wisdom of the ancients.}
As mentioned in \Cref{subsection:technical_challenges}, the problem with the sequential reconstructive approach is that, by allowing each Byzantine process to impose its own value, we can end up with $f = t \in O(n)$ (wasted) reconstructions of invalid values (with $O(n^2)$ messages each), for a total of $O(n^3)$ messages.
Making an analogy to a parliamentary system (e.g., of some island in ancient Greece~\cite{lamport2001paxos}), this is the equivalent of allowing every single member of parliament to present their proposal to all others.
This is somewhat wasteful.
In many modern parliamentary systems, since time is limited, proposals are first filtered \emph{internally} within each party before each party presents \emph{one} proposal to the whole assembly.
Hence, no matter how many bad proposals a party might have internally, the whole assembly only discusses one per party.
The cost of dealing with bad actors (and proposals) is shifted to the parties, which are individually smaller than the whole assembly.
This is (essentially) the crucial realization of~\cite{berman1992bit,CoanW92}.
By adopting a recursive framework with two ``parties'' at each level,~\cite{berman1992bit,CoanW92} obtain non-early-stopping solutions with optimal $O(n^2)$ exchanged messages (albeit still $O(n^2L)$ exchanged bits).

\smallskip
\noindent \textbf{\slowext in a nutshell.}
To design \slowext, we adapt the recursive framework of~\cite{berman1992bit,CoanW92} to long values.
More precisely, we follow the recent variant of the framework proposed by~\cite{Momose2021,Lenzen2022} that utilizes (1) the graded consensus~\cite{AW23,Abraham2022} primitive (instead of the ``universal exchange'' primitive of~\cite{berman1992bit}; see \Cref{mod:graded-consensus}), and (2) the committee broadcast primitive (see \Cref{mod:committee-broadcast}).
At each recursive iteration, processes are statically partitioned into two halves (according to their identifiers) that run the algorithm among $n / 2$ processes (inside that half's committee broadcast primitive) in sequential order.
The recursion stops once a validated agreement instance with only a single process is reached; at this point, the process decides its proposal.
A graphical depiction of \slowext is given in the gray part of \Cref{fig:early stopping recursive framework}.

\begin{figure}[ht]
   \centering
   \includegraphics[width=\textwidth]{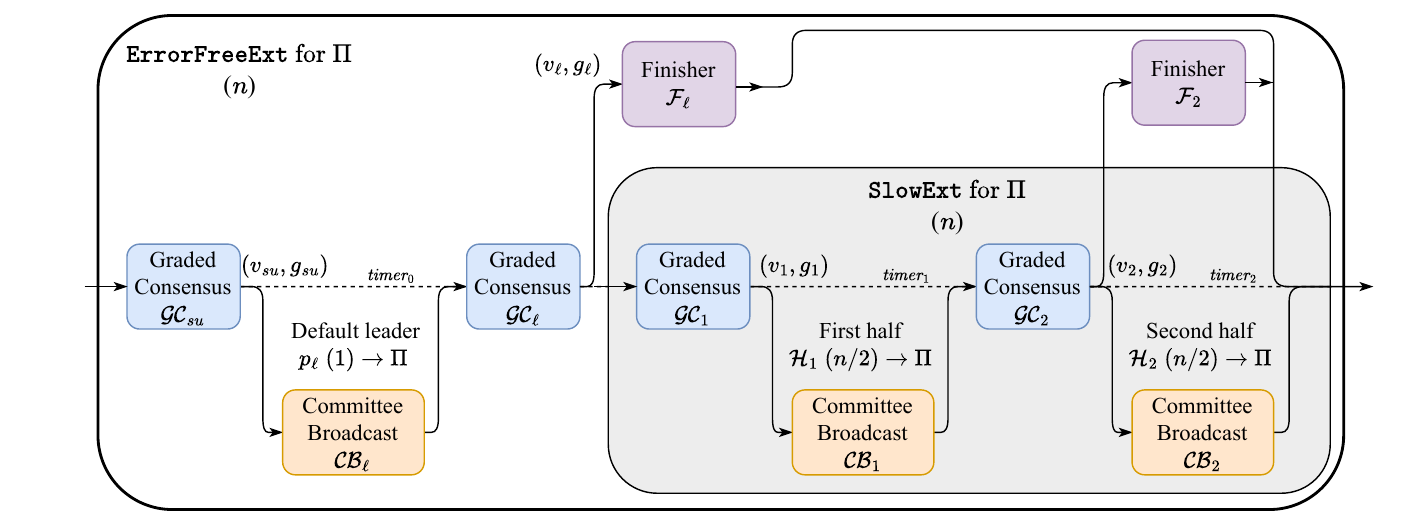}
   \caption{The recursive structure of \name (and \slowext).}
   \label{fig:early stopping recursive framework}
\end{figure}

Crucially, as $t < n / 3$, at least one half contains less than one-third of faulty processes.
Therefore, there exists a ``healthy'' (non-corrupted) half that successfully executes the recursive call (i.e., successfully executes the committee broadcast primitive).
However, agreement achieved among a healthy half must be preserved, i.e., preventing an unhealthy half from ruining the ``healthy decision'' is imperative.
To this end, the recursive framework utilizes the graded consensus primitive that allows the correct processes to stick with their previously made (if any) decision.
For example, suppose that the first half of processes is healthy.
Hence, after executing \slowext among the first half of processes (i.e., in the first committee broadcast primitive), all correct processes obtain the same value (due to the optimistic consensus property of committee broadcast).
In this case, the graded consensus primitive $\mathcal{GC}_2$ ensures that correct processes cannot change their values due to the actions of the second half, thus preserving the previously achieved agreement.
By implementing both the graded consensus and committee broadcast primitives with only $O(nL + n^2 \log n)$ bits (see \Cref{section:graded_consensus_concrete_implementations,section:name_detailed_proof}), \slowext achieves near-optimal asymptotic bit complexity:
\begin{equation*}
\label{equation:ext_complexity_sum}
\sum^{\log n}_{i=0} 2^i\cdot\Big( \frac{n}{2^i}L + \Big(\frac{n}{2^i}\Big)^2\log \Big(\frac{n}{2^i}\Big)\Big) \leq \sum^{\log n}_{i=0} \Big(nL + \frac{n^2}{2^i}\log n\Big) \in O\big( (nL + n^2) \log n \big).
\end{equation*}

\subsection{\name: Overview}
\label{subsection:name_pseudocode}

The pseudocode for \name is provided in \Cref{algorithm:recursive_early_stopping_king_2} (\Cref{subsection:recursive_proof}), whereas its graphical presentation can be found in \Cref{fig:early stopping recursive framework}.
Below, we give key insights for obtaining \name from \slowext.

\smallskip
\noindent \textbf{Why is \slowext not early-stopping?}
\slowext does not achieve early stopping as \slowext allocates a predetermined number of rounds for each recursive call: processes cannot \emph{prematurely} terminate a recursive call even if they have already decided.
In particular, each recursive call consumes the \emph{maximum} number of rounds necessary for its completion.
This maximum number of rounds is proportional to the upper bound $t$ on the number of Byzantine processes rather than the actual number $f \leq t$ of Byzantine processes. 
As a result, \slowext incurs round complexity dependent on $t$ rather than $f$.

\smallskip
\noindent \textbf{From \slowext to \name.}
To achieve early stopping from \slowext, \name mirrors the binary approach of~\cite{Lenzen2022} and carefully adapts it to long $L$-bit values.
The first key ingredient is the introduction of the finisher instance $\mathcal{F}_2$ that we position (1) before the committee broadcast instance $\mathcal{CB}_2$ led by the second half of processes, and (2) after the graded consensus instance $\mathcal{GC}_2$.
In brief, $\mathcal{F}_2$ leverages the presence of the graded consensus instance $\mathcal{GC}_2$ to check if $\mathcal{GC}_2$ ensured agreement among correct processes.
If that is the case, then $\mathcal{F}_2$ allows correct processes to terminate immediately (i.e., in $O(\delta)$ time) after the termination of the committee broadcast instance $\mathcal{CB}_1$ led by the first half of processes.

However, the introduction of $\mathcal{F}_2$ to tackle early-stopping brings its share of technical difficulties. 
Indeed, since the actual number of failures $f$ is unknown, processes cannot remain perfectly synchronized: a correct process $p_i$ might decide (and terminate) at some time $\tau$ thinking this is the maximum time before all correct processes decide given the failures $p_i$ observed, whereas another correct process $p_j$ might still be running after time $\tau$ as it has observed more failures than $p_i$.
To handle the aforementioned desynchronization, \name relies on \emph{weak synchronization} ensuring that correct processes execute different sub-modules with at most $2\delta$ desynchronization time: if the first correct process starts executing a sub-module at time $\tau$, then all correct processes start executing the same sub-module by time $\tau + 2\delta$.
To achieve this weak synchronization, we follow the standard approach of~\cite{SrikanthToueg85,Srikanth1987}.
Furthermore, to handle the $2\delta$ desynchronization in \name's sub-modules, we extend the round duration of graded consensus instances from the original $\delta$ time to $3\delta$ time.
(The specification of the other sub-modules directly tackles the aforementioned desynchronization.)
We emphasize that at some point $\tau$, correct processes might be in different rounds: e.g., a correct process $p_i$ can be in round $4$, whereas another correct process $p_j$ is in round $5$.
However, the round duration of $3\delta$ ensures that all correct processes \emph{overlap} in each round for (at least) $\delta$ time.
As message delays are bounded by $\delta$, the $\delta$-time-overlap is enough to ensure that each correct process hears all $r$-round-messages from all correct processes before leaving round $r$.
(We emphasize that this is a well-known simulation technique; see, e.g.,~\cite{Lenzen2022,cohen2023make}.)

It is important to mention that \name starts with a single standard ``Phase King'' iteration: (1) the committee broadcast instance $\mathcal{CB}_l$ with a predetermined leader $p_\ell$, (2) the graded consensus instance $\mathcal{GC}_\ell$, and (3) the finisher instance $\mathcal{F}_\ell$.
This iteration is added to prevent \name from running for $\Theta(\log n)$ time when there are only $O(1)$ faults.
Indeed, if the predetermined leader $p_\ell$ is correct, the committee broadcast instance $\mathcal{CB}_\ell$ ensures that all correct processes propose the same valid value $v$ to $\mathcal{GC}_\ell$ in $O(1)$ time after starting \name.
Then, the strong unanimity property of $\mathcal{GC}_\ell$ ensures that all correct processes decide $(v, 1)$ from $\mathcal{GC}_\ell$ and input $(v, 1)$ to $\mathcal{F}_\ell$.
This enables $\mathcal{F}_\ell$ to make all correct processes decide $v$ immediately (i.e., in $O(\delta)$ time) after starting.

Finally, the graded consensus instance $\mathcal{GC}_{su}$ (together with $\mathcal{GC}_\ell$) ensures the \strongVal property.
If all correct processes propose the same value $v$ to \name, then (1) all correct processes decide $(v, 1)$ from $\mathcal{GC}_{su}$ and propose $v$ to $\mathcal{GC}_\ell$, (2) all correct processes decide $(v, 1)$ from $\mathcal{GC}_\ell$ and input $(v, 1)$ to $\mathcal{F}_\ell$, and (3) output $v$ from $\mathcal{F}_\ell$ and decide $v$ from \name.

\subsection{Proof Sketch} \label{subsection:name_proof_sketch}

This subsection provides a proof sketch of the following theorem:

\begin{theorem} \label{theorem:name_correct}
\name is an error-free early-stopping validated agreement algorithm with $O\big( (nL + n^2) \log n \big)$ bit complexity.
\end{theorem}
We underline that \name achieves \emph{balanced} bit complexity as its \emph{per-process} complexity is $O\big( (L + n) \log n \big)$.
Recall that the formal proof of \Cref{theorem:name_correct} is relegated to \Cref{section:name_detailed_proof}.
This subsection only discusses the key intermediate results ensured by \name. 

\smallskip
\noindent \textbf{Gluing all sub-modules together.} Processes execute each sub-module within $2\delta$ time of each other, thus enabling the associated implementations to realize the corresponding specifications. 
The consistency property of the graded consensus primitive ensures a similar consistency for the inputs to the following committee broadcast primitive. 
Under this condition, the \strongVal property of the underlying validated agreement protocol ensures agreement if the recursive call is executed with a healthy (non-corrupted) committee. 

\smallskip
\noindent \textbf{Ensuring \strongVal.} Strong unanimity is implied by (1) the strong unanimity properties of $\mathcal{GC}_{su}$ and $\mathcal{GC}_{\ell}$, (2) the stability property of $\mathcal{CB}_\ell$, and (3) liveness and agreement of $\mathcal{F}_{\ell}$. 

\smallskip
\noindent \textbf{Finisher's ``lock''.} If a process decides a value $v$ via a finisher $\mathcal{F} \in \{\mathcal{F}_\ell, \mathcal{F}_2\}$, the justification property of $\mathcal{F}$, combined with the consistency property of the graded consensus $\mathcal{GC} \in \{\mathcal{GC}_\ell, \mathcal{GC}_2\}$ positioned immediately before, ensures that every correct process outputs $(v, \cdot)$ from $\mathcal{GC}$.

\smallskip
\noindent \textbf{From a correct leader or the first healthy committee to a common valid decision.} 
If the predetermined leader $p_\ell$ is correct, all correct processes agree on a common value after $\mathcal{F}_\ell$: this holds due to (1) the optimistic consensus property of $\mathcal{CB}_\ell$, (2) the strong unanimity property of $\mathcal{GC}_\ell$, and (3) the liveness and agreement properties of $\mathcal{F}_\ell$.
Similarly, if $p_\ell$ is faulty, but the first half of processes is healthy, all correct processes agree on a common value after $\mathcal{F}_2$.
Importantly, if some correct process decides via $\mathcal{F}_\ell$, the finisher's lock (see the paragraph above), combined with strong unanimity of $\mathcal{GC}_1$ and $\mathcal{GC}_2$ and the stability property of $\mathcal{CB}_1$, guarantees agreement.

\smallskip
\noindent \textbf{From the second healthy committee to a common valid decision.} 
If a correct process does not decide via $\mathcal{F}_\ell$ or $\mathcal{F}_2$, it means that both the predetermined leader $p_\ell$ and the first half of processes are unhealthy, which implies that the second half is healthy.
If some correct process decides via $\mathcal{F}_2$, the finisher's lock, combined with $\mathcal{CB}_2$'s \strongVal, preserves agreement. 
Let us emphasize that if some correct process decides via $\mathcal{F}_\ell$, the agreement is ensured due to (1) the finisher's lock, (2) the strong unanimity properties of $\mathcal{GC}_1$ and $\mathcal{GC}_2$, and (3) the stability property of $\mathcal{CB}_1$.

\smallskip
\noindent \textbf{Complexity.} The per-process bit complexity $\mathcal{B}(n)$ of \name follows from the equation $\mathcal{B}(n) \leq O(L + n\log n ) + \max\big(\mathcal{B}(\lfloor \frac{n}{2} \rfloor ), \mathcal{B}(\lceil \frac{n}{2} \rceil )\big)$. Similarly, the early stopping property holds due to the following equations: (1) $\mathcal{L}(n,f) \in O(\delta)$ if the predetermined leader $p_\ell$ is correct, and (2) $\mathcal{L}(n,f) \leq O(\delta) + \mathcal{L}(|\mathcal{H}_1|,f_1) + \mathcal{L}(|\mathcal{H}_2|,f_2)$ otherwise, where $f_1$ (resp., $f_2$) denotes the actual number of faulty processes among the first (resp., second) half of processes $\mathcal{H}_1$ (resp., $\mathcal{H}_2$).
\section{Related Work}\label{section:related_work_extended}

\noindent \textbf{Graded consensus.}
Graded consensus and other closely related problems such as (binding) crusader agreement~\cite{ABY22,MostefaouiMR15} have been collectively studied by Attiya and Welch~\cite{AW23}, who refer to this family of problems as ``connected consensus''.
In brief, connected consensus can be viewed as approximate agreement on spider graphs, consisting of a central clique (could be a single vertex) to which different paths (“branches”) are attached, each representing a potential value.

Graded consensus and the related problems have proven extremely helpful for solving Byzantine agreement.
In fully asynchronous environments, an alternating sequence of quadratic instances of binding crusader agreement and (reasonably fair) common coin allows for a binary Byzantine agreement with \strongVal, an expected $O(n^2)$ bit complexity and an expected $O(1)$ latency~\cite{ABY22,MostefaouiMR15}.
In synchrony, graded consensus has been employed in various adaptions of the ``Phase King'' methodology~\cite{Lenzen2022,CoanW92,berman1992bit,Momose2021,PhaseKingBlogPost} to achieve quadratic consensus for constant-sized values.

Following the seminal Dolev-Reischuk lower bound~\cite{dolev1985bounds} stating that $\Omega(nt)$ messages must be exchanged for solving Byzantine agreement, Abraham and Stern~\cite{AbrahamStern22} have proven that deterministic binary crusader broadcast (and thus graded consensus) also requires $\Omega(nt)$ exchanged messages (assuming $L \in O(1)$) if the adversary has simulation capabilities.\footnote{The ability to simulate encompasses classic unauthenticated  protocols but does not consider resource-restricted model \cite{Garay2020} where the adversary’s capability to simulate other processes can be restricted assuming a per-process bounded rate of cryptographic puzzle-solving capability without any authentication mechanism \cite{Andrychowicz2015,Katz2014}.}
When $L \in O(1)$, Attiya and Welch~\cite{AW23} present an optimally resilient ($t < n / 3$), error-free connected consensus algorithm with optimal $O(n^2)$ bit complexity even in asynchrony.
Momose and Ren~\cite{Momose2021} present a synchronous graded consensus algorithm secure against a computationally bounded adversary with (1) optimal resilience ($t < n / 2$), and (2) $O(n^2\kappa)$ bit complexity (again, for $L \in O(1)$), where $\kappa$ denotes the size of cryptographic objects.
However, for longer $L$-bit values, related work is not as rich.
Abraham and Asharov~\cite{Abraham2022} solve gradecast \cite{FeldmanMicali88},  with optimal $O(nL)$ bit complexity in synchrony, albeit with a statistical error (their protocol is not error-free). 

As mentioned in \Cref{section:introduction}, error-free graded consensus with $O( L + n \log n )$ per-process bits can be easily derived from the technique introduced by Chen in the COOL protocol~\cite{Chen2021}. A quasi-identical extraction appeared for synchronous gradecast in \cite{AsharovChandramouli24,ZLC23}, asynchronous reliable broadcast \cite{ADD0VXZ22}, and $n/5$-resilient asynchronous graded consensus \cite{Li2021,free_partial_sync}. 
Recall that COOL is an error-free synchronous Byzantine agreement algorithm that satisfies (only) \strongVal, exchanges $O (nL + n^2 \log n )$ bits, and terminates in $O(n)$ rounds.\footnote{As mentioned in \Cref{section:introduction}, the COOL protocol trivially achieves $O(f+1)$ round complexity by using \cite{Lenzen2022} as the underlying binary consensus protocol.}
The structure of COOL consists of four phases.
The first three phases reduce the number of possible decision values to at most one, i.e., only one value ``survives'' the first three phases.
At the end of the third phase, all processes engage in a one-bit Byzantine agreement algorithm to decide whether a non-$\bot$ or the $\bot$ value should be decided from COOL.
If it is decided that a non-$\bot$ value should be decided, all processes participate in the fourth reconstruction phase whose goal is to allow all processes to obtain that one non-$\bot$ value that has previously survived the first three phases.

\smallskip
\noindent \textbf{Reductions and equivalences between Byzantine agreement problems.}
Interactive consistency is an agreement variant where correct processes agree on the proposals of all processes. Any solvable Byzantine agreement problem (including validated agreement) can be reduced to interactive consistency at no cost \cite{WBAlowerBound}. This problem can further be reduced to $n$ (parallel) instances of Byzantine broadcast, thus its other name ``parallel broadcast'' \cite{AsharovChandramouli24}. In the honest-majority setting $(n>2t)$, Byzantine broadcast and agreement with \strongVal are computationally equivalent.

Without trusted setup, any Byzantine agreement problem is solvable if and only if $n>3t$~\cite{FLM85,PSL80,pfitzmann1996information} (if we ignore the resource-restricted model~\cite{Garay2020}).
Reducing validated agreement to interactive consistency is excessive, given the unavoidable $\Omega(n^{2}L)$ lower bound associated with interactive consistency. This is obviously true for long inputs, and remains applicable to medium-sized inputs (for instance, when $L = O(n)$ and $L = \omega(\log n)$), since in this case, the $\Omega(n^2 L)$ lower bound dominates the $O(n^2 \log n)$ bit complexity of \name. 

Even when considering amortized complexity, it is easy to see situations where interactive consistency may be excessive. For instance, consider processes that aim to implement state-machine replication via multi-shot interactive consistency for clients external to the system who wish to publish transactions on the corresponding ledger \cite{AKGN18}. If a client sends transactions to a sub-linear number of servers, there's a risk these transactions may never be published. Conversely, if clients send their transactions to a linear number of different servers, the same transaction could appear multiple times, thereby reducing the throughput by the same factor.

\smallskip
\noindent \textbf{Setup-free sub-quadratic randomized Byzantine agreement.}
There exist randomized statistically secure (resilient against a computationally unbounded adversary with probability $1-\frac{1}{2^{\lambda}}$, where $\lambda$ represents some security parameter) Byzantine broadcast protocols that achieve sub-quadratic expected bit complexity (assuming $L \in O(1)$).
However, these algorithms either (1) tackle only static adversaries that corrupt processes at the beginning of an execution~\cite{King2006a,King2009,braud2013fast,Gelles23},\footnote{We underline that \name, as it is a deterministic protocol, tolerates an adaptive adversary capable of corrupting processes during (and not only before) any execution (see \Cref{section:preliminaries}). 
This remains true even if such an adaptive adversary is equipped with after-the-fact-removal capabilities~\cite{AbrahamRevisited,Abraham2023revisited}.} or (2) require private (also known as secure) communication channels~\cite{KingSaia10,King2011}.
Facing static corruption only, there exists a solution to the Byzantine agreement problem that does not rely on $n$ instances of Byzantine broadcast \cite{King2006a}.
(1) Processes agree (with high probability) on a $\lambda$-sized committee, which then executes any Byzantine agreement protocol.
(2) The committee then disseminates the decided value to all processes outside the committee.
(3) A non-committee process decides on a value received by the majority of processes in the committee.
This protocol is correct as long as there is a correct majority in the committee, that is, with probability $1-\mathit{exp}(-\Omega(\lambda))$ due to the standard Chernoff bound.  

\noindent
For agreement with either or both strong unanimity and external validity, this approach can be adapted to an (even strongly) adaptive adversary, but at the cost of at least a quadratic overhead. Specifically, the process begins with an all-to-all dispersal phase, where each process broadcasts a candidate value. This is followed by an alternating sequence of graded consensus and leader election (LE) steps, which might fail but not with a probability greater than $1 - \Omega(1)$. The selection of the dispersed value for the next graded consensus instance relies on the (hopefully) chosen correct leader. Note that, in this case, the transformation is error-free, i.e., the probability of successfully solving the problem is 1. Therefore, a solution for leader or committee election with an expected communication cost of $X$ can be transformed into a solution with the same expected latency, but with an expected communication complexity of $O(nL + X)$ in the case of a static adversary, and $O(n^2L + X)$ for a strongly adaptive adversary.

\smallskip
\noindent \textbf{Randomized Byzantine agreement with constant expected latency.}
Using private (also known as secure) channels,\footnote{Establishing a threshold setup to use a threshold pseudo-random function \cite{CKS05}, as in \cite{ANS23}, necessitates a distributed key generation protocol that also relies on private channels \cite{DXK023}.} it is possible to achieve various types of Byzantine agreement with a constant expected latency, even against adaptive adversaries, by employing some random election techniques. This applies to synchronous interactive consistency \cite{AsharovChandramouli24, ANS23}, synchronous validated agreement \cite{ANS23}, asynchronous common subset (ACS), also referred to as agreement on a core set or vector consensus \cite{das2024asynchronous, shoup2024theoretical, cohen2023concurrent, AAPS23}, asynchronous atomic broadcast (without relying on ACS) \cite{sui2023signature, cheng2024jumbo}, and asynchronous validated agreement \cite{feng2024making,Lu2020,ZZZDHWL22,LL022}.

\smallskip
\noindent \textbf{Randomized Byzantine agreement in the full information model.}
In the full information model (i.e., without private channels), Bar-Joseph and Ben-Or established a lower bound of $\tilde{\Omega}(\sqrt{n})$ latency for any protocol facing a strongly adaptive adversary, even when only crash failures are considered \cite{Bar-JosephB98}. 
However, by slightly limiting the adversary's adaptiveness, logarithmic latency becomes achievable under omission failures \cite{RSS18}.
While a recent near-optimal solution exists in the full information model against a strongly adaptive adversary with omission failures \cite{hajiaghayi2024nearly}, no current protocol achieves better than $O(f)$ latency in the same model with arbitrary (malicious) failures. Nevertheless, under a static adversary, logarithmic latency is attainable, and constant latency is possible with a slight relaxation of resiliency, as demonstrated in \cite{GPV06}.
A comparison of the state-of-the-art algorithms can be found in \Cref{tab:synchronous_protocols}.

\begin{table}[h!]
\scriptsize
\centering
    \begin{tabular}{|c|c|c|c|c|c|}
        \hline
        \textbf{Protocol} & \textbf{Relaxation} & \textbf{Resiliency} & \textbf{Latency} & \textbf{Comm} & \textbf{Problem} \\
        \hline
        \textbf{\name} & None & $n/3$ & $O(f)$ & $O((nL + n^2)\log(n))$ & S + E \\
        \hline
        \textbf{\nameopt} & Hash & $n/3$ & $O(f)$ & $O(nL + \kappa n^3)$ & S + E \\
        
        \hline \hline

        GK \cite{Gelles23} & Static Adv. & $n/(3 + \epsilon)$ & $\mathit{polylog}(n)$ & $\tilde{O}(n^{3/2})$ & LE \\
        \hline

        GPV \cite{GPV06} & Static Adv. & $n/(3 + \epsilon)$ & $O(\log(n)/\epsilon^2)$ & $O(n^3)$ &  LE \\
        \hline
        GPV \cite{GPV06} & Static Adv. & $n/\log^{1.58}(n)$ & $O(1)$ & $O(n^3)$ &  LE \\
        
        \hline \hline
        BB \cite{Bar-JosephB98} & Crash & $\Theta(n)$ & $\tilde{\Omega}(\sqrt{n})$ & - & Lower Bound \\
        \hline

        ACDNPS \cite{Abraham2023revisited} & Omission & $\Theta(n)$ & - & $\Omega(n^2)$ & Lower Bound (S) \\
        \hline

        HKO \cite{hajiaghayi2024nearly} & Omission & $n/30$ & $\tilde{\Theta}(\sqrt{n})$ & $\tilde{\Theta}(n^2)$ & Binary \\

        \hline
        RSS \cite{RSS18} & Delayed Adv, Omission & $n/10$ & $O(\log(n))$ & $O((nL +n^2)\log(n))$ & S + E \\

\hline \hline
        CMS \cite{CMS85,CMS89} & Static Adv, Omission & $n/3$ & $O(1)$ & $O(n^2)$ & Binary \\
        
        \hline \hline
        KS \cite{KingSaia10,King2011} & Priv. Chan. & $n/(3+\epsilon)$ & $\mathit{polylog(n)}$ & $\tilde{O}(n^{3/2})$ & Binary \\
        \hline
        AC \cite{AsharovChandramouli24} & Priv. Chan. & $n/3$ & $O(1)$ & $O(Ln^2+n^3 \log^2(n))$ & IC $\rightarrow$ (S + E) \\
        \hline
    \end{tabular}
    \caption{Overview of state-of-the-art signature-free synchronous protocols in relaxed models. 
    $L$ represents the size of the decided values, and $\kappa$ denotes the size of the hash values.
    S stands for ``\strongVal'', E stands for ``external validity'', LE stands for ``leader election'' and IC stands for ``interactive consistency''. The term binary refers to binary agreement with ``\strongVal''. 
}

    \label{tab:synchronous_protocols}
\end{table}

\smallskip
\noindent \textbf{Error correction.}
Many distributed tasks~\cite{das2021asynchronous,Yurek2022,Shoup2023,Choudhury2023,Das2022,Gagol2019,Keidar2021,Camenisch2022,Duan2022,Abraham2022,Lu2020,Nayak2020,AsharovChandramouli24,ANS23,das2024asynchronous,shoup2024theoretical,cohen2023concurrent,AAPS23,sui2023signature,cheng2024jumbo,feng2024making} employ coding techniques such as erasure codes~\cite{blahut1983theory, Hendricks2007b, alhaddad2021succinct} or error-correction codes~\cite{reed1960, BenOr1993}.
For the comprehensive list of distributed primitive utilizing coding techniques, we refer the reader to \cite{civit_et_al:LIPIcs.DISC.2023.13}.

\medskip
\noindent \textbf{Disperse-agree-retrieval (DARE) paradigm.}
As introduced in \Cref{section:introduction}, existing deterministic authenticated validated agreement protocols with $O(nL + n^2 \kappa)$ communication, employ the DARE paradigm \cite{ada_dare_to_appear_podc24,civit_et_al:LIPIcs.DISC.2023.13}. In this paradigm, processes initially agree on a ``self-certifying" proof of retrievability (PoR) (also known as a ``proof of dispersal''). This proof ensures that processes can retrieve the (valid) pre-image of a digest. Signature-free protocols that use random election techniques through private channels \cite{feng2024making,sui2023signature} can circumvent the need for an explicit PoR. Indeed, the use of randomization allows a constant expected number of reconstructions (with costly $\Omega(nL)$ communication) of non-valid values dispersed by Byzantine processes.

It might be feasible to achieve setup-free validated agreement with $O(nL + \mathit{poly}(\kappa,n))$ communication and without private channels, by employing an extension protocol built upon an explicit PoR. Such a strategy would involve an interactive Succinct ARGument (SARG) system. This argument system would enable processes to prove that the digest of their proposal corresponds to a valid pre-image without having to transmit the proposal itself (the proposal plays the role of the witness in the argument system). As discussed in \cite[Section 6.5.2 and 7.3]{Thaler22} and initially observed in \cite{ZGKPP17}, it is possible to implement such an argument system by using some polynomial commitment scheme \cite[Section 7.3]{Thaler22} to turn the GKR interactive proof for arithmetic circuit evaluation  \cite{GKR08,GKR15,GRK17} into a succinct argument for arithmetic circuit satisfiability.
Subsequently, processes could (1) agree on which digests are both valid and consistently dispersed, (2) execute a validated agreement protocol (designed for values with length independent of $L$) to reach consensus on a retrievable digest, and (3) retrieve the corresponding valid pre-image. Given the deterministic nature and the conceptual simplicity of \nameopt, this approach was not pursued in this paper.
\section{Concluding Remarks} \label{section:conclusion}

This paper introduces \nameopt and \name, two synchronous signature-free algorithms for validated Byzantine agreement.
Both algorithms are (1) optimally resilient (tolerating up to $t < n / 3$ Byzantine failures), and (2) early stopping (terminating in $O(f + 1)$ rounds, where $f \leq t$ denotes the actual number of failures).
On one side, \nameopt utilizes only collision-resistant hashes, achieving a bit complexity of $O(nL + n^3\kappa)$, which is optimal when $L \geq n^2 \kappa$ (with $\kappa$ being the size of a hash value). 
Conversely, \name is error-free, avoids cryptography entirely, and achieves a bit complexity of $O\big( (nL + n^2) \log n \big)$, which is nearly optimal for any $L$.
In the future, we plan to focus on the following open questions:
\begin{compactitem}
    \item Is it possible to design an error-free validated agreement algorithm with a bit complexity of $O(nL)$?
    Our \name algorithm achieves only $O(n L \log n)$ bit complexity.

    \item Can \nameopt be optimized to achieve $O(nL)$ bit complexity for a wider range of proposal sizes $L$?
    Currently, \nameopt allows for optimal $O(nL)$ bit complexity only when $L \geq n^2 \kappa$.
\end{compactitem}

\bibstyle{plainurl}
\bibliography{references}

\begin{thebibliography}{100}

\bibitem{abd2005fault}
Michael Abd-El-Malek, Gregory~R Ganger, Garth~R Goodson, Michael~K Reiter, and Jay~J Wylie.
\newblock {Fault-Scalable Byzantine Fault-Tolerant Services}.
\newblock {\em ACM SIGOPS Operating Systems Review}, 39(5):59--74, 2005.

\bibitem{Abraham2022}
Ittai Abraham and Gilad Asharov.
\newblock Gradecast in synchrony and reliable broadcast in asynchrony with optimal resilience, efficiency, and unconditional security.
\newblock In Alessia Milani and Philipp Woelfel, editors, {\em {PODC} '22: {ACM} Symposium on Principles of Distributed Computing, Salerno, Italy, July 25 - 29, 2022}, pages 392--398. {ACM}, 2022.
\newblock \href {https://doi.org/10.1145/3519270.3538451} {\path{doi:10.1145/3519270.3538451}}.

\bibitem{AAPS23}
Ittai Abraham, Gilad Asharov, Arpita Patra, and Gilad Stern.
\newblock Perfectly secure asynchronous agreement on a core set in constant expected time.
\newblock {\em {IACR} Cryptol. ePrint Arch.}, page 1130, 2023.
\newblock URL: \url{https://eprint.iacr.org/2023/1130}.

\bibitem{ABY22}
Ittai Abraham, Naama Ben{-}David, and Sravya Yandamuri.
\newblock {Efficient and Adaptively Secure Asynchronous Binary Agreement via Binding Crusader Agreement}.
\newblock In Alessia Milani and Philipp Woelfel, editors, {\em {PODC} '22: {ACM} Symposium on Principles of Distributed Computing, Salerno, Italy, July 25 - 29, 2022}, pages 381--391. {ACM}, 2022.
\newblock \href {https://doi.org/10.1145/3519270.3538426} {\path{doi:10.1145/3519270.3538426}}.

\bibitem{AbrahamRevisited}
Ittai Abraham, T.{-}H.~Hubert Chan, Danny Dolev, Kartik Nayak, Rafael Pass, Ling Ren, and Elaine Shi.
\newblock {Communication Complexity of Byzantine Agreement, Revisited}.
\newblock In Peter Robinson and Faith Ellen, editors, {\em Proceedings of the 2019 {ACM} Symposium on Principles of Distributed Computing, {PODC} 2019, Toronto, ON, Canada, July 29 - August 2, 2019}, pages 317--326. {ACM}, 2019.

\bibitem{Abraham2023revisited}
Ittai Abraham, T.{-}H.~Hubert Chan, Danny Dolev, Kartik Nayak, Rafael Pass, Ling Ren, and Elaine Shi.
\newblock Communication complexity of byzantine agreement, revisited.
\newblock {\em Distributed Comput.}, 36(1):3--28, 2023.
\newblock URL: \url{https://doi.org/10.1007/s00446-022-00428-8}, \href {https://doi.org/10.1007/S00446-022-00428-8} {\path{doi:10.1007/S00446-022-00428-8}}.

\bibitem{PhaseKingBlogPost}
Ittai Abraham and Andrew Lewis-Pye.
\newblock {Phase-king through the lens of gradecast: A simple unauthenticated synchronous byzantine agreement protocol. Decentralized Thoughts, Blog Post}.
\newblock \url{https://decentralizedthoughts.github.io/2022-06-09-phase-king-via-gradecast/}.

\bibitem{abraham2016solidus}
Ittai Abraham, Dahlia Malkhi, Kartik Nayak, Ling Ren, and Alexander Spiegelman.
\newblock {Solidus: An Incentive-compatible Cryptocurrency Based on Permissionless Byzantine Consensus}.
\newblock {\em CoRR, abs/1612.02916}, 2016.

\bibitem{solida}
Ittai Abraham, Dahlia Malkhi, Kartik Nayak, Ling Ren, and Alexander Spiegelman.
\newblock {Solida: {A} Blockchain Protocol Based on Reconfigurable Byzantine Consensus}.
\newblock In James Aspnes, Alysson Bessani, Pascal Felber, and Jo{\~{a}}o Leit{\~{a}}o, editors, {\em 21st International Conference on Principles of Distributed Systems, {OPODIS} 2017, Lisbon, Portugal, December 18-20, 2017}, volume~95 of {\em LIPIcs}, pages 25:1--25:19. Schloss Dagstuhl - Leibniz-Zentrum f{\"{u}}r Informatik, 2017.

\bibitem{abraham2019asymptotically}
Ittai Abraham, Dahlia Malkhi, and Alexander Spiegelman.
\newblock Asymptotically optimal validated asynchronous byzantine agreement.
\newblock In Peter Robinson and Faith Ellen, editors, {\em Proceedings of the 2019 {ACM} Symposium on Principles of Distributed Computing, {PODC} 2019, Toronto, ON, Canada, July 29 - August 2, 2019}, pages 337--346. {ACM}, 2019.
\newblock \href {https://doi.org/10.1145/3293611.3331612} {\path{doi:10.1145/3293611.3331612}}.

\bibitem{ANS23}
Ittai Abraham, Kartik Nayak, and Nibesh Shrestha.
\newblock Communication and round efficient parallel broadcast protocols.
\newblock {\em {IACR} Cryptol. ePrint Arch.}, page 1172, 2023.
\newblock URL: \url{https://eprint.iacr.org/2023/1172}.

\bibitem{AbrahamStern22}
Ittai Abraham and Gilad Stern.
\newblock New dolev-reischuk lower bounds meet blockchain eclipse attacks.
\newblock In Eshcar Hillel, Roberto Palmieri, and Etienne Rivi{\`{e}}re, editors, {\em 26th International Conference on Principles of Distributed Systems, {OPODIS} 2022, December 13-15, 2022, Brussels, Belgium}, volume 253 of {\em LIPIcs}, pages 16:1--16:18. Schloss Dagstuhl - Leibniz-Zentrum f{\"{u}}r Informatik, 2022.
\newblock URL: \url{https://doi.org/10.4230/LIPIcs.OPODIS.2022.16}, \href {https://doi.org/10.4230/LIPICS.OPODIS.2022.16} {\path{doi:10.4230/LIPICS.OPODIS.2022.16}}.

\bibitem{adya2002farsite}
Atul Adya, William~J. Bolosky, Miguel Castro, Gerald Cermak, Ronnie Chaiken, John~R. Douceur, Jon Howell, Jacob~R. Lorch, Marvin Theimer, and Roger Wattenhofer.
\newblock {FARSITE:} federated, available, and reliable storage for an incompletely trusted environment.
\newblock In David~E. Culler and Peter Druschel, editors, {\em 5th Symposium on Operating System Design and Implementation {(OSDI} 2002), Boston, Massachusetts, USA, December 9-11, 2002}. {USENIX} Association, 2002.
\newblock URL: \url{http://www.usenix.org/events/osdi02/tech/adya.html}.

\bibitem{ADD0VXZ22}
Nicolas Alhaddad, Sourav Das, Sisi Duan, Ling Ren, Mayank Varia, Zhuolun Xiang, and Haibin Zhang.
\newblock Balanced byzantine reliable broadcast with near-optimal communication and improved computation.
\newblock In Alessia Milani and Philipp Woelfel, editors, {\em {PODC} '22: {ACM} Symposium on Principles of Distributed Computing, Salerno, Italy, July 25 - 29, 2022}, pages 399--417. {ACM}, 2022.
\newblock \href {https://doi.org/10.1145/3519270.3538475} {\path{doi:10.1145/3519270.3538475}}.

\bibitem{alhaddad2021succinct}
Nicolas Alhaddad, Sisi Duan, Mayank Varia, and Haibin Zhang.
\newblock {Succinct Erasure Coding Proof Systems}.
\newblock {\em Cryptology ePrint Archive}, 2021.
\newblock URL: \url{https://eprint.iacr.org/2021/1500.pdf}.

\bibitem{amir2006scaling}
Yair Amir, Claudiu Danilov, Danny Dolev, Jonathan Kirsch, John Lane, Cristina Nita{-}Rotaru, Josh Olsen, and David Zage.
\newblock Steward: Scaling byzantine fault-tolerant replication to wide area networks.
\newblock {\em {IEEE} Trans. Dependable Secur. Comput.}, 7(1):80--93, 2010.
\newblock \href {https://doi.org/10.1109/TDSC.2008.53} {\path{doi:10.1109/TDSC.2008.53}}.

\bibitem{Andrychowicz2015}
Marcin Andrychowicz and Stefan Dziembowski.
\newblock Pow-based distributed cryptography with no trusted setup.
\newblock In Rosario Gennaro and Matthew Robshaw, editors, {\em Advances in Cryptology - {CRYPTO} 2015 - 35th Annual Cryptology Conference, Santa Barbara, CA, USA, August 16-20, 2015, Proceedings, Part {II}}, volume 9216 of {\em Lecture Notes in Computer Science}, pages 379--399. Springer, 2015.
\newblock \href {https://doi.org/10.1007/978-3-662-48000-7\_19} {\path{doi:10.1007/978-3-662-48000-7\_19}}.

\bibitem{AKGN18}
Antonio~Fern{\'{a}}ndez Anta, Kishori~M. Konwar, Chryssis Georgiou, and Nicolas~C. Nicolaou.
\newblock Formalizing and implementing distributed ledger objects.
\newblock {\em {SIGACT} News}, 49(2):58--76, 2018.
\newblock \href {https://doi.org/10.1145/3232679.3232691} {\path{doi:10.1145/3232679.3232691}}.

\bibitem{AsharovChandramouli24}
Gilad Asharov and Anirudh Chandramouli.
\newblock Perfect (parallel) broadcast in constant expected rounds via statistical {VSS}.
\newblock In Marc Joye and Gregor Leander, editors, {\em Advances in Cryptology - {EUROCRYPT} 2024 - 43rd Annual International Conference on the Theory and Applications of Cryptographic Techniques, Zurich, Switzerland, May 26-30, 2024, Proceedings, Part {V}}, volume 14655 of {\em Lecture Notes in Computer Science}, pages 310--339. Springer, 2024.
\newblock \href {https://doi.org/10.1007/978-3-031-58740-5\_11} {\path{doi:10.1007/978-3-031-58740-5\_11}}.

\bibitem{AW23}
Hagit Attiya and Jennifer~L. Welch.
\newblock Multi-valued connected consensus: {A} new perspective on crusader agreement and adopt-commit.
\newblock In Alysson Bessani, Xavier D{\'{e}}fago, Junya Nakamura, Koichi Wada, and Yukiko Yamauchi, editors, {\em 27th International Conference on Principles of Distributed Systems, {OPODIS} 2023, December 6-8, 2023, Tokyo, Japan}, volume 286 of {\em LIPIcs}, pages 6:1--6:23. Schloss Dagstuhl - Leibniz-Zentrum f{\"{u}}r Informatik, 2023.
\newblock URL: \url{https://doi.org/10.4230/LIPIcs.OPODIS.2023.6}, \href {https://doi.org/10.4230/LIPICS.OPODIS.2023.6} {\path{doi:10.4230/LIPICS.OPODIS.2023.6}}.

\bibitem{Bar-JosephB98}
Ziv Bar{-}Joseph and Michael Ben{-}Or.
\newblock A tight lower bound for randomized synchronous consensus.
\newblock In Brian~A. Coan and Yehuda Afek, editors, {\em Proceedings of the Seventeenth Annual {ACM} Symposium on Principles of Distributed Computing, {PODC} '98, Puerto Vallarta, Mexico, June 28 - July 2, 1998}, pages 193--199. {ACM}, 1998.
\newblock \href {https://doi.org/10.1145/277697.277733} {\path{doi:10.1145/277697.277733}}.

\bibitem{Beerliova-Trubiniova2007}
Zuzana Beerliova-Trubiniova and Martin Hirt.
\newblock {Simple and efficient perfectly-secure asynchronous MPC}.
\newblock {\em Lecture Notes in Computer Science (including subseries Lecture Notes in Artificial Intelligence and Lecture Notes in Bioinformatics)}, 4833 LNCS:376--392, 2007.
\newblock \href {https://doi.org/10.1007/978-3-540-76900-2_23} {\path{doi:10.1007/978-3-540-76900-2_23}}.

\bibitem{BenOr1993}
Michael Ben{-}Or, Ran Canetti, and Oded Goldreich.
\newblock Asynchronous secure computation.
\newblock In S.~Rao Kosaraju, David~S. Johnson, and Alok Aggarwal, editors, {\em Proceedings of the Twenty-Fifth Annual {ACM} Symposium on Theory of Computing, May 16-18, 1993, San Diego, CA, {USA}}, pages 52--61. {ACM}, 1993.
\newblock \href {https://doi.org/10.1145/167088.167109} {\path{doi:10.1145/167088.167109}}.

\bibitem{BGW88}
Michael Ben{-}Or, Shafi Goldwasser, and Avi Wigderson.
\newblock Completeness theorems for non-cryptographic fault-tolerant distributed computation (extended abstract).
\newblock In Janos Simon, editor, {\em Proceedings of the 20th Annual {ACM} Symposium on Theory of Computing, May 2-4, 1988, Chicago, Illinois, {USA}}, pages 1--10. {ACM}, 1988.
\newblock \href {https://doi.org/10.1145/62212.62213} {\path{doi:10.1145/62212.62213}}.

\bibitem{berman1992bit}
Piotr Berman, Juan~A Garay, and Kenneth~J Perry.
\newblock {Bit Optimal Distributed Consensus}.
\newblock In {\em Computer science: research and applications}, pages 313--321. Springer, 1992.

\bibitem{blahut1983theory}
Richard~E Blahut.
\newblock {\em Theory and practice of error control codes}, volume 126.
\newblock Addison-Wesley Reading, 1983.

\bibitem{BT85}
Gabriel Bracha and Sam Toueg.
\newblock Asynchronous consensus and broadcast protocols.
\newblock {\em J. ACM}, 32(4):824–840, October 1985.
\newblock \href {https://doi.org/10.1145/4221.214134} {\path{doi:10.1145/4221.214134}}.

\bibitem{braud2013fast}
Nicolas Braud{-}Santoni, Rachid Guerraoui, and Florian Huc.
\newblock Fast byzantine agreement.
\newblock In Panagiota Fatourou and Gadi Taubenfeld, editors, {\em {ACM} Symposium on Principles of Distributed Computing, {PODC} '13, Montreal, QC, Canada, July 22-24, 2013}, pages 57--64. {ACM}, 2013.
\newblock \href {https://doi.org/10.1145/2484239.2484243} {\path{doi:10.1145/2484239.2484243}}.

\bibitem{buchman2016tendermint}
Ethan Buchman.
\newblock {\em {Tendermint: Byzantine Fault Tolerance in the Age of Blockchains}}.
\newblock PhD thesis, University of Guelph, 2016.
\newblock URL: \url{https://atrium.lib.uoguelph.ca/server/api/core/bitstreams/0816af2c-5fd4-4d99-86d6-ced4eef2fb52/content}.

\bibitem{Cachin2001}
Christian Cachin, Klaus Kursawe, Frank Petzold, and Victor Shoup.
\newblock {Secure and Efficient Asynchronous Broadcast Protocols}.
\newblock In Joe Kilian, editor, {\em Advances in Cryptology - {CRYPTO} 2001, 21st Annual International Cryptology Conference, Santa Barbara, California, USA, August 19-23, 2001, Proceedings}, volume 2139 of {\em Lecture Notes in Computer Science}, pages 524--541. Springer, 2001.
\newblock \href {https://doi.org/10.1007/3-540-44647-8\_31} {\path{doi:10.1007/3-540-44647-8\_31}}.

\bibitem{CKS05}
Christian Cachin, Klaus Kursawe, and Victor Shoup.
\newblock {Random Oracles in Constantinople: Practical Asynchronous Byzantine Agreement Using Cryptography}.
\newblock {\em J. Cryptol.}, 18(3):219--246, 2005.
\newblock \href {https://doi.org/10.1007/s00145-005-0318-0} {\path{doi:10.1007/s00145-005-0318-0}}.

\bibitem{Camenisch2022}
Jan Camenisch, Manu Drijvers, Timo Hanke, Yvonne{-}Anne Pignolet, Victor Shoup, and Dominic Williams.
\newblock Internet computer consensus.
\newblock In Alessia Milani and Philipp Woelfel, editors, {\em {PODC} '22: {ACM} Symposium on Principles of Distributed Computing, Salerno, Italy, July 25 - 29, 2022}, pages 81--91. {ACM}, 2022.
\newblock \href {https://doi.org/10.1145/3519270.3538430} {\path{doi:10.1145/3519270.3538430}}.

\bibitem{CL02}
Miguel Castro and Barbara Liskov.
\newblock {Practical {B}yzantine Fault Tolerance and Proactive Recovery}.
\newblock {\em ACM Transactions on Computer Systems}, 20(4), 2002.

\bibitem{Chandran2015}
Nishanth Chandran, Wutichai Chongchitmate, Juan~A. Garay, Shafi Goldwasser, Rafail Ostrovsky, and Vassilis Zikas.
\newblock {The hidden graph model: Communication locality and optimal resiliency with adaptive faults}.
\newblock In {\em ITCS 2015 - Proceedings of the 6th Innovations in Theoretical Computer Science}, pages 153--162, 2015.
\newblock \href {https://doi.org/10.1145/2688073.2688102} {\path{doi:10.1145/2688073.2688102}}.

\bibitem{Chen2021extended}
Jinyuan Chen.
\newblock Fundamental limits of byzantine agreement.
\newblock {\em CoRR}, abs/2009.10965, 2020.
\newblock URL: \url{https://arxiv.org/abs/2009.10965}, \href {https://arxiv.org/abs/2009.10965} {\path{arXiv:2009.10965}}.

\bibitem{Chen2021}
Jinyuan Chen.
\newblock Optimal error-free multi-valued byzantine agreement.
\newblock In Seth Gilbert, editor, {\em 35th International Symposium on Distributed Computing, {DISC} 2021, October 4-8, 2021, Freiburg, Germany (Virtual Conference)}, volume 209 of {\em LIPIcs}, pages 17:1--17:19. Schloss Dagstuhl - Leibniz-Zentrum f{\"{u}}r Informatik, 2021.
\newblock \href {https://doi.org/10.4230/LIPIcs.DISC.2021.17} {\path{doi:10.4230/LIPIcs.DISC.2021.17}}.

\bibitem{cheng2024jumbo}
Hao Cheng, Yuan Lu, Zhenliang Lu, Qiang Tang, Yuxuan Zhang, and Zhenfeng Zhang.
\newblock Jumbo: Fully asynchronous bft consensus made truly scalable.
\newblock {\em arXiv preprint arXiv:2403.11238}, 2024.

\bibitem{CMS85}
Benny Chor, Michael Merritt, and David~B. Shmoys.
\newblock Simple constant-time consensus protocols in realistic failure models (extended abstract).
\newblock In Michael~A. Malcolm and H.~Raymond Strong, editors, {\em Proceedings of the Fourth Annual {ACM} Symposium on Principles of Distributed Computing, Minaki, Ontario, Canada, August 5-7, 1985}, pages 152--162. {ACM}, 1985.
\newblock \href {https://doi.org/10.1145/323596.323610} {\path{doi:10.1145/323596.323610}}.

\bibitem{CMS89}
Benny Chor, Michael Merritt, and David~B. Shmoys.
\newblock Simple constant-time consensus protocols in realistic failure models.
\newblock {\em J. {ACM}}, 36(3):591--614, 1989.
\newblock \href {https://doi.org/10.1145/65950.65956} {\path{doi:10.1145/65950.65956}}.

\bibitem{Choudhury2023}
Ashish Choudhury.
\newblock Almost-surely terminating asynchronous byzantine agreement against general adversaries with optimal resilience.
\newblock In {\em 24th International Conference on Distributed Computing and Networking, {ICDCN} 2023, Kharagpur, India, January 4-7, 2023}, pages 167--176. {ACM}, 2023.
\newblock \href {https://doi.org/10.1145/3571306.3571333} {\path{doi:10.1145/3571306.3571333}}.

\bibitem{ada_dare_to_appear_podc24}
Pierre Civit, Muhammad~Ayaz Dzulfikar, Seth Gilbert, Rachid Guerraoui, Jovan Komatovic, and Manuel Vidigueira.
\newblock {DARE} to agree: Byzantine agreement with optimal resilience and adaptive communication.
\newblock In Ran Gelles, Dennis Olivetti, and Petr Kuznetsov, editors, {\em Proceedings of the 43rd {ACM} Symposium on Principles of Distributed Computing, {PODC} 2024, Nantes, France, June 17-21, 2024}, pages 145--156. {ACM}, 2024.
\newblock \href {https://doi.org/10.1145/3662158.3662792} {\path{doi:10.1145/3662158.3662792}}.

\bibitem{free_partial_sync}
Pierre Civit, Muhammad~Ayaz Dzulfikar, Seth Gilbert, Rachid Guerraoui, Jovan Komatovic, Manuel Vidigueira, and Igor Zablotchi.
\newblock Partial synchrony for free? new bounds for byzantine agreement via a generic transformation across network models.
\newblock {\em CoRR}, abs/2402.10059, 2024.
\newblock URL: \url{https://doi.org/10.48550/arXiv.2402.10059}, \href {https://arxiv.org/abs/2402.10059} {\path{arXiv:2402.10059}}, \href {https://doi.org/10.48550/ARXIV.2402.10059} {\path{doi:10.48550/ARXIV.2402.10059}}.

\bibitem{civit_et_al:LIPIcs.DISC.2023.13}
Pierre Civit, Seth Gilbert, Rachid Guerraoui, Jovan Komatovic, Matteo Monti, and Manuel Vidigueira.
\newblock {Every Bit Counts in Consensus}.
\newblock In Rotem Oshman, editor, {\em 37th International Symposium on Distributed Computing (DISC 2023)}, volume 281 of {\em Leibniz International Proceedings in Informatics (LIPIcs)}, pages 13:1--13:26, Dagstuhl, Germany, 2023. Schloss Dagstuhl -- Leibniz-Zentrum f{\"u}r Informatik.
\newblock URL: \url{https://drops.dagstuhl.de/entities/document/10.4230/LIPIcs.DISC.2023.13}, \href {https://doi.org/10.4230/LIPIcs.DISC.2023.13} {\path{doi:10.4230/LIPIcs.DISC.2023.13}}.

\bibitem{WBAlowerBound}
Pierre Civit, Seth Gilbert, Rachid Guerraoui, Jovan Komatovic, Anton Paramonov, and Manuel Vidigueira.
\newblock All byzantine agreement problems are expensive.
\newblock In Ran Gelles, Dennis Olivetti, and Petr Kuznetsov, editors, {\em Proceedings of the 43rd {ACM} Symposium on Principles of Distributed Computing, {PODC} 2024, Nantes, France, June 17-21, 2024}, pages 157--169. {ACM}, 2024.
\newblock \href {https://doi.org/10.1145/3662158.3662780} {\path{doi:10.1145/3662158.3662780}}.

\bibitem{CoanW92}
Brian~A. Coan and Jennifer~L. Welch.
\newblock {Modular Construction of a Byzantine Agreement Protocol with Optimal Message Bit Complexity}.
\newblock {\em Inf. Comput.}, 97(1):61--85, 1992.
\newblock \href {https://doi.org/10.1016/0890-5401(92)90004-Y} {\path{doi:10.1016/0890-5401(92)90004-Y}}.

\bibitem{cohen2023concurrent}
Ran Cohen, Pouyan Forghani, Juan~A. Garay, Rutvik Patel, and Vassilis Zikas.
\newblock Concurrent asynchronous byzantine agreement in expected-constant rounds, revisited.
\newblock In Guy~N. Rothblum and Hoeteck Wee, editors, {\em Theory of Cryptography - 21st International Conference, {TCC} 2023, Taipei, Taiwan, November 29 - December 2, 2023, Proceedings, Part {IV}}, volume 14372 of {\em Lecture Notes in Computer Science}, pages 422--451. Springer, 2023.
\newblock \href {https://doi.org/10.1007/978-3-031-48624-1\_16} {\path{doi:10.1007/978-3-031-48624-1\_16}}.

\bibitem{cohen2023make}
Shir Cohen, Idit Keidar, and Alexander Spiegelman.
\newblock Make every word count: Adaptive byzantine agreement with fewer words.
\newblock In Eshcar Hillel, Roberto Palmieri, and Etienne Rivi{\`{e}}re, editors, {\em 26th International Conference on Principles of Distributed Systems, {OPODIS} 2022, December 13-15, 2022, Brussels, Belgium}, volume 253 of {\em LIPIcs}, pages 18:1--18:21. Schloss Dagstuhl - Leibniz-Zentrum f{\"{u}}r Informatik, 2022.
\newblock URL: \url{https://doi.org/10.4230/LIPIcs.OPODIS.2022.18}, \href {https://doi.org/10.4230/LIPICS.OPODIS.2022.18} {\path{doi:10.4230/LIPICS.OPODIS.2022.18}}.

\bibitem{correia2019byzantine}
Miguel Correia.
\newblock {From Byzantine Consensus to Blockchain Consensus}.
\newblock In {\em Essentials of Blockchain Technology}, pages 41--80. Chapman and Hall/CRC, 2019.

\bibitem{CGL18}
Tyler Crain, Vincent Gramoli, Mikel Larrea, and Michel Raynal.
\newblock {DBFT:} efficient leaderless byzantine consensus and its application to blockchains.
\newblock In {\em 17th {IEEE} International Symposium on Network Computing and Applications, {NCA} 2018, Cambridge, MA, USA, November 1-3, 2018}, pages 1--8. {IEEE}, 2018.
\newblock \href {https://doi.org/10.1109/NCA.2018.8548057} {\path{doi:10.1109/NCA.2018.8548057}}.

\bibitem{das2024asynchronous}
Sourav Das, Sisi Duan, Shengqi Liu, Atsuki Momose, Ling Ren, and Victor Shoup.
\newblock Asynchronous consensus without trusted setup or public-key cryptography.
\newblock {\em Cryptology ePrint Archive}, 2024.

\bibitem{Das2022}
Sourav Das, Vinith Krishnan, Irene~Miriam Isaac, and Ling Ren.
\newblock Spurt: Scalable distributed randomness beacon with transparent setup.
\newblock In {\em 43rd {IEEE} Symposium on Security and Privacy, {SP} 2022, San Francisco, CA, USA, May 22-26, 2022}, pages 2502--2517. {IEEE}, 2022.
\newblock \href {https://doi.org/10.1109/SP46214.2022.9833580} {\path{doi:10.1109/SP46214.2022.9833580}}.

\bibitem{DXK023}
Sourav Das, Zhuolun Xiang, Lefteris Kokoris{-}Kogias, and Ling Ren.
\newblock Practical asynchronous high-threshold distributed key generation and distributed polynomial sampling.
\newblock In Joseph~A. Calandrino and Carmela Troncoso, editors, {\em 32nd {USENIX} Security Symposium, {USENIX} Security 2023, Anaheim, CA, USA, August 9-11, 2023}, pages 5359--5376. {USENIX} Association, 2023.
\newblock URL: \url{https://www.usenix.org/conference/usenixsecurity23/presentation/das}.

\bibitem{das2021asynchronous}
Sourav Das, Zhuolun Xiang, and Ling Ren.
\newblock Asynchronous data dissemination and its applications.
\newblock In Yongdae Kim, Jong Kim, Giovanni Vigna, and Elaine Shi, editors, {\em {CCS} '21: 2021 {ACM} {SIGSAC} Conference on Computer and Communications Security, Virtual Event, Republic of Korea, November 15 - 19, 2021}, pages 2705--2721. {ACM}, 2021.
\newblock \href {https://doi.org/10.1145/3460120.3484808} {\path{doi:10.1145/3460120.3484808}}.

\bibitem{delporte2021weakest}
Carole Delporte{-}Gallet, Hugues Fauconnier, and Michel Raynal.
\newblock On the weakest information on failures to solve mutual exclusion and consensus in asynchronous crash-prone read/write systems.
\newblock {\em J. Parallel Distributed Comput.}, 153:110--118, 2021.
\newblock URL: \url{https://doi.org/10.1016/j.jpdc.2021.03.015}, \href {https://doi.org/10.1016/J.JPDC.2021.03.015} {\path{doi:10.1016/J.JPDC.2021.03.015}}.

\bibitem{DHS15}
David Derler, Christian Hanser, and Daniel Slamanig.
\newblock Revisiting cryptographic accumulators, additional properties and relations to other primitives.
\newblock In Kaisa Nyberg, editor, {\em Topics in Cryptology - {CT-RSA} 2015, The Cryptographer's Track at the {RSA} Conference 2015, San Francisco, CA, USA, April 20-24, 2015. Proceedings}, volume 9048 of {\em Lecture Notes in Computer Science}, pages 127--144. Springer, 2015.
\newblock \href {https://doi.org/10.1007/978-3-319-16715-2\_7} {\path{doi:10.1007/978-3-319-16715-2\_7}}.

\bibitem{dolev1985bounds}
Danny Dolev and R{\"{u}}diger Reischuk.
\newblock Bounds on information exchange for byzantine agreement.
\newblock {\em J. {ACM}}, 32(1):191--204, 1985.
\newblock \href {https://doi.org/10.1145/2455.214112} {\path{doi:10.1145/2455.214112}}.

\bibitem{dolev1990early}
Danny Dolev, R{\"{u}}diger Reischuk, and H.~Raymond Strong.
\newblock Early stopping in byzantine agreement.
\newblock {\em J. {ACM}}, 37(4):720--741, 1990.
\newblock \href {https://doi.org/10.1145/96559.96565} {\path{doi:10.1145/96559.96565}}.

\bibitem{Duan2022}
Sisi Duan and Haibin Zhang.
\newblock Foundations of dynamic {BFT}.
\newblock In {\em 43rd {IEEE} Symposium on Security and Privacy, {SP} 2022, San Francisco, CA, USA, May 22-26, 2022}, pages 1317--1334. {IEEE}, 2022.
\newblock \href {https://doi.org/10.1109/SP46214.2022.9833787} {\path{doi:10.1109/SP46214.2022.9833787}}.

\bibitem{FeldmanMicali88}
Paul Feldman and Silvio Micali.
\newblock Optimal algorithms for byzantine agreement.
\newblock In Janos Simon, editor, {\em Proceedings of the 20th Annual {ACM} Symposium on Theory of Computing, May 2-4, 1988, Chicago, Illinois, {USA}}, pages 148--161. {ACM}, 1988.
\newblock \href {https://doi.org/10.1145/62212.62225} {\path{doi:10.1145/62212.62225}}.

\bibitem{feng2024making}
Hanwen Feng, Zhenliang Lu, Tiancheng Mai, and Qiang Tang.
\newblock Making hash-based mvba great again.
\newblock {\em Cryptology ePrint Archive}, 2024.

\bibitem{FLM85}
Michael~J. Fischer, Nancy~A. Lynch, and Michael Merritt.
\newblock {Easy Impossibility Proofs for Distributed Consensus Problems}.
\newblock In Michael~A. Malcolm and H.~Raymond Strong, editors, {\em Proceedings of the Fourth Annual {ACM} Symposium on Principles of Distributed Computing, Minaki, Ontario, Canada, August 5-7, 1985}, pages 59--70. {ACM}, 1985.
\newblock \href {https://doi.org/10.1145/323596.323602} {\path{doi:10.1145/323596.323602}}.

\bibitem{Gagol2019}
Adam Gagol, Damian Lesniak, Damian Straszak, and Michal Swietek.
\newblock Aleph: Efficient atomic broadcast in asynchronous networks with byzantine nodes.
\newblock In {\em Proceedings of the 1st {ACM} Conference on Advances in Financial Technologies, {AFT} 2019, Zurich, Switzerland, October 21-23, 2019}, pages 214--228. {ACM}, 2019.
\newblock \href {https://doi.org/10.1145/3318041.3355467} {\path{doi:10.1145/3318041.3355467}}.

\bibitem{Garay2020}
Juan Garay, Aggelos Kiayias, Rafail~M. Ostrovsky, Giorgos Panagiotakos, and Vassilis Zikas.
\newblock {Resource-Restricted Cryptography: Revisiting MPC Bounds in the Proof-of-Work Era}.
\newblock {\em Lecture Notes in Computer Science (including subseries Lecture Notes in Artificial Intelligence and Lecture Notes in Bioinformatics)}, 12106 LNCS:129--158, 2020.
\newblock \href {https://doi.org/10.1007/978-3-030-45724-2_5} {\path{doi:10.1007/978-3-030-45724-2_5}}.

\bibitem{Garg2019}
Sanjam Garg, Aarushi Goel, and Abhishek Jain.
\newblock The broadcast message complexity of secure multiparty computation.
\newblock In Steven~D. Galbraith and Shiho Moriai, editors, {\em Advances in Cryptology - {ASIACRYPT} 2019 - 25th International Conference on the Theory and Application of Cryptology and Information Security, Kobe, Japan, December 8-12, 2019, Proceedings, Part {I}}, volume 11921 of {\em Lecture Notes in Computer Science}, pages 426--455. Springer, 2019.
\newblock \href {https://doi.org/10.1007/978-3-030-34578-5\_16} {\path{doi:10.1007/978-3-030-34578-5\_16}}.

\bibitem{GelashviliKSSX22}
Rati Gelashvili, Lefteris Kokoris{-}Kogias, Alberto Sonnino, Alexander Spiegelman, and Zhuolun Xiang.
\newblock Jolteon and ditto: Network-adaptive efficient consensus with asynchronous fallback.
\newblock In Ittay Eyal and Juan~A. Garay, editors, {\em Financial Cryptography and Data Security - 26th International Conference, {FC} 2022, Grenada, May 2-6, 2022, Revised Selected Papers}, volume 13411 of {\em Lecture Notes in Computer Science}, pages 296--315. Springer, 2022.
\newblock \href {https://doi.org/10.1007/978-3-031-18283-9\_14} {\path{doi:10.1007/978-3-031-18283-9\_14}}.

\bibitem{Gelles23}
Yuval Gelles and Ilan Komargodski.
\newblock Optimal load-balanced scalable distributed agreement.
\newblock In Bojan Mohar, Igor Shinkar, and Ryan O'Donnell, editors, {\em Proceedings of the 56th Annual {ACM} Symposium on Theory of Computing, {STOC} 2024, Vancouver, BC, Canada, June 24-28, 2024}, pages 411--422. {ACM}, 2024.
\newblock \href {https://doi.org/10.1145/3618260.3649736} {\path{doi:10.1145/3618260.3649736}}.

\bibitem{chen2016algorand}
Yossi Gilad, Rotem Hemo, Silvio Micali, Georgios Vlachos, and Nickolai Zeldovich.
\newblock {Algorand: Scaling Byzantine Agreements for Cryptocurrencies}.
\newblock In {\em Proceedings of the 26th Symposium on Operating Systems Principles}, SOSP '17, page 51–68, New York, NY, USA, 2017. Association for Computing Machinery.
\newblock \href {https://doi.org/10.1145/3132747.3132757} {\path{doi:10.1145/3132747.3132757}}.

\bibitem{goldreich1987play}
Oded Goldreich, Silvio Micali, and Avi Wigderson.
\newblock How to play any mental game or {A} completeness theorem for protocols with honest majority.
\newblock In Alfred~V. Aho, editor, {\em Proceedings of the 19th Annual {ACM} Symposium on Theory of Computing, 1987, New York, New York, {USA}}, pages 218--229. {ACM}, 1987.
\newblock \href {https://doi.org/10.1145/28395.28420} {\path{doi:10.1145/28395.28420}}.

\bibitem{GKR08}
Shafi Goldwasser, Yael~Tauman Kalai, and Guy~N. Rothblum.
\newblock Delegating computation: interactive proofs for muggles.
\newblock In Cynthia Dwork, editor, {\em Proceedings of the 40th Annual {ACM} Symposium on Theory of Computing, Victoria, British Columbia, Canada, May 17-20, 2008}, pages 113--122. {ACM}, 2008.
\newblock \href {https://doi.org/10.1145/1374376.1374396} {\path{doi:10.1145/1374376.1374396}}.

\bibitem{GKR15}
Shafi Goldwasser, Yael~Tauman Kalai, and Guy~N. Rothblum.
\newblock Delegating computation: Interactive proofs for muggles.
\newblock {\em J. {ACM}}, 62(4):27:1--27:64, 2015.
\newblock \href {https://doi.org/10.1145/2699436} {\path{doi:10.1145/2699436}}.

\bibitem{GPV06}
Shafi Goldwasser, Elan Pavlov, and Vinod Vaikuntanathan.
\newblock Fault-tolerant distributed computing in full-information networks.
\newblock In {\em 47th Annual {IEEE} Symposium on Foundations of Computer Science {(FOCS} 2006), 21-24 October 2006, Berkeley, California, USA, Proceedings}, pages 15--26. {IEEE} Computer Society, 2006.
\newblock \href {https://doi.org/10.1109/FOCS.2006.30} {\path{doi:10.1109/FOCS.2006.30}}.

\bibitem{GRK17}
Shafi Goldwasser, Guy~N. Rothblum, and Yael~Tauman Kalai.
\newblock Delegating computation: Interactive proofs for muggles.
\newblock {\em Electron. Colloquium Comput. Complex.}, {TR17-108}, 2017.
\newblock URL: \url{https://eccc.weizmann.ac.il/report/2017/108}, \href {https://arxiv.org/abs/TR17-108} {\path{arXiv:TR17-108}}.

\bibitem{GLTZ24}
Vincent Gramoli, Zhenliang Lu, Qiang Tang, and Pouriya Zarbafian.
\newblock Optimal asynchronous byzantine consensus with fair separability.
\newblock {\em {IACR} Cryptol. ePrint Arch.}, page 545, 2024.
\newblock URL: \url{https://eprint.iacr.org/2024/545}.

\bibitem{hajiaghayi2024nearly}
Mohammad Hajiaghayi, Dariusz~R. Kowalski, and Jan Olkowski.
\newblock Nearly-optimal consensus tolerating adaptive omissions: Why a lot of randomness is needed?
\newblock In Ran Gelles, Dennis Olivetti, and Petr Kuznetsov, editors, {\em Proceedings of the 43rd {ACM} Symposium on Principles of Distributed Computing, {PODC} 2024, Nantes, France, June 17-21, 2024}, pages 321--331. {ACM}, 2024.
\newblock \href {https://doi.org/10.1145/3662158.3662826} {\path{doi:10.1145/3662158.3662826}}.

\bibitem{Hendricks2007b}
James Hendricks, Gregory~R. Ganger, and Michael~K. Reiter.
\newblock Verifying distributed erasure-coded data.
\newblock In Indranil Gupta and Roger Wattenhofer, editors, {\em Proceedings of the Twenty-Sixth Annual {ACM} Symposium on Principles of Distributed Computing, {PODC} 2007, Portland, Oregon, USA, August 12-15, 2007}, pages 139--146. {ACM}, 2007.
\newblock \href {https://doi.org/10.1145/1281100.1281122} {\path{doi:10.1145/1281100.1281122}}.

\bibitem{Katz2014}
Jonathan Katz, Andrew Miller, and Elaine Shi.
\newblock Pseudonymous secure computation from time-lock puzzles.
\newblock {\em {IACR} Cryptol. ePrint Arch.}, page 857, 2014.
\newblock URL: \url{http://eprint.iacr.org/2014/857}.

\bibitem{Keidar2021}
Idit Keidar, Eleftherios Kokoris{-}Kogias, Oded Naor, and Alexander Spiegelman.
\newblock All you need is {DAG}.
\newblock In Avery Miller, Keren Censor{-}Hillel, and Janne~H. Korhonen, editors, {\em {PODC} '21: {ACM} Symposium on Principles of Distributed Computing, Virtual Event, Italy, July 26-30, 2021}, pages 165--175. {ACM}, 2021.
\newblock \href {https://doi.org/10.1145/3465084.3467905} {\path{doi:10.1145/3465084.3467905}}.

\bibitem{K0GJ20}
Mahimna Kelkar, Fan Zhang, Steven Goldfeder, and Ari Juels.
\newblock Order-fairness for byzantine consensus.
\newblock In Daniele Micciancio and Thomas Ristenpart, editors, {\em Advances in Cryptology - {CRYPTO} 2020 - 40th Annual International Cryptology Conference, {CRYPTO} 2020, Santa Barbara, CA, USA, August 17-21, 2020, Proceedings, Part {III}}, volume 12172 of {\em Lecture Notes in Computer Science}, pages 451--480. Springer, 2020.
\newblock \href {https://doi.org/10.1007/978-3-030-56877-1\_16} {\path{doi:10.1007/978-3-030-56877-1\_16}}.

\bibitem{Keller2023}
Hannah Keller, Claudio Orlandi, Anat Paskin-Cherniavsky, and Divya Ravi.
\newblock {MPC with Low Bottleneck-Complexity: Information-Theoretic Security and More}.
\newblock In {\em 4th Conference on Information-Theoretic Cryptography (ITC)}, volume 267, pages 1--21, Aarhus, Denmark, 2023.
\newblock \href {https://doi.org/10.4230/LIPIcs.ITC.2023.11} {\path{doi:10.4230/LIPIcs.ITC.2023.11}}.

\bibitem{King2009}
Valerie King and Jared Saia.
\newblock From almost everywhere to everywhere: Byzantine agreement with {\~{o}}(n\({}^{\mbox{3/2}}\)) bits.
\newblock In Idit Keidar, editor, {\em Distributed Computing, 23rd International Symposium, {DISC} 2009, Elche, Spain, September 23-25, 2009. Proceedings}, volume 5805 of {\em Lecture Notes in Computer Science}, pages 464--478. Springer, 2009.
\newblock \href {https://doi.org/10.1007/978-3-642-04355-0\_47} {\path{doi:10.1007/978-3-642-04355-0\_47}}.

\bibitem{KingSaia10}
Valerie King and Jared Saia.
\newblock Breaking the {$O(n^2)$} bit barrier: scalable byzantine agreement with an adaptive adversary.
\newblock In Andr{\'{e}}a~W. Richa and Rachid Guerraoui, editors, {\em Proceedings of the 29th Annual {ACM} Symposium on Principles of Distributed Computing, {PODC} 2010, Zurich, Switzerland, July 25-28, 2010}, pages 420--429. {ACM}, 2010.
\newblock \href {https://doi.org/10.1145/1835698.1835798} {\path{doi:10.1145/1835698.1835798}}.

\bibitem{King2011}
Valerie King and Jared Saia.
\newblock Breaking the \emph{O}(\emph{n}\({}^{\mbox{2}}\)) bit barrier: Scalable byzantine agreement with an adaptive adversary.
\newblock {\em J. {ACM}}, 58(4):18:1--18:24, 2011.
\newblock \href {https://doi.org/10.1145/1989727.1989732} {\path{doi:10.1145/1989727.1989732}}.

\bibitem{King2006a}
Valerie King, Jared Saia, Vishal Sanwalani, and Erik Vee.
\newblock Scalable leader election.
\newblock In {\em Proceedings of the Seventeenth Annual {ACM-SIAM} Symposium on Discrete Algorithms, {SODA} 2006, Miami, Florida, USA, January 22-26, 2006}, pages 990--999. {ACM} Press, 2006.
\newblock URL: \url{http://dl.acm.org/citation.cfm?id=1109557.1109667}.

\bibitem{kotla2007zyzzyva}
Ramakrishna Kotla, Lorenzo Alvisi, Michael Dahlin, Allen Clement, and Edmund~L. Wong.
\newblock Zyzzyva: speculative byzantine fault tolerance.
\newblock In Thomas~C. Bressoud and M.~Frans Kaashoek, editors, {\em Proceedings of the 21st {ACM} Symposium on Operating Systems Principles 2007, {SOSP} 2007, Stevenson, Washington, USA, October 14-17, 2007}, pages 45--58. {ACM}, 2007.
\newblock \href {https://doi.org/10.1145/1294261.1294267} {\path{doi:10.1145/1294261.1294267}}.

\bibitem{kotla2004high}
Ramakrishna Kotla and Michael Dahlin.
\newblock High throughput byzantine fault tolerance.
\newblock In {\em 2004 International Conference on Dependable Systems and Networks {(DSN} 2004), 28 June - 1 July 2004, Florence, Italy, Proceedings}, pages 575--584. {IEEE} Computer Society, 2004.
\newblock \href {https://doi.org/10.1109/DSN.2004.1311928} {\path{doi:10.1109/DSN.2004.1311928}}.

\bibitem{lamport2001paxos}
Leslie Lamport.
\newblock {Paxos Made Simple}.
\newblock {\em ACM SIGACT News (Distributed Computing Column) 32, 4 (Whole Number 121, December 2001)}, pages 51--58, 2001.

\bibitem{LSP82}
Leslie Lamport, Robert Shostak, and Marshall Pease.
\newblock {The {B}yzantine Generals Problem}.
\newblock {\em ACM Transactions on Programming Languages and Systems}, 4(3):382--401, 1982.

\bibitem{lamport2019concurrency}
Leslie Lamport, Robert Shostak, and Marshall Pease.
\newblock Concurrency: The works of leslie lamport.
\newblock {\em Association for Computing Machinery}, pages 203--226, 2019.

\bibitem{Lamport1982}
Leslie Lamport, Robert~E. Shostak, and Marshall~C. Pease.
\newblock The byzantine generals problem.
\newblock {\em {ACM} Trans. Program. Lang. Syst.}, 4(3):382--401, 1982.
\newblock \href {https://doi.org/10.1145/357172.357176} {\path{doi:10.1145/357172.357176}}.

\bibitem{Lenzen2022}
Christoph Lenzen and Sahar Sheikholeslami.
\newblock A recursive early-stopping phase king protocol.
\newblock In Alessia Milani and Philipp Woelfel, editors, {\em {PODC} '22: {ACM} Symposium on Principles of Distributed Computing, Salerno, Italy, July 25 - 29, 2022}, pages 60--69. {ACM}, 2022.
\newblock \href {https://doi.org/10.1145/3519270.3538425} {\path{doi:10.1145/3519270.3538425}}.

\bibitem{Li2021}
Fan Li and Jinyuan Chen.
\newblock Communication-efficient signature-free asynchronous byzantine agreement.
\newblock In {\em {IEEE} International Symposium on Information Theory, {ISIT} 2021, Melbourne, Australia, July 12-20, 2021}, pages 2864--2869. {IEEE}, 2021.
\newblock \href {https://doi.org/10.1109/ISIT45174.2021.9518010} {\path{doi:10.1109/ISIT45174.2021.9518010}}.

\bibitem{LL022}
Yuan Lu, Zhenliang Lu, and Qiang Tang.
\newblock Bolt-dumbo transformer: Asynchronous consensus as fast as the pipelined {BFT}.
\newblock In Heng Yin, Angelos Stavrou, Cas Cremers, and Elaine Shi, editors, {\em Proceedings of the 2022 {ACM} {SIGSAC} Conference on Computer and Communications Security, {CCS} 2022, Los Angeles, CA, USA, November 7-11, 2022}, pages 2159--2173. {ACM}, 2022.
\newblock \href {https://doi.org/10.1145/3548606.3559346} {\path{doi:10.1145/3548606.3559346}}.

\bibitem{LL0W20}
Yuan Lu, Zhenliang Lu, Qiang Tang, and Guiling Wang.
\newblock {Dumbo-MVBA: Optimal Multi-Valued Validated Asynchronous Byzantine Agreement, Revisited}.
\newblock In Yuval Emek and Christian Cachin, editors, {\em {PODC} '20: {ACM} Symposium on Principles of Distributed Computing, Virtual Event, Italy, August 3-7, 2020}, pages 129--138. {ACM}, 2020.
\newblock \href {https://doi.org/10.1145/3382734.3405707} {\path{doi:10.1145/3382734.3405707}}.

\bibitem{Lu2020}
Yuan Lu, Zhenliang Lu, Qiang Tang, and Guiling Wang.
\newblock Dumbo-mvba: Optimal multi-valued validated asynchronous byzantine agreement, revisited.
\newblock In Yuval Emek and Christian Cachin, editors, {\em {PODC} '20: {ACM} Symposium on Principles of Distributed Computing, Virtual Event, Italy, August 3-7, 2020}, pages 129--138. {ACM}, 2020.
\newblock \href {https://doi.org/10.1145/3382734.3405707} {\path{doi:10.1145/3382734.3405707}}.

\bibitem{luu2015scp}
Loi Luu, Viswesh Narayanan, Kunal Baweja, Chaodong Zheng, Seth Gilbert, and Prateek Saxena.
\newblock {SCP: A Computationally-Scalable Byzantine Consensus Protocol For Blockchains}.
\newblock {\em Cryptology ePrint Archive}, 2015.

\bibitem{macwilliams1977theory}
Florence~Jessie MacWilliams and Neil James~Alexander Sloane.
\newblock {\em {The Theory of Error-Correcting Codes}}, volume~16.
\newblock Elsevier, 1977.

\bibitem{malkhi2019flexible}
Dahlia Malkhi, Kartik Nayak, and Ling Ren.
\newblock Flexible byzantine fault tolerance.
\newblock In Lorenzo Cavallaro, Johannes Kinder, XiaoFeng Wang, and Jonathan Katz, editors, {\em Proceedings of the 2019 {ACM} {SIGSAC} Conference on Computer and Communications Security, {CCS} 2019, London, UK, November 11-15, 2019}, pages 1041--1053. {ACM}, 2019.
\newblock \href {https://doi.org/10.1145/3319535.3354225} {\path{doi:10.1145/3319535.3354225}}.

\bibitem{merkle-tree-crypto87}
Ralph~C. Merkle.
\newblock A digital signature based on a conventional encryption function.
\newblock In Carl Pomerance, editor, {\em Advances in Cryptology - {CRYPTO} '87, {A} Conference on the Theory and Applications of Cryptographic Techniques, Santa Barbara, California, USA, August 16-20, 1987, Proceedings}, volume 293 of {\em Lecture Notes in Computer Science}, pages 369--378. Springer, 1987.
\newblock \href {https://doi.org/10.1007/3-540-48184-2\_32} {\path{doi:10.1007/3-540-48184-2\_32}}.

\bibitem{momose2021multi}
Atsuki Momose and Ling Ren.
\newblock Multi-threshold byzantine fault tolerance.
\newblock In Yongdae Kim, Jong Kim, Giovanni Vigna, and Elaine Shi, editors, {\em {CCS} '21: 2021 {ACM} {SIGSAC} Conference on Computer and Communications Security, Virtual Event, Republic of Korea, November 15 - 19, 2021}, pages 1686--1699. {ACM}, 2021.
\newblock \href {https://doi.org/10.1145/3460120.3484554} {\path{doi:10.1145/3460120.3484554}}.

\bibitem{Momose2021}
Atsuki Momose and Ling Ren.
\newblock {Optimal Communication Complexity of Authenticated Byzantine Agreement}.
\newblock In Seth Gilbert, editor, {\em 35th International Symposium on Distributed Computing, {DISC} 2021, October 4-8, 2021, Freiburg, Germany (Virtual Conference)}, volume 209 of {\em LIPIcs}, pages 32:1--32:16. Schloss Dagstuhl - Leibniz-Zentrum f{\"{u}}r Informatik, 2021.
\newblock \href {https://doi.org/10.4230/LIPIcs.DISC.2021.32} {\path{doi:10.4230/LIPIcs.DISC.2021.32}}.

\bibitem{MostefaouiMR15}
Achour Most{\'{e}}faoui, Hamouma Moumen, and Michel Raynal.
\newblock {Signature-Free Asynchronous Binary Byzantine Consensus with t {\textless} n/3, {$O(n^2)$} Messages, and {$O(1)$} Expected Time}.
\newblock {\em J. {ACM}}, 62(4):31:1--31:21, 2015.
\newblock \href {https://doi.org/10.1145/2785953} {\path{doi:10.1145/2785953}}.

\bibitem{Nayak2020}
Kartik Nayak, Ling Ren, Elaine Shi, Nitin~H. Vaidya, and Zhuolun Xiang.
\newblock Improved extension protocols for byzantine broadcast and agreement.
\newblock In Hagit Attiya, editor, {\em 34th International Symposium on Distributed Computing, {DISC} 2020, October 12-16, 2020, Virtual Conference}, volume 179 of {\em LIPIcs}, pages 28:1--28:17. Schloss Dagstuhl - Leibniz-Zentrum f{\"{u}}r Informatik, 2020.
\newblock \href {https://doi.org/10.4230/LIPIcs.DISC.2020.28} {\path{doi:10.4230/LIPIcs.DISC.2020.28}}.

\bibitem{PSL80}
Marshall~C. Pease, Robert~E. Shostak, and Leslie Lamport.
\newblock Reaching agreement in the presence of faults.
\newblock {\em J. {ACM}}, 27(2):228--234, 1980.
\newblock \href {https://doi.org/10.1145/322186.322188} {\path{doi:10.1145/322186.322188}}.

\bibitem{pfitzmann1996information}
Birgit Pfitzmann and Michael Waidner.
\newblock {\em Information-theoretic pseudosignatures and byzantine agreement for $t \geq n/3$}.
\newblock Citeseer, 1996.

\bibitem{reed1960}
Irving~S Reed and Gustave Solomon.
\newblock Polynomial codes over certain finite fields.
\newblock {\em Journal of the society for industrial and applied mathematics}, 8(2):300--304, 1960.

\bibitem{RSS18}
Peter Robinson, Christian Scheideler, and Alexander Setzer.
\newblock Breaking the {$\tilde{\Omega}(\sqrt{n})$} barrier: Fast consensus under a late adversary.
\newblock In Christian Scheideler and Jeremy~T. Fineman, editors, {\em Proceedings of the 30th on Symposium on Parallelism in Algorithms and Architectures, {SPAA} 2018, Vienna, Austria, July 16-18, 2018}, pages 173--182. {ACM}, 2018.
\newblock \href {https://doi.org/10.1145/3210377.3210399} {\path{doi:10.1145/3210377.3210399}}.

\bibitem{Shoup00}
Victor Shoup.
\newblock Practical threshold signatures.
\newblock In Bart Preneel, editor, {\em Advances in Cryptology - {EUROCRYPT} 2000, International Conference on the Theory and Application of Cryptographic Techniques, Bruges, Belgium, May 14-18, 2000, Proceeding}, volume 1807 of {\em Lecture Notes in Computer Science}, pages 207--220. Springer, 2000.
\newblock \href {https://doi.org/10.1007/3-540-45539-6\_15} {\path{doi:10.1007/3-540-45539-6\_15}}.

\bibitem{shoup2024theoretical}
Victor Shoup.
\newblock A theoretical take on a practical consensus protocol.
\newblock {\em Cryptology ePrint Archive}, 2024.

\bibitem{Shoup2023}
Victor Shoup and Nigel~P Smart.
\newblock Lightweight asynchronous verifiable secret sharing with optimal resilience.
\newblock {\em Cryptology ePrint Archive}, 2023.

\bibitem{song2024flexbft}
Anping Song and Cenhao Zhou.
\newblock Flexbft: A flexible and effective optimistic asynchronous bft protocol.
\newblock {\em Applied Sciences}, 14(4):1461, 2024.

\bibitem{spiegelman2020search}
Alexander Spiegelman.
\newblock In search for an optimal authenticated byzantine agreement.
\newblock {\em arXiv preprint arXiv:2002.06993}, 2020.

\bibitem{SrikanthToueg85}
T.~K. Srikanth and Sam Toueg.
\newblock Optimal clock synchronization.
\newblock In Michael~A. Malcolm and H.~Raymond Strong, editors, {\em Proceedings of the Fourth Annual {ACM} Symposium on Principles of Distributed Computing, Minaki, Ontario, Canada, August 5-7, 1985}, pages 71--86. {ACM}, 1985.
\newblock \href {https://doi.org/10.1145/323596.323603} {\path{doi:10.1145/323596.323603}}.

\bibitem{Srikanth1987}
T.~K. Srikanth and Sam Toueg.
\newblock Optimal clock synchronization.
\newblock {\em J. {ACM}}, 34(3):626--645, 1987.
\newblock \href {https://doi.org/10.1145/28869.28876} {\path{doi:10.1145/28869.28876}}.

\bibitem{sui2023signature}
Xiao Sui and Sisi Duan.
\newblock Signature-free atomic broadcast with optimal {$ O (n^2) $} messages and {$ O (1) $} expected time.
\newblock {\em Cryptology ePrint Archive}, 2023.
\newblock \url{https://eprint.iacr.org/2023/1549}.

\bibitem{Thaler22}
Justin Thaler.
\newblock Proofs, arguments, and zero-knowledge.
\newblock {\em Found. Trends Priv. Secur.}, 4(2-4):117--660, 2022.
\newblock \href {https://doi.org/10.1561/3300000030} {\path{doi:10.1561/3300000030}}.

\bibitem{veronese2011efficient}
Giuliana~Santos Veronese, Miguel Correia, Alysson~Neves Bessani, Lau~Cheuk Lung, and Paulo Ver{\'{\i}}ssimo.
\newblock Efficient byzantine fault-tolerance.
\newblock {\em {IEEE} Trans. Computers}, 62(1):16--30, 2013.
\newblock \href {https://doi.org/10.1109/TC.2011.221} {\path{doi:10.1109/TC.2011.221}}.

\bibitem{YPAKT22}
Lei Yang, Seo~Jin Park, Mohammad Alizadeh, Sreeram Kannan, and David Tse.
\newblock Dispersedledger: High-throughput byzantine consensus on variable bandwidth networks.
\newblock In Amar Phanishayee and Vyas Sekar, editors, {\em 19th {USENIX} Symposium on Networked Systems Design and Implementation, {NSDI} 2022, Renton, WA, USA, April 4-6, 2022}, pages 493--512. {USENIX} Association, 2022.
\newblock URL: \url{https://www.usenix.org/conference/nsdi22/presentation/yang}.

\bibitem{Yurek2022}
Thomas Yurek, Licheng Luo, Jaiden Fairoze, Aniket Kate, and Andrew Miller.
\newblock hbacss: How to robustly share many secrets.
\newblock In {\em 29th Annual Network and Distributed System Security Symposium, {NDSS} 2022, San Diego, California, USA, April 24-28, 2022}. The Internet Society, 2022.
\newblock URL: \url{https://www.ndss-symposium.org/ndss-paper/auto-draft-245/}.

\bibitem{ZGKPP17}
Yupeng Zhang, Daniel Genkin, Jonathan Katz, Dimitrios Papadopoulos, and Charalampos Papamanthou.
\newblock vsql: Verifying arbitrary {SQL} queries over dynamic outsourced databases.
\newblock In {\em 2017 {IEEE} Symposium on Security and Privacy, {SP} 2017, San Jose, CA, USA, May 22-26, 2017}, pages 863--880. {IEEE} Computer Society, 2017.
\newblock \href {https://doi.org/10.1109/SP.2017.43} {\path{doi:10.1109/SP.2017.43}}.

\bibitem{ZZZDHWL22}
You Zhou, Zongyang Zhang, Haibin Zhang, Sisi Duan, Bin Hu, Licheng Wang, and Jianwei Liu.
\newblock Dory: Asynchronous {BFT} with reduced communication and improved efficiency.
\newblock {\em {IACR} Cryptol. ePrint Arch.}, page 1709, 2022.
\newblock URL: \url{https://eprint.iacr.org/2022/1709}.

\bibitem{ZLC23}
Jianjun Zhu, Fan Li, and Jinyuan Chen.
\newblock Communication-efficient and error-free gradecast with optimal resilience.
\newblock In {\em {IEEE} International Symposium on Information Theory, {ISIT} 2023, Taipei, Taiwan, June 25-30, 2023}, pages 108--113. {IEEE}, 2023.
\newblock \href {https://doi.org/10.1109/ISIT54713.2023.10206579} {\path{doi:10.1109/ISIT54713.2023.10206579}}.

\end{thebibliography}

\appendix
\section{Graded Consensus: Concrete Implementations} \label{section:graded_consensus_concrete_implementations}

This section presents the implementations of the graded consensus primitive (see \Cref{mod:graded-consensus}) employed in \name and \nameopt.
In \Cref{subsection:hash_based_graded_consensus}, we give our graded consensus algorithm employed in \nameopt: this algorithm exchanges $O(n^2\kappa)$ bits (where $\kappa$ denotes the size of a hash) and terminates in $2$ rounds.
In \Cref{subsection:error_free_graded_consensus_proof}, we introduce our graded consensus algorithm employed in \name: this algorithm exchanges $O(nL + n^2 \log n)$ bits and terminates in $7$ rounds.
We note that both algorithms satisfy all properties specified in \Cref{mod:graded-consensus} only if $t < n/3$. However, the bit and round complexity are guaranteed to be the same even if $t \ge n/3$.

\subsection{Graded Consensus with \texorpdfstring{$O(n^2\kappa)$}{O(n²kappa)} Bits and \texorpdfstring{$2$}{2} Rounds for \nameopt}
\label{subsection:hash_based_graded_consensus}

Our implementation is given in \Cref{algorithm:hash_based_graded_consensus}.
We underline that our implementation is highly inspired by an asynchronous implementation for $t < n / 5$ provided by Attiya and Welch~\cite[Algorithm 2]{AW23}.

\begin{algorithm} [h]
\caption{Graded consensus with $O(n^2\kappa)$ bits: Pseudocode (for process $p_i$)}
\label{algorithm:hash_based_graded_consensus}
\begin{algorithmic} [1]
\footnotesize

\State \textbf{Local variables:}
\State \hskip2em $\mathsf{Hash\_Value}$ $h_i \gets p_i$'s proposal
\State \hskip2em $\mathsf{Hash\_Value}$ $\mathit{branch}_i \gets \bot$

\medskip
\State \emph{Round 1:}
\State \hskip2em \textbf{broadcast} $\langle \textsc{proposal}, h_i \rangle$

\smallskip
\State \hskip2em \textbf{if} exists $\mathsf{Hash\_Value}$ $h$ such that $\langle \textsc{proposal}, h \rangle$ is received from $n - t$ processes:
\State \hskip4em $\mathit{branch}_i \gets h$

\medskip
\State \emph{Round 2:}
\State \hskip2em \textbf{broadcast} $\langle \textsc{branch}, \mathit{branch}_i \rangle$

\smallskip
\State \hskip2em \textbf{if} $\mathit{branch}_i = \bot$:
\State \hskip4em \textbf{if} exists $\mathsf{Hash\_Value}$ $h$ such that $\langle \textsc{branch}, h \rangle$ is received from $t + 1$ processes:
\State \hskip6em \textbf{trigger} $\mathsf{decide}(h, 0)$ \label{line:hash_based_graded_consensus_decide_1}
\State \hskip4em \textbf{else:}
\State \hskip6em \textbf{trigger} $\mathsf{decide}(h_i, 0)$ \label{line:hash_based_graded_consensus_decide_2}
\State \hskip2em \textbf{else:}
\State \hskip4em \textbf{if} $\langle \textsc{branch}, \mathit{branch}_i \rangle$ is received from $n - t$ processes:
\State \hskip6em \textbf{trigger} $\mathsf{decide}(\mathit{branch}_i, 1)$ \label{line:hash_based_graded_consensus_decide_3}
\State \hskip4em \textbf{else:}
\State \hskip6em \textbf{trigger} $\mathsf{decide}(\mathit{branch}_i, 0)$ \label{line:hash_based_graded_consensus_decide_4}
 
\end{algorithmic}
\end{algorithm}

\smallskip
\noindent \textbf{Proof of correctness.}
We start by proving the strong unanimity property.

\begin{theorem} [Strong unanimity]
\Cref{algorithm:hash_based_graded_consensus} satisfies strong unanimity.
\end{theorem}
\begin{proof}
Suppose all correct processes propose the same hash value $h$.
Hence, every correct process $p_i$ updates $\mathit{branch}_i$ to $h$ at the end of the first round.
Therefore, all correct processes broadcast a \textsc{branch} message with $h$ in the second round, which implies that all correct processes receive $h$ from (at least) $n - t$ processes in the second round (as there are at least $n - t$ correct processes).
Therefore, every correct process decides $h$ with grade $1$, which concludes the proof.
\end{proof}

Next, we prove the justification property.

\begin{theorem} [Justification]
\Cref{algorithm:hash_based_graded_consensus} satisfies justification.
\end{theorem}
\begin{proof}
Consider any correct process $p_i$ and let $p_i$ decide a hash value $h$.
We distinguish three possibilities:
\begin{compactitem}
    \item Let $p_i$ decide $h$ at line~\ref{line:hash_based_graded_consensus_decide_1}.
    In this case, $p_i$ has received a \textsc{branch} message for $h$ from a correct process $p_j$.
    Process $p_j$ has previously received a \textsc{proposal} message for $h$ from at least one correct process, which proves the statement of the lemma in this case.

    \item Let $p_i$ decide $h$ at line~\ref{line:hash_based_graded_consensus_decide_2}.
    In this case, $h$ is $p_i$'s proposal, which proves the statement of the lemma.

    \item Let $p_i$ decide $h$ at line~\ref{line:hash_based_graded_consensus_decide_3} or line~\ref{line:hash_based_graded_consensus_decide_4}.
    In this case, $\mathit{branch}_i = h$.
    Hence, $p_i$ has previously (in the first round) received a \textsc{proposal} message for $h$ from a correct process.
    Thus, the statement of the lemma is satisfied in this case as well.
\end{compactitem}
As the statement of the lemma holds in every possible scenario, the proof is concluded.
\end{proof}

The following theorem proves the integrity property.

\begin{theorem} [Integrity]
\Cref{algorithm:hash_based_graded_consensus} satisfies integrity.
\end{theorem}
\begin{proof}
Integrity is trivially satisfied, as each correct process decides at most once.
\end{proof}

Next, we prove the termination property.

\begin{theorem} [Termination]
\Cref{algorithm:hash_based_graded_consensus} satisfies termination.
\end{theorem}
\begin{proof}
All correct processes decide simultaneously as (1) all correct processes propose simultaneously, and (2) \Cref{algorithm:hash_based_graded_consensus} is executed for exactly $2$ rounds by all correct processes.
\end{proof}

Lastly, we need to prove the consistency property.
To this end, we first show that if a correct process $p_i$ sends a \textsc{branch} message for some hash value $h_i$ and another correct process $p_j$ sends a \textsc{branch} message for some hash value $h_j$, then $h_i = h_j$.

\begin{lemma} \label{lemma:same_branch}
If a correct process $p_i$ sends a \textsc{branch} message for a hash value $h_i$ and another correct process $p_j$ sends a \textsc{branch} message for a hash value $h_j$, then $h_i = h_j$.
\end{lemma}
\begin{proof}
As $p_i$ (resp., $p_j$) sends a \textsc{branch} message for a hash value $h_i$ (resp., $h_j$), $p_i$ (resp., $p_j$) has previously received a \textsc{proposal} message for $h_i$ (resp., $h_j$) from $n - t$ processes.
Due to the quorum intersection (as $n - t + n - t - n = n - 2t \geq t + 1$), there exists at least one correct process whose \textsc{proposal} message is received by both $p_i$ and $p_j$.
Therefore, $h_i = h_j$. 
\end{proof}

Finally, we are ready to prove the consistency property.

\begin{theorem} [Consistency]
\Cref{algorithm:hash_based_graded_consensus} satisfies consistency.
\end{theorem}
\begin{proof}
Suppose a correct process $p_i$ decides some hash value $h$ with grade $1$.
Importantly, process $p_i$ does so at line~\ref{line:hash_based_graded_consensus_decide_3} and $p_i$ has previously received a \textsc{branch} message for $h$ from a correct process.
Suppose another correct process $p_j$ decides some hash value $h'$ (with some binary grade).
In order to prove the consistency property, we need to show that $h' = h$.
Let us consider three possible scenarios:
\begin{compactitem}
    \item Let $p_j$ decide $h'$ at line~\ref{line:hash_based_graded_consensus_decide_1}.
    In this case, $p_j$ has received a \textsc{branch} message for $h'$ from a correct process.
    As $p_i$ has received a \textsc{branch} message for $h$ from a correct process, Lemma~\ref{lemma:same_branch} proves that $h = h'$ in this case.

    \item Let $p_j$ decide $h'$ at line~\ref{line:hash_based_graded_consensus_decide_2}.
    As $p_i$ has received $n - t$ \textsc{branch} messages for $h$, process $p_j$ has received at least $n - 2t \geq t + 1$ \textsc{branch} messages for $h$.
    Therefore, this case is impossible.

    \item Let $p_j$ decide $h'$ at line~\ref{line:hash_based_graded_consensus_decide_3} or line~\ref{line:hash_based_graded_consensus_decide_4}.
    In this case, $p_j$ has previously sent a \textsc{branch} message for $h'$.
    Given that $p_i$ has received a \textsc{branch} message for $h$ from a correct process, Lemma~\ref{lemma:same_branch} proves that $h = h'$.
\end{compactitem}
As $h = h'$ in any possible case, the proof is concluded.
\end{proof}


\smallskip
\noindent \textbf{Proof of complexity.}
We start by proving that correct processes send $O(n^2\kappa)$ bits.

\begin{theorem} [Exchanged bits]
Correct processes send $O(n^2\kappa)$ bits in \Cref{algorithm:hash_based_graded_consensus}.
\end{theorem}
\begin{proof}
Each process sends $O(n)$ \textsc{proposal} and \textsc{branch} messages, each with $O(\kappa)$ bits.
Therefore, all correct processes send $O(n^2\kappa)$ bits.
\end{proof}

Finally, it is straightforward to observe that \Cref{algorithm:hash_based_graded_consensus} terminates in $2$ rounds.

\begin{theorem} [Rounds]
\Cref{algorithm:hash_based_graded_consensus} terminates in $2$ rounds.
\end{theorem}

\subsection{Graded Consensus with \texorpdfstring{$O(nL + n^2 \log n)$}{O(nL + n²log(n))}  Bits and \texorpdfstring{$7$}{7} Rounds for \name} 
\label{subsection:error_free_graded_consensus_proof}

Our implementation internally utilizes a binary (processes propose $0$ or $1$) graded consensus algorithm (e.g.,~\cite{Lenzen2022}) and \reducecool, a sub-protocol employed in the COOL protocol proposed in~\cite{Chen2021}.
(COOL is an error-free Byzantine agreement algorithm that satisfies \emph{only} \strongVal and achieves $O(nL + n^2 \log n)$ bit complexity.)
We now briefly introduce the \reducecool protocol.

\medskip
\noindent \textbf{\reducecool~\cite{Chen2021}.}
The formal specification of \reducecool is given in \Cref{mod:reducecool}.
Intuitively, \reducecool reduces the number of outputs to at most 1; if all processes start with the same value, that value is preserved by \reducecool.
Moreover, the processes output a binary success indicator. 
Again, a unanimous proposal leads to a common positive success indicator (i.e., $1$). 
Moreover, if a correct process outputs a positive success indicator, it is guaranteed that all correct processes agree on a common value proposed by some correct process.

\begin{module}
\caption{\reducecool~\cite{Chen2021}}
\label{mod:reducecool}
\footnotesize
\begin{algorithmic}[1]
\Statex
\Statex \textbf{Events:}
\begin{compactitem}
    \item \emph{request} $\mathsf{input}(v \in \mathsf{Value})$: a process inputs a value $v$.

    \item \emph{indication} $\mathsf{output}(v' \in \mathsf{Value} \cup \{\phi\}, s' \in \{0, 1\})$: a process outputs value $v'$ and a success indicator $s'$; $\phi$ denotes a predetermined default value.
\end{compactitem}

\medskip 
\Statex \textbf{Notes:} 
\begin{compactitem}
    \item All correct processes input a value exactly once, and they all do so simultaneously (i.e., in the same round).
\end{compactitem}

\medskip 
\Statex \textbf{Properties:}
\Statex Let $(\omega_i, \mathit{vote}_i)$ denote the output received by a correct process $p_i$.
The following holds:
\begin{compactitem}
\item \emph{Safety:} If $\omega_i \neq \phi$ (where $\phi$ is a predetermined default value) and $\mathit{vote}_i = 1$, then $\omega_i$ is the value input by $p_i$.

\item \emph{Consistency:} If there exists a correct process $p_j$ that outputs a pair $(\cdot, 1)$, then there exists a value $\omega \neq \phi$ such that every correct process outputs $(\omega, \cdot)$.

\item \emph{Obligation:} If every correct process inputs the same value $\omega$, then every correct process outputs $(\omega \neq \phi, 1)$.

\item \emph{Termination:} All correct processes output a pair simultaneously (i.e., in the same round). 
\end{compactitem}

\end{algorithmic}
\end{module}

For completeness, we recall the pseudocode of \reducecool in \Cref{algorithm:reducecool_2} and prove that \reducecool satisfies the specification introduced in \Cref{mod:reducecool} (\Cref{theorem:reducecool_correct}).

\begin{algorithm} [hp]
\caption{\reducecool~\cite{Chen2021}: Pseudocode (for process $p_i$)}
\label{algorithm:reducecool_2}
\begin{algorithmic} [1]
\footnotesize
\State Set $\omega^{(i)} = w_i$. \BlueComment{Value $w_i$ is the proposal of process $p_i$.}
\State Process $p_i$ encodes its message into $n$ symbols, i.e., $[y_1^{(i)}, y_2^{(i)}, ..., y_n^{i}] \gets \mathsf{RSEnc}(w_i, n, k)$, where $k = \lfloor \frac{t}{5} \rfloor + 1$.

\medskip
\State \textbf{Phase 1:}
\State Process $p_i$ sends $(y^{(i)}_j , y^{(i)}_i)$ to process $p_j$, $\forall j \in [1 , n]$.
\For{$j = 1 \text{ to } n$}
    \If{$\big( (y^{(j)}_i , y^{(j)}_j) = (y^{(i)}_i , y^{(i)}_j) \big)$}
        \State Process $p_i$ sets $u_i(j) \gets 1$.
    \Else
        \State Process $p_i$ sets $u_i(j) \gets 0$.
    \EndIf
    \EndFor
    \If{$\left(\sum_{j=1}^{n} u_i(j) \geq n - t\right)$}
        \State Process $p_i$ sets its success indicator as $s_i \gets 1$.
    \Else
        \State Process $p_i$ sets $s_i \gets 0$ and $\omega^{(i)} \gets \phi$.
    \EndIf

\State Process $p_i$ sends the value of $s_i$ to all processes.
\State Process $p_i$ creates sets $S_p \gets \{j : s_j = p, j \in [1 : n]\}$, $p \in \{0, 1\}$, from received $\{s_j\}_{j=1}^{n}$.

\medskip
\State \textbf{Phase 2:}
\If{$(s_i = 1)$}
    \State Process $p_i$ sets $u_i(j) \gets 0$, $\forall j \in S_0$.
    \If{$\left(\sum_{j=1}^{n} u_i(j) < n - t\right)$}
        \State Process $p_i$ sets $s_i \gets 0$ and $\omega^{(i)} \gets \phi$.
    
    \State Process $p_i$ sends the value of $s_i$ to all processes.
    \EndIf
    \EndIf
    \State Process $p_i$ updates $S_0$ and $S_1$ based on the newly received success indicators.

\medskip
\State \textbf{Phase 3:}
\If{$(s_i = 1)$}
    \State Process $i$ sets $u_i(j) \gets 0$, $\forall j \in S_0$.
    \If{$\left(\sum_{j=1}^{n} u_i(j) < n - t\right)$}
        \State Process $p_i$ sets $s_i \gets 0$ and $\omega^{(i)} \gets \phi$.
    
    \State Process $p_i$ sends the value of $s_i$ to all processes.
    \EndIf
    \EndIf
    
    \State Process $p_i$ updates $S_0$ and $S_1$ based on the newly received success indicators.
    \State Process $p_i$ sends $\langle \textsc{symbol}, y_j^{(i)} \rangle$ to process $p_j$, $\forall j \in [1, n]$.
    \If{$\left(|S_1| \geq 2t + 1\right)$} \label{line:reduce_cool_S1geq2tp1}
        \State Process $p_i$ sets the binary vote as $v_i = 1$. \label{line:binary_vote_1}
    \Else
        \State Process $p_i$ sets the binary vote as $v_i = 0$.
    \EndIf

\medskip
\State \textbf{Phase 4:}
\If{$s_i = 1$}
    \State Process $p_i$ sends $\langle \textsc{reconstruct}, y_i^{(i)} \rangle$ to all processes. \label{line:reduce_cool_bcast_reconstruct_1}
\Else
    \State Let $y_i^{(i)} \gets \mathsf{majority}(\text{RS symbols received via \textsc{symbol} messages})$. \label{line:majority}
    \State Process $p_i$ sends $\langle \textsc{reconstruct}, y_i^{(i)} \rangle$ to all processes. \label{line:reduce_cool_bcast_reconstruct_2}
\EndIf

\smallskip
\If{$s_i = 1$}
    \State \textbf{return} $(\omega^{(i)}, v_i)$ \label{line:reduce_cool_decide_1_its}
\Else  
    \State For every process $p_j$, let $m_j$ denote the RS symbol sent by process $p_j$ via a \textsc{reconstruct} message (if such a \hphantom{allll}message is not received, $m_j = \bot$).
    \State \textbf{return} $\big( \mathsf{RSDec}(k, t, \{ (j, m_j) \,|\, p_j \in \Pi \}, 0 \big)$ \label{line:reduce_cool_decide_2_its}
\EndIf
    
\end{algorithmic}
\end{algorithm}

\begin{theorem} [\reducecool is correct] \label{theorem:reducecool_correct}
\reducecool (\Cref{algorithm:reducecool_2}) satisfies the specification of \Cref{mod:reducecool}.
\end{theorem}
\begin{proof}
The proof of the theorem follows directly from the extended version~\cite{Chen2021extended} of~\cite{Chen2021}:
\begin{compactitem}
    \item Termination is ensured by~\cite[Lemma 6]{Chen2021extended}.
    \item Safety is ensured as $p_i$ never updates its variable $\omega^{(i)}$ to any other non-$\phi$ value (see \Cref{algorithm:reducecool_2}).

    \item Obligation is ensured by \cite[Lemma 5 (appendix H)]{Chen2021extended}.

    \item Consistency is ensured by \cite[Lemma 3]{Chen2021extended} and \cite[Lemma 4 (appendix G, equations 130 to 133)]{Chen2021extended}.
    Let $(\omega_i, \mathit{vote}_i)$ be the pair output by a correct process $p_i$.
    More precisely,~\cite[Lemma 3]{Chen2021extended} proves the following property:
    \begin{compactitem}
        \item \emph{Non-duplicity:} If there exists a correct process $p_i$ such that, at the end of phase 3 of the \reducecool protocol (\Cref{algorithm:reducecool_2}), the binary vote $v_i$ is equal to $1$ (line~\ref{line:binary_vote_1}), then $|\{\omega_j\,|\, \omega_j \neq \phi \text{ and } p_j \text{ is correct} \}| \leq 1$. \label{item:non_duplicity}
    \end{compactitem}
    Moreover,~\cite[Lemma 4 (appendix G, equations 130 to 133)]{Chen2021extended} ensures the following:
    \begin{compactitem}
    \item \emph{Retrievability:} If there exists a correct process $p_i$ such that, at the end of phase 3 of the \reducecool protocol (\Cref{algorithm:reducecool_2}), the binary vote $v_i$ is equal to $1$ (line~\ref{line:binary_vote_1}), then $| \{ p_j \,|\, p_j \text{ is correct and } p_j \in S_1 \text{ at process $p_i$ and }  \omega_j \neq \phi\}| \geq t + 1$. \label{item:retrievability}
     \end{compactitem}
Suppose a correct process $p_i$ obtains $(\omega_i, 1)$ from \reducecool.
Hence, $v_i = 1$ at the end of phase 3 of the \reducecool protocol (\Cref{algorithm:reducecool_2}).
By the retrievability property, a set $K$ of at least $t + 1$ correct processes $p_k$ has obtained $w_k \neq \phi$ and $s_k = 1$ at the end of phase 3 of the \reducecool protocol and, due to the non-duplicity property, holds the same value $\bar{w} = \omega_i$, i.e., for every correct process $p_k \in K$, $w_k = \bar{w} = \omega_i$. 
Thus, for every correct process $p_l$ with $s_l = 0$, $p_l$ obtains a correctly-encoded RS symbol of value $\omega_i$ at line~\ref{line:majority} (as at most $t$ incorrect and at least $t + 1$ correct symbols are received by $p_l$). 
Therefore, every correct process that sends a symbol (line~\ref{line:reduce_cool_bcast_reconstruct_1} or line~\ref{line:reduce_cool_bcast_reconstruct_2}) does send a correctly-encoded RS symbol of value $\omega_i$, which means that any correct process that decides at line~\ref{line:reduce_cool_decide_2_its} does decide the same value $\omega_i$.
\end{compactitem}
As all properties are indeed satisfied, the theorem holds.
\end{proof}

\smallskip
\noindent \textbf{Our implementation.}
Our implementation of the graded consensus primitive with $O(nL + n^2 \log n)$ bits is given in \Cref{algorithm:gccool_protocol_2}.

\begin{algorithm} [h]
\caption{Graded consensus with $O(nL + n^2 \log n)$ bits: Pseudocode (for process $p_i$)}
\label{algorithm:gccool_protocol_2}
\begin{algorithmic} [1]
\footnotesize
\State \textbf{Uses:}
\State \hskip2em \reducecool~\cite{Chen2021}, \textbf{instance} $\mathcal{R}_{\mathit{cool}}$ \BlueComment{bits: $O(nL + n^2\log n )$; rounds: $5$}
\State \hskip2em Graded consensus proposed in~\cite{Lenzen2022}, \textbf{instance} $\mathcal{GC}$ \BlueComment{bits: $O(n^2)$; rounds: $2$}

\medskip
\State \textbf{Local variables:}
\State \hskip2em $\mathsf{Value}$ $v_i \gets p_i$'s proposal

\medskip
\State \textbf{Pseudocode:}
\State \hskip2em Initialize local variable $(\omega_i, \mathit{vote}_i) \in \mathsf{Value} \times \{ 0, 1 \}$
\State \hskip2em \textcolor{blue}{ \(\triangleright\) $\mathcal{R}_{\mathit{cool}}$ ensures that if $(\omega_i, 1)$ is returned, then $(\omega_i, \cdot)$ is returned to all correct processes}
\State \hskip2em Let $(\omega_i, \mathit{vote}_i) \gets \mathcal{R}_{\mathit{cool}}(v_i)$ \BlueComment{$v_i$ is the input, and a pair $(\omega_i, \mathit{vote}_i)$ is the output}

\smallskip
\State \hskip2em Initialize local variable $(b_i, g_i) \in \{0, 1\} \times \{0, 1\}$
\State \hskip2em \textcolor{blue}{ \(\triangleright\) if $p_i$ decides $(b, \cdot)$ from $\mathcal{GC}$, then $b$ was proposed by a correct process (due to $\mathcal{GC}$'s justification)}
\State \hskip2em Let $(b_i, g_i) \gets \mathcal{GC}.\mathsf{propose}(\mathit{vote}_i)$ \BlueComment{$\mathit{vote}_i$ is the proposal, and $b_i$ is decided with grade $g_i$}
\State \hskip2em \textbf{if} $b_i = 0$:
\State \hskip4em \textbf{trigger} $\mathsf{decide}(v_i, 0)$ \label{line:decide_1_reduce_cool_gc}
\State \hskip2em \textbf{else:} \BlueComment{$b_i = 1$}
\State \hskip4em \textbf{trigger} $\mathsf{decide}(\omega_i, g_i)$ \label{line:decide_2_reduce_cool_gc}

\end{algorithmic}
\end{algorithm}

\smallskip
\noindent \textbf{Proof of correctness.}
We start by proving the strong unanimity.

\begin{theorem} [Strong unanimity]
\Cref{algorithm:gccool_protocol_2} satisfies strong unanimity.
\end{theorem}
\begin{proof}
Suppose all correct processes propose the same value $v$.
By the obligation property of the instance $\mathcal{R}_{\mathit{cool}}$ of \reducecool, every correct process outputs $(v, 1)$ from $\mathcal{R}_{\mathit{cool}}$.
Hence, every correct process proposes $1$ to the graded consensus instance $\mathcal{GC}$.
This implies that every correct process decides $(1, 1)$ from $\mathcal{GC}$ (due to the strong unanimity property of $\mathcal{GC}$).
Thus, every correct process decides $v$ with grade $1$ from \Cref{algorithm:gccool_protocol_2}, which concludes the proof.
\end{proof}

Next, we prove the justification property.

\begin{theorem} [Justification]
\Cref{algorithm:gccool_protocol_2} satisfies justification.
\end{theorem}
\begin{proof}
Consider any correct process $p_i$ that decides some value $v$.
We distinguish two possibilities:
\begin{compactitem}
    \item Let $p_i$ decide $v$ at line~\ref{line:decide_1_reduce_cool_gc}.
    In this case, $v$ is $p_i$'s proposal, which proves the statement.

    \item Let $p_i$ decide $v$ at line~\ref{line:decide_2_reduce_cool_gc}.
    In this case, $p_i$ has decided $b_i = 1$ from $\mathcal{GC}$.
    The justification property of $\mathcal{GC}$ ensures that some correct process $p_j$ has previously proposed $1$ to $\mathcal{GC}$.
    Therefore, $p_j$ has decided a pair $(v', 1)$ from $\mathcal{R}_{\mathit{cool}}$.
    The safety property of $\mathcal{R}_{\mathit{cool}}$ ensures that $v'$ is $p_j$'s proposal.
    Moreover, the consistency property of $\mathcal{R}_{\mathit{cool}}$ ensures that every correct process outputs $v'$ from $\mathcal{R}_{\mathit{cool}}$.
    Thus, $v' = v$.
    As $v$ is $p_j$'s proposal, the statement of the lemma holds in this case.
\end{compactitem}
The theorem holds as its statement stands in both possible cases.
\end{proof}

The following theorem proves the consistency property.

\begin{theorem} [Consistency]
\Cref{algorithm:gccool_protocol_2} satisfies consistency.
\end{theorem}
\begin{proof}
Let a correct process $p_i$ decide a value $v$ with grade $1$; note that $p_i$ decides at line~\ref{line:decide_2_reduce_cool_gc}.
This implies that $p_i$ has previously decided $(1, 1)$ from $\mathcal{GC}$, which further implies that all correct processes decide $(1, \cdot)$ from $\mathcal{GC}$ (due to the consistency property of $\mathcal{GC}$).
Moreover, the justification property of $\mathcal{GC}$ ensures that some correct process $p_j$ has previously proposed $1$ to $\mathcal{GC}$.
Hence, process $p_j$ has output a pair $(v', 1)$ from $\mathcal{R}_{\mathit{cool}}$, for some value $v'$.
The consistency property of $\mathcal{R}_{\mathit{cool}}$ ensures that $v' = v$ is output by every correct process from $\mathcal{R}_{\mathit{cool}}$.
As every correct process necessarily decides at line~\ref{line:decide_2_reduce_cool_gc} (since no correct process decides $(0, \cdot)$ from $\mathcal{GC}$), no correct process decides a value different from $v$ from \Cref{algorithm:gccool_protocol_2}. 
Therefore, the consistency property is ensured.
\end{proof}

Next, we prove the integrity property.

\begin{theorem} [Integrity]
\Cref{algorithm:gccool_protocol_2} satisfies integrity.
\end{theorem}
\begin{proof}
The theorem follows immediately from \Cref{algorithm:gccool_protocol_2}.
\end{proof}

Lastly, we prove the termination property.

\begin{theorem} [Termination]
\Cref{algorithm:gccool_protocol_2} satisfies termination.
\end{theorem}
\begin{proof}
The termination property follows from the termination properties of $\mathcal{R}_{\mathit{cool}}$ and $\mathcal{GC}$.
\end{proof}

\smallskip
\noindent \textbf{Proof of complexity.}
We first show that correct processes exchange $O(nL + n^2 \log n)$ bits.

\begin{theorem} [Exchanged bits]
Correct processes send $O(nL + n^2 \log n)$ bits in \Cref{algorithm:gccool_protocol_2}.
\end{theorem}
\begin{proof}
Each RS symbol is of size $O( \frac{L}{n} + \log n )$ bits.
Then, every correct process sends $4n$ RS symbols in \reducecool. Apart from that, each correct process sends $O(n)$ binary success indicators in \reducecool. 
Furthermore, each correct process sends $O(n)$ bits in the binary graded consensus (see~\cite{Lenzen2022}).
Hence, each correct process sends $O( L + n \log n )$ bits, which leads to $O(nL + n^2 \log n)$ bits in total.
\end{proof}

Lastly, we prove that \Cref{algorithm:gccool_protocol_2} terminates in $7 \in O(1)$ rounds.

\begin{theorem} [Rounds]
\Cref{algorithm:gccool_protocol_2} terminates in $7$ rounds.
\end{theorem}
\begin{proof}
The algorithm terminates in $7$ rounds as $\mathcal{R}_{\mathit{cool}}$ terminates in $5$ rounds and $\mathcal{GC}$ terminates in $2$ rounds.
\end{proof}
\section{\nameopt: Missing Implementation \& Proof} \label{section:nameopt_proof}

We first define the cryptographic primitives utilized in \nameopt and its building blocks (\Cref{section:cryptographic_preliminaries}).
Then, we present our implementation of the data dissemination primitive (see \Cref{mod:hashfin}) and prove its correctness and complexity (\Cref{subsection:hashfin_proof}).
Finally, we provide \nameopt's proof (\Cref{subsection:nameopt_proof}).

\subsection{Cryptographic Primitives}\label{section:cryptographic_preliminaries}

This section provides the formal definition of cryptographic accumulators \cite{DHS15} and digests.

\smallskip
\noindent \textbf{Cryptographic accumulators.}
We follow the presentation of~\cite{Nayak2020}.
Let $\kappa$ denote a security parameter, and let $D = \{d_1, d_2, \ldots, d_n\}$ be a set of $n$ values.
An accumulator scheme consists of the following four primitives:
\begin{compactitem}
    \item $ak \leftarrow \mathsf{Gen}(1^\kappa, n)$: An algorithm accepting a security parameter $\kappa$ in unary notation $1^\kappa$ and a limit $n$ on the set size, and returns an accumulator key $ak$.
    \item $z \leftarrow \mathsf{Eval}(ak, D)$: An algorithm that takes the accumulator key $ak$ and a set $D$ of values, and returns an accumulation value $z$ for $D$.
    \item $w_i \leftarrow \mathsf{CreateWit}(ak, z, d_i)$: Given an accumulator key $ak$, an accumulation value $z$, and a value $d_i \in D$, this algorithm returns a witness $w_i$.
    \item $\mathit{true}/\mathit{false} \leftarrow \mathsf{Verify}(ak, z, w_i, d_i)$: This algorithm takes an accumulator key $ak$, an accumulation value $z$, a witness $w_i$, and a value $d_i$. It returns $\mathit{true}$ if $w_i$ correctly corresponds to $d_i \in D$, and $\mathit{false}$ otherwise, except with negligible probability $negl(\kappa)$.
\end{compactitem}

In this paper, we use simple Merkle trees, where the accumulator key $ak$ is the hash function, values in $D$ are the leaves of the Merkle tree, the accumulation value $z$ is the Merkle tree root, and the witness $w$ is the Merkle tree proof. That is why the description leaves out auxiliary information $\mathit{aux}$ typically included in cryptographic accumulators. We remark that $\mathsf{Eval}$, $\mathsf{CreateWit}$, and $\mathsf{Verify}$ are deterministic for a Merkle-tree-based accumulators. Furthermore, the accumulation value and witness each have $O(\kappa)$ and $O(\kappa \log n)$ bits, respectively.

A secure accumulator scheme is required to be correct. Correctness says that for all honestly generated accumulator keys, all honestly computed accumulation values, and witnesses, the $\mathsf{Verify}$ algorithm always returns $\mathit{true}$. Let us underline that the correctness of the Merkle-tree-based accumulator is trivial.
Moreover, the accumulator scheme is assumed to be collision-free, i.e., for any accumulator key $ak \gets \mathsf{Gen}(1^{\kappa}, n)$, it is computationally impossible to establish $(\{d_1, ..., d_n\}, d', w')$ such that (1) $d' \notin \{d_1, ..., d_n\}$, (2) $z \leftarrow \mathsf{Eval}(ak, \{d_1, ..., d_n\})$, and (3) $\mathsf{Verify}(ak, z, w', d') = \mathit{true}$. In our case, collision-resistance reduces to the collision-resistance of the underlying chosen hash function.

\smallskip
\noindent \textbf{Digests.}
Concretely, $\mathsf{digest}(v) = \mathsf{Eval}(ak, \{(1,P_v(1)), ..., (n,P_v(n))\})$, where $\mathsf{Eval}$ is the accumulator evaluation function (see the paragraph above), $ak$ is the accumulator key (see above), and $[P_v(1), ..., P_v(n)] = \mathsf{RSEnc}(v,n,t+1)$ is the Reed-Solomon encoding of value $v$ (see \Cref{subsection:preliminaries}). The collision-resistance of the $\mathsf{digest}(\cdot)$ function reduces to the collision-resistance of the underlying cryptographic accumulator scheme.
We underline that this construction is standard (see, e.g., \cite{Nayak2020}).

\subsection{Data Dissemination: Implementation \& Proof} \label{subsection:hashfin_proof}

Our implementation of the data dissemination primitive is given in \Cref{algorithm:hash_based_add}.

\begin{algorithm} [ht]
\caption{Data dissemination: Pseudocode (for process $p_i$)}
\label{algorithm:hash_based_add}
\begin{algorithmic} [1]
\footnotesize

\State \textbf{Rules:}
\State \hskip2em Any \textsc{disperse} or \textsc{reconstruct} message with an invalid witness is ignored.
\State \hskip2em The rules at lines~\ref{line:received_disperse_hash_add} and~\ref{line:received_reconstruct_hash_add} are activated only if $p_i$ has previously invoked an $\mathsf{input}(\cdot, \cdot)$ request.

\medskip
\State \textbf{Local variables:}
\State \hskip2em $\mathsf{Digest}$ $d_i \gets \bot$
\State \hskip2em $\mathsf{Boolean}$ $\mathit{reconstruct\_broadcast}_i \gets \mathit{false}$

\medskip
\State \textbf{upon} $\mathsf{input}(v \in \mathsf{Value} \cup \{\bot\}, d \in \mathsf{Digest})$: \label{line:input_hashfin}
\State \hskip2em $d_i \gets d$
\State \hskip2em \textbf{if} $v \neq \bot$: \BlueComment{if true, then $d_i = \mathsf{digest}(v)$ (see Module~\ref{mod:hashfin})} \label{line:check_value_hashfin}
\State \hskip4em Let $[m_1, m_2, ..., m_n] \gets \mathsf{RSEnc}(v, n, t + 1)$ \label{line:encode_value_hashfin}
\State \hskip4em \textbf{for each} process $p_j$:
\State \hskip6em Let $\mathcal{W}_j \gets \mathsf{CreateWit}\big( d_i, (j, m_j), [(1, m_1), (2, m_2), ..., (n, m_n)] \big)$
\State \hskip6em \textbf{send} $\langle \textsc{disperse}, d_i, (j, m_j), \mathcal{W}_j \rangle$ to process $p_j$ \label{line:send_disperse_hashfin}


\medskip
\State \textbf{upon} receiving a $\langle \textsc{disperse}, d_i, (i, m_i), \mathcal{W}_i \rangle$ message and $\mathit{reconstruct\_broadcast}_i = \mathit{false}$: \label{line:received_disperse_hash_add}
\State \hskip2em $\mathit{reconstruct\_broadcast}_i \gets \mathit{true}$
\State \hskip2em \textbf{broadcast} $\langle \textsc{reconstruct}, d_i, (i, m_i), \mathcal{W}_i \rangle$ \label{line:broadcast_reconstruct_hash_add}

\medskip
\State \textbf{upon} receiving a $\langle \textsc{reconstruct}, d_i, \cdot, \cdot \rangle$ message from $t + 1$ processes and $\mathit{reconstruct\_broadcast}_i = \mathit{true}$: \label{line:received_reconstruct_hash_add}
\State \hskip2em \textbf{trigger} $\mathsf{output}\big( \mathsf{RSDec}(t + 1, 0, \text{the received $t + 1$ RS symbols}) \big)$ \label{line:output_hashfin}

\end{algorithmic}
\end{algorithm}

\noindent \textbf{Proof of correctness.}
We say that $m_i$ is a \emph{correctly encoded $i$-th RS symbol} if and only if $m_i$ is the $i$-th RS symbol in $\mathsf{RSEnc}(v^{\star}, n, t + 1)$.
First, we prove that if any correct process receives a pair $(i, m_i)$ accompanied by a valid witness via a \textsc{disperse} or \textsc{reconstruct} message, then $m_i$ is a correctly encoded $i$-th RS symbol.

\begin{lemma} \label{lemma:rs_symbols_correct}
If a correct process receives a pair $(i, m_i)$ accompanied by a valid witness via a \textsc{disperse} or \textsc{reconstruct} message, then $m_i$ is a correctly encoded $i$-th RS symbol.
\end{lemma}
\begin{proof}
Let $p_j$ be any correct process that receives the aforementioned pair $(i, m_i)$.
The pair is accompanied by a valid witness that proves that $(i, m_i)$ belongs to a collection whose digest is $d^{\star} = \mathsf{digest}(v^{\star})$.
Therefore, $m_i$ is indeed a correctly encoded $i$-th RS symbol.
\end{proof}

We are ready to prove the safety property of \Cref{algorithm:hash_based_add}.

\begin{theorem} [Safety]
\Cref{algorithm:hash_based_add} satisfies safety.
\end{theorem}
\begin{proof}
Consider any correct process $p_i$ that outputs some value $v$ (line~\ref{line:output_hashfin}).
Before outputting $v$, $p_i$ has received (at least) $t + 1$ \textsc{reconstruct} messages from as many different processes (line~\ref{line:received_reconstruct_hash_add}).
Due to Lemma~\ref{lemma:rs_symbols_correct}, each such RS symbol is correctly encoded, which implies that $v = v^{\star}$.
\end{proof}

Next, we prove the liveness property of \Cref{algorithm:hash_based_add}.

\begin{theorem} [Liveness]
\Cref{algorithm:hash_based_add} satisfies liveness.
\end{theorem}
\begin{proof}
Recall that we assume that (1) all correct processes input a pair by some time $\tau$, and (2) at least one correct process inputs $v^{\star}$.
Moreover, we assume that no correct process stops unless it has previously output a value.
Hence, every correct process $p_i$ receives a \textsc{disperse} message with the $i$-th RS symbol by time $\tau + \delta$ (line~\ref{line:received_disperse_hash_add}), which implies that $p_i$ broadcasts a \textsc{reconstruct} message with a correctly encoded $i$-th RS symbol (due to Lemma~\ref{lemma:rs_symbols_correct}) by time $\tau + \delta$ (line~\ref{line:broadcast_reconstruct_hash_add}).
Thus, every correct process receives $n - t \geq t + 1$ \textsc{reconstruct} messages by time $\tau + 2\delta$ (line~\ref{line:received_reconstruct_hash_add}), which implies that every correct process outputs a value from \Cref{algorithm:hash_based_add} by time $\tau + 2\delta$ (line~\ref{line:output_hashfin}).
\end{proof}

Lastly, we prove that \Cref{algorithm:hash_based_add} satisfies integrity.

\begin{theorem} [Integrity]
\Cref{algorithm:hash_based_add} satisfies integrity.
\end{theorem}
\begin{proof}
The theorem follows from the fact that no correct process activates the rule at line~\ref{line:received_reconstruct_hash_add} unless it has previously invoked an $\mathsf{input}(\cdot, \cdot)$ request.
\end{proof}

\smallskip
\noindent \textbf{Proof of complexity.}
We now prove that all correct processes send $O(nL + n^2 \kappa \log n)$ bits.

\begin{theorem} [Exchanged bits]
Correct processes collectively send $O(nL + n^2\kappa \log n )$ bits.
\end{theorem}
\begin{proof}
Each correct process broadcasts at most one \textsc{disperse} and at most one \textsc{reconstruct} message.
Each \textsc{disperse} and \textsc{reconstruct} message contains $O(\kappa + \log n + \frac{L}{n} + \log n + \kappa \log n)$ bits.
Thus, each correct process sends $n \cdot O(\frac{L}{n} + \kappa \log n) = O(L + n\kappa \log n)$ bits, which implies that correct processes collectively send $O(nL + n^2 \kappa \log n)$ bits.
\end{proof}

\subsection{Proof of Correctness \& Complexity} \label{subsection:nameopt_proof}

Finally, we prove the correctness and complexity of \nameopt (\Cref{algorithm:optimal_trivial}).

\smallskip
\noindent \textbf{Proof of correctness.}
We say that a correct process $p_i$ \emph{commits} a digest $D$ in view $V \in [1, t + 1]$ if and only if $p_i$ invokes $\mathcal{DD}.\mathsf{input}(\cdot, D)$ in view $V$ (line~\ref{line:commit_hash_value}).
Recall that if a correct process $p_i$ commits a digest in some view $V < t + 1$, $p_i$ does not stop before it has completed view $V + 1$ (line~\ref{line:wait_for_completion}). 
We start by proving that if the first correct process commits a digest $D$ in some view $V$, then all correct processes commit $D$ (and no other digest) in view $V$ or $V + 1$ (given $V < t + 1$).

\begin{lemma} \label{lemma:first_commit}
Let $p_i$ be the first correct process to commit a digest.
Let $p_i$ commit a digest $D$ in view $V$.
Then, the following holds:
\begin{compactitem}
    \item If any correct process commits a digest $D'$, then $D' = D$.

    \item If $V < t + 1$, then every correct process commits a digest in view $V$ or $V + 1$.
\end{compactitem}
\end{lemma}
\begin{proof}
Recall that \nameopt ensures that each correct process $p_i$ commits at most one digest (due to the $\mathit{committed\_view}_i$ variable).
As $p_i$ is the first correct process to commit a digest, no correct process commits in any view smaller than $V$.
Moreover, if any correct process commits any digest in view $V$, that digest is $D$ due to the consistency property of $\mathcal{GC}_2[V]$.

Let $V < t + 1$.
Consider now view $V + 1$.
Recall that $p_i$ (and any other correct process that has committed in view $V$) does not stop before completing view $V + 1$ (line~\ref{line:wait_for_completion}).
As $p_i$ has committed $D$ in view $V$, the consistency property of $\mathcal{GC}_2[V]$ ensures that every correct process $p_k$ updates its $\mathit{locked}_k$ variable to $D$ (Step 4).
Therefore, all correct processes propose $D$ to $\mathcal{GC}_1[V + 1]$ in view $V + 1$ (Step 1).
The strong unanimity property of $\mathcal{GC}_1[V + 1]$ guarantees that all correct processes decide $(D, 1)$ from $\mathcal{GC}_1[V + 1]$.
Thus, every correct process broadcasts a $\langle \textsc{support}, D \rangle$ message (Step 2).
This implies that every correct process $p_k$ updates its $\mathit{vote}_k$ variable to $D$ (Step 3), and proposes $D$ to $\mathcal{GC}_2[V + 1]$ (Step 4).
Therefore, every correct process $p_k$ decides $(D, 1)$ from $\mathcal{GC}_2[V + 1]$ (due to the strong unanimity property of $\mathcal{GC}_2[V + 1]$).
Finally, every correct process $p_k$ commits $D$ in view $V + 1$ (Step 5), which concludes the proof.
\end{proof}

The following lemma proves that if a correct process decides $(D \neq \bot, \cdot)$ from $\mathcal{GC}_1[V]$ in some view $V \in [2, t + 1]$, then a correct process has sent a $\langle \textsc{support}, D \rangle$ message in a view $V' < V$.

\begin{lemma} \label{lemma:decide_gc1_implies_happy}
If any correct process decides $(D \neq \bot, \cdot)$ from $\mathcal{GC}_1[V]$ in any view $V \in [2, t + 1]$, then a correct process has sent a $\langle \textsc{support}, D \rangle$ message in a view $V' < V$. 
\end{lemma}
\begin{proof}
Due to the justification property of $\mathcal{GC}_1[V]$, a correct process $p_i$ has proposed $D$ to $\mathcal{GC}_1[V]$, which implies that $\mathit{locked}_i = D$ at the beginning of view $V$.
Let $V' < V$ denote the view at which $p_i$ updates $\mathit{locked}_i$ to $D$ (Step 5).
The justification property of $\mathcal{GC}_2[V']$ ensures that some correct process $p_k$ has proposed $D$ to $\mathcal{GC}_2[V']$ in view $V'$. 
Thus, the value of the $\mathit{vote}_k$ variable at process $p_k$ in view $V'$ was $D$, which implies that $p_k$ has received a $\langle \textsc{support}, D \rangle$ message from $2t + 1$ processes in view $V'$.
As there are at most $t$ faulty processes, (at least) $t + 1$ of those messages are sent by correct processes, which concludes the proof.
\end{proof}

Next, we prove that if a correct process sends a $\textsc{support}$ message for a digest, the value that corresponds to that digest is (1) known to at least one correct process, and (2) valid. 

\begin{lemma} \label{lemma:happy_implies_known}
If a correct process sends a $\langle \textsc{support}, D \rangle$ in some view $V$, then there exists a correct process $p_i$ that has set its $\mathit{known\_values}_i[D]$ variable to a value $v$ in some view $V' \leq V$, where $D = \mathsf{digest}(v)$ and $\mathsf{valid}(v) = \mathit{true}$.
\end{lemma}
\begin{proof}
Let $p^*$ be the first correct process to ever send a $\langle \textsc{support}, D \rangle$ message, for the given digest $D$.
Let the aforementioned message by $p^*$ be sent in some view $V^*$.
To prove the lemma, it suffices to show that some correct process $p_i$ has set its $\mathit{known\_values}_i[D]$ variable to a value $v$ in some view $V' \leq V^*$, where $D = \mathsf{digest}(v)$ and $\mathsf{valid}(v) = \mathit{true}$.
We consider three scenarios:
\begin{compactitem}
    \item Let $p^*$ send the message at line~\ref{line:broadcast_happy_1}.
    In this case, $p^*$ has decided $(D \neq \bot, 1)$ from $\mathcal{GC}_1[V^*]$ in view $V^*$.
    (Note that this implies that $V^* \in [2, t + 1]$.)
    Due to Lemma~\ref{lemma:decide_gc1_implies_happy}, there exists a correct process that has sent a $\textsc{support}$ message for $D$ in a view smaller than $V^*$.
    This contradicts the statement that the first $\textsc{support}$ message from a correct process for $D$ is sent in view $V^*$.
    Therefore, this scenario is impossible.

    \item Let $p^*$ send the message at line~\ref{line:broadcast_happy_2}.
    In this case, $p^*$ has received $D$ from $\mathsf{leader}(V^*)$ and $D$ is accepted by $p^*$ in a view smaller than $V^*$.
    Again, this implies that a correct process has sent a \textsc{support} message for $D$ in a view smaller than $V^*$, which is impossible.

    \item Let $p^*$ send the message at line~\ref{line:broadcast_happy_3}.
    In this case, process $p^*$ has received a value $v$ from $\mathsf{leader}(V^*)$ such that $\mathsf{digest}(v) = D$ and $\mathsf{valid}(v) = \mathit{true}$.
    Therefore, the statement of the lemma is satisfied in this case.
\end{compactitem}
As the statement of the lemma is satisfied in the only possible scenario, the lemma holds.
\end{proof}

Let us denote by $\mathcal{V}$ the first view in which a correct process commits a digest.
If such a view does not exist, $\mathcal{V} = \bot$.
Recall that Lemma~\ref{lemma:first_commit} proves that if $\mathcal{V} \neq \bot$, then all correct processes commit the same digest and they do so by the end of view $\mathcal{V} + 1$.
Let $\mathcal{D}$ denote the common digest committed by correct processes; if $\mathcal{V} = \bot$, then $\mathcal{D} = \bot$.
Lastly, given that $\mathcal{V} \neq \bot$ (and $\mathcal{D} \neq \bot$), let $v^{\star}$ denote the value such that $\mathcal{D} = \mathsf{digest}(v^{\star})$; 
note that $v^{\star}$ is unique due to the collision-resistance of the accumulator primitive. 
The following lemma proves that if $\mathcal{V} \neq \bot$ (i.e., a correct process commits a digest), then some correct process $p_i$ includes a value $v^{\star}$ in its $\mathit{known\_values}_i$ variable in a view smaller than or equal to $\mathcal{V}$.

\begin{lemma} \label{lemma:committed_implies_known}
Let $\mathcal{V} \neq \bot$.
Then, there exists a correct process $p_i$ that sets its $\mathit{known\_values}_i[\mathcal{D}]$ variable to value $v^{\star}$ in a view $V \leq \mathcal{V}$.
\end{lemma}
\begin{proof}
As $\mathcal{D}$ is committed in view $\mathcal{V}$ by some correct process $p_i$, process $p_i$ has decided $(\mathcal{D}, 1)$ from $\mathcal{GC}_2[\mathcal{V}]$ in view $\mathcal{V}$.
Therefore, the justification property of $\mathcal{GC}_2[\mathcal{V}]$ ensures that some correct process $p_j$ has proposed $\mathcal{D}$ to $\mathcal{GC}_2[\mathcal{V}]$.
This implies that $\mathit{vote}_j = \mathcal{D}$, which further proves that $p_j$ has received a $\textsc{support}$ message for $\mathcal{D}$ in view $\mathcal{V}$ from a correct process.
Therefore, the lemma follows directly from Lemma~\ref{lemma:happy_implies_known}.
\end{proof}

Next, we prove that value $v^{\star}$ is valid.

\begin{lemma} \label{lemma:valid_v_star}
Let $\mathcal{V} \neq \bot$.
Then, $\mathsf{valid}(v^{\star}) = \mathit{true}$.
\end{lemma}
\begin{proof}
As $\mathcal{V} \neq \bot$, Lemma~\ref{lemma:committed_implies_known} proves that a correct process $p_i$ sets its $\mathit{known\_values}_i[\mathcal{D}]$ variable to value $v^{\star}$ in some view $V \leq \mathcal{V}$.
Since process $p_i$ does so only upon checking the validity of $v^{\star}$ (line~\ref{line:check_validity_received_from_leader}), $v^{\star}$ is indeed valid.
\end{proof}

Next, we prove that all correct processes commit by the (end of the) first view with a correct leader.

\begin{lemma} \label{lemma:first_correct_leader}
Let $V \in [1, t + 1]$ denote the first view such that $\mathsf{leader}(V)$ is correct.
Then, all correct processes commit a digest by the end of view $V$.
\end{lemma}
\begin{proof}
If there exists any correct process that commits in any view smaller than $V$, the lemma follows from Lemma~\ref{lemma:first_commit}.
Hence, let no correct process commit in any view smaller than $V$.
In this case, we distinguish two scenarios:
\begin{compactitem}
    \item The leader decides $(\bot, \cdot)$ from $\mathcal{GC}_1[V]$.
    Hence, the leader broadcasts its value $v$ (Step 2) in view $V$, which implies that all correct processes broadcast a \textsc{support} message for $D$.
    Thus, every correct process $p_k$ updates its $\mathit{vote}_k$ variable to $D$ (Step 3) and proposes $D$ to $\mathcal{GC}_2[D]$ (Step 4).
    Lastly, all correct processes decide $(D, 1)$ from $\mathcal{GC}_2[V]$ (due to its strong unanimity property), which concludes the proof in this case.
    
    \item The leader decides $(D, \cdot)$ from $\mathcal{GC}_1[V]$, for some digest $D \neq \bot$.
        (This implies $V \in [2, t + 1]$.) We further distinguish two cases:
    \begin{compactitem}
        \item There exists a correct process that decides $(D', 1)$ from $\mathcal{GC}_1[V]$. By the consistency property of $\mathcal{GC}_1[V]$, $D = D'$. We now prove that $D$ must be accepted by all correct processes.
    Due to the justification property of $\mathcal{GC}_1[V]$, there exists a correct process $p_j$ that has proposed $D$ to $\mathcal{GC}_1[V]$, which implies that $\mathit{locked}_j = D$ at the beginning of view $V$.
    Let $V' < V$ denote the view in which process $p_j$ sets its $\mathit{locked}_j$ variable to $D$.
    Hence, process $p_j$ has decided $(D, \cdot)$ from $\mathcal{GC}_2[V']$ in view $V'$.
    Therefore, the justification property of $\mathcal{GC}_2[V']$ ensures that some correct process $p_k$ has previously proposed $D$ to $\mathcal{GC}_2[V']$, which implies that $\mathit{vote}_k = D$ in view $V'$.
    Hence, $p_k$ has received a $\langle \textsc{support}, D \rangle$ message from $2t + 1$ processes in view $V'$, out of which at least $t + 1$ processes are correct.
    Thus, every correct process $p_z$ receives (at least) $t + 1$ $\langle \textsc{support}, D \rangle$ messages in view $V'$, which implies that $D \in \mathit{accepted}_z[V']$.

    Due to the consistency property of $\mathcal{GC}_1[V]$, each correct process including the leader decides $(D, \cdot)$ from $\mathcal{GC}_1[V]$.
    Therefore, the leader broadcasts $D$ (Step 2) in view $V$. Moreover, as each correct process decides $(D, \cdot)$ from $\mathcal{GC}_1[V]$ and $D$ is accepted by all correct processes, each correct process broadcasts a \textsc{support} message for $D$ in view $V$ (via line~\ref{line:broadcast_happy_1} or line~\ref{line:broadcast_happy_2}). This further means that all correct processes propose $D$ to $\mathcal{GC}_2[V]$ (Step 4). Finally, all correct processes decide $(D, 1)$ from $\mathcal{GC}_2[V]$ (due to its strong unanimity property), which concludes the proof in this case.

        \item No correct process decides $(\cdot, 1)$ from $\mathcal{GC}_1[V]$. Thus, the leader decides $(D, 0)$.
    The proof is the same as the previous case, but as no correct process decides $(\cdot, 1)$, then no correct process enters line~\ref{line:broadcast_happy_1}. As proven previously, $D$ must be accepted by each correct process. Therefore, each correct process broadcasts a $\textsc{support}$ message for $D$ at line~\ref{line:broadcast_happy_2}. The rest of the proofs follow from the previous case.
    
    \end{compactitem}
\end{compactitem}
The lemma holds, as its statement is satisfied in all possible scenarios.
\end{proof}

We are ready to prove that $\mathcal{V} \neq \bot$. 
Concretely, we prove that $\mathcal{V} \in [1, f + 1]$.

\begin{lemma} \label{lemma:commit_eventually_occurs}
$\mathcal{V} \in [1, f + 1]$.
\end{lemma}
\begin{proof}
Recall that $f \leq t$.
Hence, there exists a view with a correct leader.
Let $V$ be the first such view.
Lemma~\ref{lemma:first_correct_leader} proves that $\mathcal{V} \leq V$.
As $V \in [1, f + 1]$, the lemma holds.
\end{proof}

We are ready to prove that \nameopt ensures that all required assumptions of $\mathcal{F}$ are ensured.

\begin{lemma} \label{lemma:required_assumptions_satisfied}
\nameopt (\Cref{algorithm:optimal_trivial}) ensures that all assumptions required by $\mathcal{DD}$ are satisfied.
\end{lemma}
\begin{proof}
All correct processes input a pair to $\mathcal{DD}$ as all correct processes commit a digest due to lemmas~\ref{lemma:first_commit},~\ref{lemma:first_correct_leader} and~\ref{lemma:commit_eventually_occurs}.
No correct process stops unless it has previously output a value from $\mathcal{DD}$ as a correct process $p_i$ stops (line~\ref{line:halt_optimal_2}) only after it has output a value from $\mathcal{DD}$ (line~\ref{line:upon_output_optimal_2}).
Recall that all correct processes commit the same digest $\mathcal{D}$ (by Lemma~\ref{lemma:first_commit}).
Moreover, the following holds for value $v^{\star}$:
\begin{compactitem}
    \item At least one correct process inputs $(v^{\star} \neq \bot, \mathcal{D})$ to $\mathcal{DD}$.
    As no correct process inputs a pair to $\mathcal{DD}$ before view $\mathcal{V}$ (due to the definition of view $\mathcal{V}$) and there exists at least a single correct process that sets its respective $\mathit{known\_values}_{*}[\mathcal{D}]$ variable to value $v^{\star}$ in view $V \leq \mathcal{V}$ (by Lemma~\ref{lemma:committed_implies_known}), some correct process is guaranteed to input $v^{\star}$ to $\mathcal{DD}$ (line~\ref{line:commit_hash_value}).

    \item If any correct process inputs $(v \neq \bot, \cdot)$, then $v = v^{\star}$ because (1) all correct processes commit the same digest $\mathcal{D}$ such that $\mathcal{D} = \mathsf{digest}(v^{\star})$, 
    (2) all correct processes input a value $v$ with $\mathcal{D} = \mathsf{digest}(v)$, 
    and (3) the $\mathsf{digest}(\cdot)$ function is collision-resistant. 
    The argument above ensures that any value $v \neq \bot$ input to $\mathcal{DD}$ must equal $v^{\star}$.

    \item If any correct process inputs $(\cdot, D)$, then $D = \mathsf{digest}(v^{\star})$ because (1) $D= \mathcal{D}$, and (2) $\mathcal{D} = \mathsf{digest}(v^{\star})$. 
\end{compactitem}
As all assumptions required by $\mathcal{DD}$ are satisfied, the lemma holds.
\end{proof}

The following theorem proves the correctness of \nameopt.

\begin{theorem} [Correctness]
\nameopt (\Cref{algorithm:optimal_trivial}) is correct, i.e., it satisfies agreement, integrity, external validity and termination.
\end{theorem}
\begin{proof}
Integrity follows trivially from the algorithm. Lemma~\ref{lemma:required_assumptions_satisfied} proves that $\mathcal{DD}$ behaves according to its specification.
Therefore, the safety property of $\mathcal{DD}$ ensures that correct processes only output $v^{\star}$ from $\mathcal{DD}$, thus proving the agreement property of \nameopt.
Similarly, the liveness property of $\mathcal{DD}$ ensures that all correct processes eventually output from $\mathcal{DD}$, thus showing the termination property of \nameopt.
Finally, the external validity property of \nameopt follows from Lemma~\ref{lemma:valid_v_star}.
\end{proof}

\smallskip
\noindent \textbf{Proof of complexity.}
First, we show that any non-leader correct process sends $O(n\kappa)$ bits in a view.

\begin{lemma} \label{lemma:non_leader_bits}
Consider any correct process $p_i$ and any view $V$ such that $\mathsf{leader}(V) \neq p_i$.
Then, $p_i$ sends $O(n\kappa)$ bits in $V$.
\end{lemma}
\begin{proof}
Process $p_i$ sends $O(n\kappa)$ bits while executing $\mathcal{GC}_1[V]$ and $\mathcal{GC}_2[V]$.
Moreover, process $p_i$ sends $O(n\kappa)$ bits via $\textsc{support}$ messages.
Hence, $p_i$ sends $O(n\kappa)$ bits in $V$.
\end{proof}

Next, we prove that a correct leader sends $O(nL + n\kappa)$ bits in a view.

\begin{lemma} \label{lemma:leader_bits}
Consider any correct process $p_i$ and any view $V$ such that $\mathsf{leader}(V) = p_i$.
Then, $p_i$ sends $O(nL + n\kappa)$ bits in $V$.
\end{lemma}
\begin{proof}
Process $p_i$ sends $O(n\kappa)$ bits while executing $\mathcal{GC}_1[V]$ and $\mathcal{GC}_2[V]$.
Moreover, process $p_i$ sends $O(n\kappa)$ bits via $\textsc{support}$ messages.
Finally, process $p_i$ sends $O(nL + n\kappa)$ bits in Step 2 of $V$.
Thus, $p_i$ sends $O(n\kappa) + O(nL + n\kappa) = O(nL + n\kappa)$ bits in $V$.
\end{proof}

Next, we prove that no correct process enters any view greater than $\min(t + 1, \mathcal{V} + 2)$.

\begin{lemma} \label{lemma:views_entered}
No correct process enters any view greater than $\min(t + 1, \mathcal{V} + 2)$.
\end{lemma}
\begin{proof}
If $\mathcal{V} + 2 > t + 1$, the lemma trivially holds.
Otherwise, Lemma~\ref{lemma:first_commit} proves that all correct processes commit a digest by (the end of) view $\mathcal{V} + 1$.
As any correct process that commits a digest in some view $V$ only enters the next view $V + 1$, no correct process enters any view greater than $\mathcal{V} + 2$.
\end{proof}

The following lemma proves that $\min(t + 1, \mathcal{V} + 2) \in O(f)$.

\begin{lemma} \label{lemma:min_o_f}
$\min(t + 1, \mathcal{V} + 2) \in O(f)$.    
\end{lemma}
\begin{proof}
By Lemma~\ref{lemma:commit_eventually_occurs}, $\mathcal{V} \in O(f)$. Therefore, $\min(t + 1, \mathcal{V} + 2) \in O(f)$.
\end{proof}

We are ready to prove the bit complexity of \nameopt.

\begin{theorem} [Bit complexity]
\nameopt (\Cref{algorithm:optimal_trivial}) achieves $O(nL + n^2 (f + \log n) \kappa)$ bit complexity.
\end{theorem}
\begin{proof}
Lemma~\ref{lemma:first_correct_leader} proves that, for every view $V \in [1, \mathcal{V} - 1]$, $\mathsf{leader}(V)$ is faulty.
Moreover, Lemma~\ref{lemma:views_entered} shows that no correct process enters any view greater than $\min(t + 1, \mathcal{V} + 2)$.
Lastly, Lemma~\ref{lemma:min_o_f} proves that $\min(t + 1, \mathcal{V} + 2) \in O(f)$.
Hence, in every view $V \in [1, \mathcal{V} - 1]$, each correct process sends $O(n\kappa)$ bits (by Lemma~\ref{lemma:non_leader_bits}).
Moreover, in every view $V \in [\mathcal{V}, \min(t + 1, \mathcal{V} + 2)]$, (1) $\mathsf{leader}(V)$ (if correct) sends $O(nL + n\kappa)$ bits (by Lemma~\ref{lemma:leader_bits}), and (2) every other correct process sends $O(n\kappa)$ bits (by Lemma~\ref{lemma:non_leader_bits}).
Therefore, while executing views $[1, \min(t + 1, \mathcal{V} + 2)]$, correct process send $f \cdot n \cdot O(n\kappa) + O(1) \cdot \big( O(nL + n\kappa) + n \cdot O(n\kappa) \big) = O(n^2 f \kappa) + O(nL + n\kappa) + O(n^2\kappa) = O(nL + n^2 f \kappa)$ bits.
Moreover, correct processes exchange $O(nL + n^2\kappa \log n )$ bits while executing $\mathcal{DD}$. 
Therefore, the bit complexity of \nameopt is $O(nL + n^2 f \kappa) + O(nL + n^2\kappa \log n) = O(nL + n^2(f + \log n)\kappa)$. 
\end{proof}

Finally, we prove that \nameopt achieves early stopping.

\begin{theorem} [Early stopping]
\nameopt (\Cref{algorithm:optimal_trivial}) achieves early stopping.
\end{theorem}
\begin{proof}
By Lemma~\ref{lemma:first_commit}, all correct processes start $\mathcal{DD}$ by the end of view $\mathcal{V} + 1$.
Therefore, Lemma~\ref{lemma:commit_eventually_occurs} proves that all correct processes start executing $\mathcal{DD}$ in $O(f)$ time (recall that each view has $6$ rounds).
The liveness property of $\mathcal{DD}$ ensures that all correct processes output a value in $O(f) + 2\delta \in O(f)$ time, which proves the theorem.
\end{proof}
\section{\name: Missing Implementation \& Proof}\label{section:name_detailed_proof}

In this section, we provide the implementation of \name's building blocks and prove the correctness and complexity of \name.
First, we introduce our implementation of the committee broadcast primitive (see \Cref{mod:committee-broadcast}) and prove its correctness and complexity (\Cref{subsection:committee_broadcast}).
Second, we present our implementation of the finisher primitive (see \Cref{mod:finisher}) and prove its correctness and complexity (\Cref{subsection:finisher_recursive}).
Lastly, we give \name's pseudocode and prove its correctness and complexity (\Cref{subsection:recursive_proof}).

\subsection{Committee Broadcast: Implementation \& Proof} \label{subsection:committee_broadcast}

In this subsection, we give an implementation of the committee broadcast primitive (see \Cref{mod:committee-broadcast}) and prove its correctness and complexity.
We start by introducing the synchronizer primitive (\Cref{subsubsection:synchronizer}), a building block of our implementation of the committee broadcast primitive.
Then, we introduce another building block of our implementation of the committed broadcast primitive -- committee dissemination (\Cref{subsubsection:committee_dissemination}).
Finally, we introduce our implementation of the committee broadcast primitive and prove its correctness and complexity (\Cref{subsubsection:committee_broadcast_implementation_proof}).

\subsubsection{Synchronizer} \label{subsubsection:synchronizer}
The specification of the synchronizer primitive is given in \Cref{mod:synchronizer}.
In brief, the synchronizer primitive ensures that correct processes are ``$2\delta$-synchronized'', i.e., the synchronizer primitive represents a barrier such that if the first correct process crosses the barrier at some time $\tau$, then all correct processes cross the barrier by time $\tau + 2\delta$.
Our implementation of the synchronizer primitive is given in \Cref{algorithm:synchronizer}.
We emphasize that \Cref{algorithm:synchronizer} is highly reminiscent of the ``amplifying'' step of Bracha's reliable broadcast algorithm~\cite{BT85}.

\begin{module}
\caption{Synchronizer}
\label{mod:synchronizer}
\footnotesize
\begin{algorithmic}[1]

\Statex \textbf{Participants:}
\begin{compactitem}
    \item $\mathsf{Entire} \subseteq \Pi$; let $x = |\mathsf{Entire}|$ and let $x'$ be the greatest integer smaller than $x / 3$.
\end{compactitem}

\medskip
\Statex \textbf{Events:}
\begin{compactitem}
    \item \emph{request} $\mathsf{wish\_to\_advance}()$: a process wishes to advance.

    \item \emph{indication} $\mathsf{advance}()$: a process advances.
\end{compactitem}

\medskip 
\Statex \textbf{Assumed behavior:}
\begin{compactitem}
    \item No correct process stops unless it has previously received an $\mathsf{advance}()$ indication.
\end{compactitem}

\medskip 
\Statex \textbf{Properties ensured only if up to $x'$ processes in $\mathsf{Entire}$ are faulty:}
\begin{compactitem}
    \item \emph{Justification:} If a correct process receives an $\mathsf{advance}()$ indication, then at least one correct process has previously invoked a $\mathsf{wish\_to\_advance}()$ request.

    \item \emph{Totality:} Let $\tau$ be the first time at which a correct process receives an $\mathsf{advance}()$ indication. 
    Then, every correct process receives an $\mathsf{advance}()$ indication by time $\tau + 2\delta$.

    \item \emph{Liveness:} Let $\tau$ be the first time by which all correct processes have invoked a $\mathsf{wish\_to\_advance}()$ request.
    Then, every correct process receives an $\mathsf{advance}()$ indication by time $\tau + \delta$.
\end{compactitem}

\medskip 
\Statex \textbf{Properties ensured even if more than $x'$ processes in $\mathsf{Entire}$ are faulty:}
\begin{compactitem}
    \item \emph{Complexity:} Each correct process sends $O(x)$ bits.
\end{compactitem}
\end{algorithmic}
\end{module}

\begin{algorithm} [ht]
\caption{Synchronizer: Pseudocode (for process $p_i$)}
\label{algorithm:synchronizer}
\begin{algorithmic} [1]
\footnotesize
\State \textbf{Participants:}
\State \hskip2em $\mathsf{Entire} \subseteq \Pi$ with $|\mathsf{Entire}| = x$ processes 

\medskip
\State \textbf{upon} $\mathsf{wish\_to\_advance}()$: 
\State \hskip2em \textbf{broadcast} $\langle \textsc{wish\_to\_advance} \rangle$ 

\medskip
\State \textbf{upon} receiving a $\langle \textsc{wish\_to\_advance} \rangle$ message from $x' + 1$ processes:
\label{line:synchronizer_received_tp1_wish_to_advance}

\State \hskip2em \textbf{broadcast} $\langle \textsc{wish\_to\_advance} \rangle$ \label{line:synchronizer_broadcast_wish_to_advance_from_wish_to_advance}

\medskip
\State \textbf{upon} receiving a $\langle \textsc{wish\_to\_advance} \rangle$ message from $2x' + 1$ processes: \label{line:synchronizer_received_2tp1_wish_to_advance}
\State \hskip2em \textbf{broadcast} $\langle \textsc{wish\_to\_advance} \rangle$
\State \hskip2em \textbf{trigger} $\mathsf{advance}()$
\end{algorithmic}
\end{algorithm}

\smallskip
\noindent \textbf{Proof of correctness \& complexity.}
The following theorem proves that \Cref{algorithm:synchronizer} implements the specification of the synchronizer primitive (\Cref{mod:synchronizer}).

\begin{theorem} [\Cref{algorithm:synchronizer} implements synchronizer]
\Cref{algorithm:synchronizer} implements the specification of the synchronizer primitive (\Cref{mod:synchronizer}).
\end{theorem}
\begin{proof}
Recall that the justification, totality and liveness properties need to be satisfied only of up to $x'$ processes in $\mathsf{Entire}$ are faulty.
To prove the theorem, we prove that \Cref{algorithm:synchronizer} satisfies all properties of the synchronizer primitive:
\begin{compactitem}
    \item Justification: 
    Let $p_i$ be any correct process that receives an $\mathsf{advance}()$ indication.
    Hence, $p_i$ has previously received a \textsc{wish\_to\_advance} message from a correct process.
    As the first correct process to broadcast a \textsc{wish\_to\_advance} message does so upon invoking a $\mathsf{wish\_to\_advance}()$ request, the justification property is satisfied.

    \item Totality: Let $p_i$ be the first correct process to receive an $\mathsf{advance}()$ indication; let this happen at some time $\tau$.
    This implies that $p_i$ has received a \textsc{wish\_to\_advance} message from $2x' + 1$ processes by time $\tau$, out of which at least $x' + 1$ are correct.
    Therefore, by time $\tau + \delta$, every correct process receives a \textsc{wish\_to\_advance} message from $x' + 1$ processes and broadcasts its own \textsc{wish\_to\_advance} messages.
    Thus, every correct process receives a \textsc{wish\_to\_advance} message from $x - x' \geq 2x' + 1$ processes by time $\tau + 2\delta$, which concludes the proof of the totality property.

    \item Liveness: By time $\tau + \delta$, every correct process receives at least $x - x' \geq 2x' + 1$ \textsc{wish\_to\_advance} messages, which proves the liveness property.

    \item Complexity: Every correct process sends only $O(x)$ \textsc{wish\_to\_advance} messages, each with $O(1)$ bits.
    Thus, every correct process sends $O(x)$ bits.
\end{compactitem}
As all properties of the synchronizer primitive are ensured by \Cref{algorithm:synchronizer}, the theorem holds.
\end{proof}

\subsubsection{Committee Dissemination} \label{subsubsection:committee_dissemination}

The specification of the committee dissemination primitive is given in~\Cref{mod:committee_dissemination_alterative}.
In brief, the committee dissemination primitive is concerned with two sets of processes: (1) $\mathsf{Committee}$, the set of processes that (might) hold the same input value, and (2) $\mathsf{Entire} \supseteq \mathsf{Committee}$, the set of processes that output the input value of the processes in $\mathsf{Committee}$.
Our implementation of the committee dissemination primitive is presented in \Cref{algorithm:delta_sync_balanced_committee_dissemination_alternative}.
We underline that the fact that the committee dissemination primitive terminates in $7 \delta$ time (see the liveness property) dictates that the $\beta$ constant in \name (see \Cref{algorithm:recursive_early_stopping_king_2}) is equal to $7$.

\begin{module}[ht]
\caption{Committee dissemination $\langle \mathsf{Entire}, \mathsf{Committee} \rangle$}
\label{mod:committee_dissemination_alterative}
\footnotesize
\begin{algorithmic}[1]

\Statex \textbf{Participants:} \BlueComment{usually, $|\mathsf{Committee}| \approx |\mathsf{Entire}|/2$}
\begin{compactitem}
    \item $\mathsf{Entire} \subseteq \Pi$; let $x = |\mathsf{Entire}|$ and let $x'$ be the greatest integer smaller than $x / 3$.
    \item $\mathsf{Committee} \subseteq \mathsf{Entire}$; let $y = |\mathsf{Committee}|$ and let $y'$ be the greatest integer smaller than $y / 3$.
\end{compactitem}

\medskip
\Statex \textbf{Events:}
\begin{compactitem}
    \item \emph{request} $\mathsf{disseminate}(v \in \mathsf{Value})$: a process disseminates a value $v$. \BlueComment{invoked only by $\mathsf{Committee}$}
    \item \emph{indication} $\mathsf{obtain}(v' \in \mathsf{Value})$: a process obtains a value $v'$. \BlueComment{received by $\mathsf{Entire}$}
\end{compactitem}

\medskip 
\Statex \textbf{Assumed behavior:}
\begin{compactitem}
    \item A correct process $p_i$ invokes a $\mathsf{disseminate}(\cdot)$ request if and only if $p_i \in \mathsf{Committee}$.

    \item No correct process stops unless it has previously obtained a value.
\end{compactitem}

\medskip 
\Statex \textbf{Properties ensured only if up to $x'$ processes in $\mathsf{Entire}$ are faulty:}
\begin{compactitem}
    \item \emph{Totality:} If a correct process obtains a value at time $\tau$, then every correct process obtains a value by time $\tau+2\delta$.

    \item \emph{Optimistic safety \& liveness:} Let all correct processes in $\mathsf{Committee}$ disseminate the same value $v^{\star}$ (once each).
    Moreover, let all correct processes do so within $2\delta$ time of each other.
    Let there be up to $y'$ faulty processes in $\mathsf{Committee}$.
    Then, the following properties are ensured:
    \begin{compactitem}
    
    \item \emph{Safety:} If a correct process obtains a value $v$, then $v = v^{\star}$.
    
    \item \emph{Liveness:} Let $\tau$ be the first time by which all correct processes in $\mathsf{Committee}$ have disseminated $v^{\star}$.
    Then, all correct processes (in $\mathsf{Entire}$) obtain a value by time $\tau + 7\delta$. 
    \end{compactitem}
\end{compactitem}

\medskip
\Statex \textbf{Properties ensured even if more than $x'$ processes in $\mathsf{Entire}$ are faulty:}
\begin{compactitem}
    \item \emph{Complexity:} Each correct process sends $O(L + x\log x)$ bits.
\end{compactitem}
\end{algorithmic}
\end{module}

\begin{algorithm} [ht]
\caption{Committee dissemination $\langle \mathsf{Entire}, \mathsf{Committee} \rangle$: Pseudocode (for process $p_i$)}
\label{algorithm:delta_sync_balanced_committee_dissemination_alternative}
\begin{algorithmic} [1]
\footnotesize

\State \textbf{Participants:}
\State \hskip2em $\mathsf{Entire} \subseteq \Pi$ with $|\mathsf{Entire}| = x$ processes
\State \hskip2em $\mathsf{Committee} \subseteq \mathsf{Entire}$ with $|\mathsf{Committee}| = y$ processes

\medskip
\State \textbf{Uses:}
\State \hskip2em Synchronizer among $\mathsf{Entire}$, \textbf{instance} $\mathcal{S}$ \BlueComment{see \Cref{mod:synchronizer}}

\medskip 
\State \textbf{Constants:}

\State \hskip2em $\mathsf{Integer}$ $\alpha = 4$

\medskip
\State \textbf{upon} $\mathsf{disseminate}(v)$: \BlueComment{executed only if $p_i \in \mathsf{Committee}$}
\State \hskip2em Let $[m_1, m_2, ..., m_{x}] \gets \mathsf{RSEnc}(v, x, x' + 1)$ 
\State \hskip2em \textbf{for each} process $p_j \in \mathsf{Entire}$:
\State \hskip4em \textbf{send} $\langle \textsc{disperse}, m_j \rangle$ to process $p_j$

\medskip
\State \textbf{upon} receiving a $\langle \textsc{disperse}, m_i \rangle$ message from $y' + 1$ processes, for some RS symbol $m_i$: \label{line:ds_bcda_received_tpp1_rs_symbols_from_committee}
\State \hskip2em \textbf{broadcast} $\langle \textsc{reconstruct}, m_i \rangle$ (to $\mathsf{Entire}$) \label{line:ds_bcda_broadcast_reconstruct}

\medskip
\State \textbf{upon} receiving $2x'+1$ $\textsc{reconstruct}$ messages: \label{line:ds_bcda_received_2tp1_reconstruct}
\State \hskip2em \textbf{invoke} $\mathcal{S}.\mathsf{wish\_to\_advance}()$

\medskip
\State \textbf{upon} $\mathcal{S}.\mathsf{advance}()$:
\State \hskip2em \textbf{wait} $\alpha \delta$ time \label{line:ds_bcda_final_waiting_step}
\State \hskip2em For every process $p_j \in \mathsf{Entire}$, let $m_j$ denote the RS symbol sent by process $p_j$ via a \textsc{reconstruct} message  \hphantom{aallll}(if such a message is not received, $m_j = \bot$). 
 \label{line:ds_bcda_received_rs_symbols}

\State \hskip2em \textbf{trigger} $\mathsf{obtain}(\mathsf{RSDec}(x'+1, x', \{(j,m_j) \,|\, p_j \in \mathsf{Entire}\})$ \label{line:ds_bcda_obtain_value} 

\end{algorithmic}
\end{algorithm}

\smallskip
\noindent \textbf{Proof of correctness.}
We start by proving the totality property.

\begin{theorem} [Totality]
\Cref{algorithm:delta_sync_balanced_committee_dissemination_alternative} satisfies totality.
\end{theorem}
\begin{proof}
Recall that the totality property needs to be satisfied only if up to $x'$ processes in $\mathsf{Entire}$ are faulty.
Let $p_i$ be a correct process that obtains a value at some time $\tau$.
Hence, $p_i$ receives an $\mathsf{advance}()$ indication from the synchronizer instance $\mathcal{S}$ at time $\tau' = \tau - 4\delta$, due to the waiting step at line~\ref{line:ds_bcda_final_waiting_step}.
The totality property of $\mathcal{S}$ guarantees that all correct processes receive an $\mathsf{advance}()$ indication by time $\tau' + 2\delta$.
Therefore, every correct process obtains a value by time $\tau' + 6\delta = \tau + 2\delta$, which means the totality property is ensured.
\end{proof}

We next prove the safety property.

\begin{theorem} [Safety]
\Cref{algorithm:delta_sync_balanced_committee_dissemination_alternative} satisfies safety.
\end{theorem}
\begin{proof}
Recall that the safety property needs to be ensured only if:
\begin{compactitem}
    \item up to $x'$ processes in $\mathsf{Entire}$ are faulty, and

    \item all correct processes in $\mathsf{Committee}$ disseminate the same value $v^{\star}$ and they do so within $2\delta$ time of each other, and

    \item up to $y'$ processes in $\mathsf{Committee}$ are faulty.
\end{compactitem}
Let $\tau$ be the time at which the first correct process invokes a $\mathsf{disseminate}(v^{\star})$ request.
By time $\tau + 2\delta$, all correct processes invoke a $\mathsf{disseminate}(v^{\star})$ request.
Hence, by time $\tau + 3\delta$, each correct process $p_j \in \mathsf{Entire}$ sends the correctly-encoded RS symbol $m_j$ of value $v^{\star}$.
Therefore, by time $\tau + 4\delta$, every correct process receives \textsc{reconstruct} messages from all correct processes.

Consider any correct process $p_i$ that triggers the $\mathsf{obtain}(\cdot)$ indication at line~\ref{line:ds_bcda_obtain_value} at some time $\tau_i$. By the waiting step at line~\ref{line:ds_bcda_final_waiting_step}, process $p_i$ received $\mathcal{S}.\mathsf{advance}()$ at time $\tau-4\delta$. By $\mathcal{S}$'s justification, some process $p_j$ triggered $\mathcal{S}.\mathsf{wish\_to\_advance}()$ at some time $\tau_j \leq \tau_i - 4 \delta$, after having received $2t+1$ \textsc{reconstruct} messages. Thus, process $p_j$ received a \textsc{reconstruct} message from a correct process at time $\tau_j$, which means $\tau \leq \tau_j$. Hence, $\tau_i \geq \tau + 4 \delta$, which guarantees  $p_i$ receives all correctly-encoded RS symbols corresponding to correct processes before triggering $\mathsf{obtain}(\cdot)$.
As all these RS symbols are for value $v^{\star}$, $p_i$ obtains $v^{\star}$, which concludes the proof.
\end{proof}

The following theorem proves the liveness property.

\begin{theorem} [Liveness]
\Cref{algorithm:delta_sync_balanced_committee_dissemination_alternative} satisfies liveness.
\end{theorem}
\begin{proof}
Recall that the liveness property needs to be ensured only if:
\begin{compactitem}
    \item up to $x'$ processes in $\mathsf{Entire}$ are faulty, and

    \item all correct processes in $\mathsf{Committee}$ disseminate the same value $v^{\star}$ and they do so within $2\delta$ time of each other, and

    \item up to $y'$ processes in $\mathsf{Committee}$ are faulty.
\end{compactitem}
Let $\tau$ be the first time by which all correct processes in $\mathsf{Committee}$ invoke a $\mathsf{disseminate}(v^{\star})$ request.
By time $\tau + 2\delta$, every correct process $p_i \in \mathsf{Entire}$ receives $2x' + 1$ \textsc{reconstruct} messages, and invokes a $\mathsf{wish\_to\_advance}()$ request on $\mathcal{S}$. The liveness property of $\mathcal{S}$ ensures that every correct process eventually receives an $\mathsf{advance}()$ indication by time $\tau + 3\delta$, waits $4\delta$ time, and triggers $\mathsf{obtain}(\cdot)$ at line~\ref{line:ds_bcda_obtain_value} by time $\tau + 7\delta$, which concludes the proof.
\end{proof}

\smallskip
\noindent \textbf{Proof of complexity.}
Finally, we prove that each correct process sends $O(L + x \log x)$ bits.

\begin{theorem} [Complexity]
Each correct process sends $O(L + x \log x)$ bits in \Cref{algorithm:delta_sync_balanced_committee_dissemination_alternative} (even if more than $x'$ processes in $\mathsf{Entire}$ are faulty).
\end{theorem}
\begin{proof}
Each RS symbols is of size $O(\frac{L}{x' + 1} + \log x) = O(\frac{L}{x} + \log x)$.
As each correct process broadcasts at most one \textsc{disperse} and one \textsc{reconstruct} message, each with $O(\frac{L}{x} + \log x)$ bits, each correct process sends $O(L + x \log x)$ bits via \textsc{disperse} and \textsc{reconstruct} messages.
Moreover, each correct process sends $O(x)$ bits in $\mathcal{S}$.
Therefore, each correct process sends $O(L + x \log x)$ bits.
\end{proof}

\subsubsection{Committee Broadcast: Implementation \& Proof} \label{subsubsection:committee_broadcast_implementation_proof}

Our implementation is presented in \Cref{algorithm:committee_broadcast}.

\begin{algorithm} [ht]
\caption{Committee broadcast $\langle \mathsf{Entire}, \mathsf{Committee}, \mathcal{VA} \rangle$: Pseudocode (for process $p_i$)}
\label{algorithm:committee_broadcast}
\begin{algorithmic} [1]
\footnotesize

\State \textbf{Participants:}
\State \hskip2em $\mathsf{Entire} \subseteq \Pi$ with $|\mathsf{Entire}| = x$ processes
\State \hskip2em $\mathsf{Committee} \subseteq \mathsf{Entire}$ with $|\mathsf{Committee}| = y$ processes

\smallskip
\State \textbf{Uses:}
\State \hskip2em Validated agreement $\mathcal{VA}$ among $\mathsf{Committee}$, \textbf{instance} $\mathcal{VA}$
\State \hskip2em Committee dissemination $\langle \mathsf{Entire}, \mathsf{Committee} \rangle$, \textbf{instance} $\mathcal{CD}$ \BlueComment{see \Cref{mod:committee_dissemination_alterative}}

\medskip
\State \textbf{Crucial property provided by $\mathcal{VA}$:}
\State \hskip2em If (1) there are up to $y'$ faulty processes in $\mathsf{Committee}$, and (2) all correct processes in $\mathsf{Committee}$ start \hphantom{aallll}executing $\mathcal{VA}$ within $2\delta$ time of each other, then all correct processes in $\mathsf{Committee}$ decide from $\mathcal{VA}$ within \hphantom{aallll}$2\delta$ time of each other. \label{line:elector_broadcast_BA_assumption}

\medskip
\State \textbf{Local variables:}
\State \hskip2em $\mathsf{Value} \times \{0, 1\}$ $(v_i, g_i) \gets \bot$ \BlueComment{$p_i$'s input}

\medskip
\State \textbf{upon} $\mathsf{input}(\mathsf{Value} \text{ } v, \{0, 1\} \ni g)$:
\State \hskip2em $(v_i, g_i) \gets (v, g)$
\State \hskip2em \textbf{invoke} $\mathcal{VA}.\mathsf{propose}(v_i)$  \BlueComment{only executed if $p_i \in \mathsf{Committee}$}\label{line:elector_broadcast_BA_propose}

\medskip
\State \textbf{upon} $\mathcal{VA}.\mathsf{decide}(v^d)$: \label{line:elector_broadcast_BA_decide}
\State \hskip2em \textbf{invoke} $\mathcal{CD}.\mathsf{disseminate}(v^d)$ \label{line:elector_broadcast_CD_disseminate}

\medskip
\State \textbf{upon} $\mathcal{CD}.\mathsf{obtain}(v^o)$: \label{line:elector_broadcast_CD_obtain}
\State \hskip2em \textbf{if} $g_i = 0$ and $\mathsf{valid}(v^o)=\mathit{true}$: \label{line:elector_broadcast_check_grade_before_returning}
\State \hskip4em \textbf{trigger} $\mathsf{output}(v^o)$ \label{line:elector_broadcast_return_obtained}
\State \hskip2em \textbf{else:}
\State \hskip4em \textbf{trigger} $\mathsf{output}(v_i)$ \label{line:elector_broadcast_return_initial}

\medskip
\State \textbf{upon} $\mathcal{B}_{\mathcal{VA}}$ bits have been sent via $\mathcal{VA}$:  
\State \hskip2em \textbf{stop executing} $\mathcal{VA}$ \BlueComment{$\mathsf{Committee}$ is overly corrupted}

\end{algorithmic}
\end{algorithm}

\smallskip
\noindent \textbf{Proof of correctness.}
Recall that the totality, stability, external validity, liveness, agreement, and \strongVal properties must be ensured only if there are up to $x' < x / 3$ faulty processes in $\mathsf{Entire}$.
We start by proving the totality property.

\begin{theorem} [Totality]
\Cref{algorithm:committee_broadcast} satisfies totality.
\end{theorem}
\begin{proof}
Let $\tau$ be the first time at which a correct process outputs a value (line \ref{line:elector_broadcast_return_obtained} or \ref{line:elector_broadcast_return_initial}); let that process be denoted by $p_i$.
Therefore, process $p_i$ obtains a value from the committee dissemination instance $\mathcal{CD}$ at time $\tau$ (line \ref{line:elector_broadcast_CD_obtain}).
The totality property of $\mathcal{CD}$ ensures that every correct process obtains a value from $\mathcal{CD}$ by time $\tau + 2\delta$, which proves the totality property of \Cref{algorithm:committee_broadcast}.
\end{proof}

Next, we prove the stability property.

\begin{theorem} [Stability]
\Cref{algorithm:committee_broadcast} satisfies stability.
\end{theorem}
\begin{proof}
The stability property is ensured due to the check at line~\ref{line:elector_broadcast_check_grade_before_returning}.
\end{proof}

The following theorem proves the external validity property.

\begin{theorem} [External validity]
\Cref{algorithm:committee_broadcast} satisfies external validity.
\end{theorem}
\begin{proof}
Let $p_i$ be a correct process that outputs some value $v$ (line \ref{line:elector_broadcast_return_obtained} or \ref{line:elector_broadcast_return_initial}).
If $p_i$ outputs $v$ at line \ref{line:elector_broadcast_return_obtained}, then $v$ is valid due to the check line~\ref{line:elector_broadcast_check_grade_before_returning}. 
Otherwise, $v$ is valid due to the assumption that any correct process inputs a valid value.
\end{proof}

Recall that the liveness, agreement, and \strongVal properties need to be satisfied only if (1) there are up to $y'$ faulty processes in $\mathsf{Committee}$, and (2) all correct processes in $\mathsf{Entire}$ start executing the committee broadcast primitive within $2\delta$ time of each other.
First, we prove liveness.

\begin{theorem} [Liveness]
\Cref{algorithm:committee_broadcast} satisfies liveness.
\end{theorem}
\begin{proof}
Suppose (1) there are up to $y'$ faulty processes in $\mathsf{Committee}$, and (2) all correct processes in $\mathsf{Entire}$ start executing \Cref{algorithm:committee_broadcast} within $2\delta$ time of each other.
In this case, $\mathcal{VA}$ satisfies the properties of validated agreement: all correct processes decide the same valid value $v$ by time $\tau + \mathcal{R}_{\mathcal{VA}}$.
Hence, all correct processes disseminate value $v$ by time $\tau + \mathcal{R}_{\mathcal{VA}}$ and they do so within $2\delta$ time of each other.
The optimistic safety and liveness properties of $\mathcal{CD}$ ensure that all correct processes in $\mathsf{Entire}$ obtain the value $v$ by time $\tau + \mathcal{R}_{\mathcal{VA}} + 7\delta$, which proves the liveness property.
\end{proof}

Next, we prove the agreement property.

\begin{theorem} [Agreement]
\Cref{algorithm:committee_broadcast} satisfies agreement.
\end{theorem}
\begin{proof}
Suppose (1) there are up to $y'$ faulty processes in $\mathsf{Committee}$, and (2) all correct processes in $\mathsf{Entire}$ start executing \Cref{algorithm:committee_broadcast} within $2\delta$ time of each other.
Therefore, $\mathcal{VA}$ satisfies the properties of validated agreement: all correct processes in $\mathsf{Committee}$ decide the same valid value $v$.
We now distinguish two scenarios:
\begin{compactitem}
    \item Let there exist a correct process $p_i \in \mathsf{Entire}$ such that $p_i$ inputs $(v', 1)$ to \Cref{algorithm:committee_broadcast}, for some value $v'$.
    In this case, due to the assumptions on the committee broadcast primitive, every correct process $p_j \in \mathsf{Entire}$ inputs $(v', \cdot)$ to \Cref{algorithm:committee_broadcast}.
    As $\mathcal{VA}$ satisfies \strongVal, $v' = v$.
    Hence, all correct processes output the same value in this case.

    \item Let no correct process $p_i \in \mathsf{Entire}$ input $(\cdot, 1)$ to \Cref{algorithm:committee_broadcast}.
    In this case, all correct processes in $\mathsf{Entire}$ output $v$.
\end{compactitem}
The agreement property is satisfied as it holds in both possible scenarios.
\end{proof}

Finally, we prove the \strongVal property.

\begin{theorem} [Strong unanimity]
\Cref{algorithm:committee_broadcast} satisfies \strongVal.
\end{theorem}
\begin{proof}
Suppose (1) there are up to $y'$ faulty processes in $\mathsf{Committee}$, and (2) all correct processes in $\mathsf{Entire}$ start executing \Cref{algorithm:committee_broadcast} within $2\delta$ time of each other.
Let all correct processes input a pair $(v, \cdot)$, for some value $v$. If a correct process inputs $(v, 1)$, then it will output $v$. 
Otherwise, As $\mathcal{VA}$ satisfies \strongVal, all correct processes in $\mathsf{Committee}$ decide $v$ from $\mathcal{VA}$ and disseminate $v$.
Therefore, the safety property of $\mathcal{CD}$ ensures that no correct process outputs a value different from $v$.
\end{proof}

\smallskip 
\noindent \textbf{Proof of complexity.}
Lastly, we prove that each correct process sends $O(L + x \log x) + \mathcal{B}_{\mathcal{VA}}$ bits.

\begin{theorem} [Complexity]
Each correct process sends $O(L + x \log x) + \mathcal{B}_{\mathcal{VA}}$ bits in \Cref{algorithm:committee_broadcast} (even if more than $x'$ processes in $\mathsf{Entire}$ are faulty).
\end{theorem}
\begin{proof}
Consider any correct process $p_i$.
Process $p_i$ sends at most $\mathcal{B}_{\mathcal{VA}}$ bits in $\mathcal{VA}$ (if $p_i \in \mathsf{Committee}$).
Moreover, process $p_i$ sends $O(L + x \log x)$ bits in $\mathcal{CD}$, which concludes the proof.
\end{proof}

\subsection{Finisher: Implementation \& Proof} \label{subsection:finisher_recursive}

Our implementation of the finisher primitive (\Cref{mod:finisher}) is presented in \Cref{algorithm:finisher_trivial}.

\begin{algorithm} [ht]
\caption{Finisher: Pseudocode (for process $p_i$)}
\label{algorithm:finisher_trivial}
\begin{algorithmic} [1]
\footnotesize

\State \textbf{Participants:}
\State \hskip2em $\mathsf{Entire} \subseteq \Pi$ with $|\mathsf{Entire}| = x$ processes

\medskip
\State \textbf{Uses:}
\State \hskip2em Synchronizer among $\mathsf{Entire}$, \textbf{instance} $\mathcal{S}$ \BlueComment{see \Cref{mod:synchronizer}} 

\medskip
\State \textbf{Local variables:}
\State \hskip2em $\mathsf{Value}$ $v_i \gets \bot$ \BlueComment{$p_i$'s input}

\medskip
\State \textbf{upon} $\mathsf{input}(\mathsf{Value} \text{ } v, \{0, 1\} \ni g)$:
\State \hskip2em $v_i \gets v$
\State \hskip2em \textbf{if} $g=1$: \label{line:finisher_check}
\State \hskip4em \textbf{invoke} $\mathcal{S}.\mathsf{wish\_to\_advance}()$  \label{line:finisher_sync_wish}

\medskip
\State \textbf{upon} $\mathcal{S}.\mathsf{advance}()$: \label{line:finisher_sync_advance}
\State \hskip2em \textbf{trigger} $\mathsf{output}(v_i)$ \BlueComment{wait for the input if necessary} \label{line:finisher_output}

\end{algorithmic}
\end{algorithm}

\smallskip
\noindent \textbf{Proof of correctness \& complexity.}
The following theorem proves the correctness and complexity of \Cref{algorithm:finisher_trivial}.

\begin{theorem} [\Cref{algorithm:finisher_trivial} implements finisher]
\Cref{algorithm:finisher_trivial} implements the specification of the finisher primitive (\Cref{mod:finisher}).
\end{theorem}
\begin{proof}
We prove all the properties of the finisher primitive:
\begin{compactitem}
    \item Preservation: A correct process necessarily outputs the value that it has previously input.

    \item Justification: Assume a correct process outputs a value (line~\ref{line:finisher_output}). 
    Due to the justification property of the synchronizer instance $\mathcal{S}$, a correct process $p_j$ has invoked a $\mathcal{S}.\mathsf{wish\_to\_advance}()$ request (line~\ref{line:finisher_sync_wish}), which implies that $p_j$ has input $(\cdot, 1)$ due to the check at line~\ref{line:finisher_check}.

    \item Agreement: This property is ensured due to the justification property and the assumption that if a correct process proposes $(v,1)$, then every correct process proposes $(v, \cdot)$.

    \item Totality: Holds due to $\mathcal{S}$'s totality property and the assumption that all processes input within $2\delta$ of each other.
    
    \item Liveness: Holds due to $\mathcal{S}$'s liveness property.
    
    \item Complexity: Holds due to $\mathcal{S}$'s complexity property.
\end{compactitem}
As all properties of the finisher primitive are satisfied, the theorem holds.
\end{proof}

\subsection{Proof of Correctness \& Complexity} \label{subsection:recursive_proof}

The implementation of \name is given in \Cref{algorithm:recursive_early_stopping_king_2}.

\setlength{\lineskip}{0pt}
\begin{algorithm} [hp]
\caption{\name: Pseudocode (for process $p_i$)}
\label{algorithm:recursive_early_stopping_king_2}
\begin{algorithmic} [1]
\fontsize{7.8pt}{7pt}\selectfont
\State \textbf{Participants:} \BlueComment{due to the recursive nature of \name, we explicitly state its set of participants}
\State \hskip2em $\mathsf{Entire} \subseteq \Pi$ with $|\mathsf{Entire}| = x$ processes 

\medskip
\State \textbf{Constants:}
\State \hskip2em $\mathsf{Set}(\mathsf{Process})$ $\mathcal{H}_1 = \{p_1, p_2, ..., p_{\lceil \frac{x}{2} \rceil}\} \setminus{\{p_\ell\}}$ \BlueComment{$p_\ell$ is a default leader, e.g., $p_\ell = p_1$}
\State \hskip2em $\mathsf{Set}(\mathsf{Process})$ $\mathcal{H}_2 = \{p_{\lceil \frac{x}{2} \rceil + 1}, ..., p_x\} \setminus{\{p_\ell\}}$
\State \hskip2em $\mathsf{Integer}$ $\beta = 7$ \BlueComment{the constant $\beta$ is explained in \Cref{subsubsection:committee_dissemination}}
\State \hskip2em \textcolor{blue}{ \(\triangleright\) the concrete values of $R_1$ and $R_2$ can be computed using the recursive equation presented in \Cref{section:name_detailed_proof}}
\State \hskip2em $\mathsf{Integer}$ $R_1 = \text{the worst-case latency complexity of \name among $\mathcal{H}_1$ with up to $|\mathcal{H}_1| / 3$ failures}$ 
\State \hskip2em $\mathsf{Integer}$ $R_2 = \text{the worst-case latency complexity of \name among $\mathcal{H}_2$ with up to $|\mathcal{H}_2| / 3$ failures}$ 

\State \hskip2em $\mathsf{Time}$ $\mathit{timeout}_1 = (R_1 + 2 + \beta)\delta$
\State \hskip2em $\mathsf{Time}$ $\mathit{timeout}_2 = (R_2 + 2 + \beta)\delta$

\medskip
\State \textbf{Uses:}
\State \hskip2em Graded consensus among $\mathsf{Entire}$, \textbf{instances} $\mathcal{GC}_{\mathit{su}}$, $\mathcal{GC}_{\ell}$, $\mathcal{GC}_1$, $\mathcal{GC}_2$ \BlueComment{bits: $O(nL + n^2 \log n)$, rounds: $8$}
\State \hskip2em Committee broadcast $\langle \mathsf{Entire}, \{p_\ell\}, \name \rangle$, \textbf{instance} $\mathcal{CB}_{\ell}$
\State \hskip2em Committee broadcast $\langle \mathsf{Entire}, \mathcal{H}_1, \name \rangle$, \textbf{instance} $\mathcal{CB}_1$
\State \hskip2em Committee broadcast $\langle \mathsf{Entire}, \mathcal{H}_2, \name \rangle$, \textbf{instance} $\mathcal{CB}_2$
\State \hskip2em Finisher, \textbf{instances} $\mathcal{F}_{\ell}$, $\mathcal{F}_2$

\medskip
\State \textbf{Local variables:}
\State \hskip2em $\mathsf{Value}$ $v_i \gets \bot$ \BlueComment{$p_i$'s proposal}
\State \hskip2em $\mathsf{Timer}$ $\mathit{timer}_0$, $\mathit{timer}_1$, $\mathit{timer}_2$

\medskip
\State \textbf{upon} $\mathsf{propose}(\mathsf{Value} \text{ } v)$: \label{line:rec_es2_propose}
\State \hskip2em $v_i \gets v$
\State \hskip2em \textbf{invoke} $\mathcal{GC}_{\mathit{su}}.\mathsf{propose}(v_i)$ \label{line:rec_es2_GCsv_propose}

\medskip
\State \textbf{upon} $\mathcal{GC}_{\mathit{su}}.\mathsf{decide}(\mathsf{Value} \text{ } v_0, \{0, 1\} \ni g_0)$: \label{line:rec_es2_GCsv_decide}
\State \hskip2em $v_i \gets v_0$ \label{line:rec_es2_GCsv_update_opinion}
\State \hskip2em $\mathit{timer}_0.\mathsf{measure}\big( (2+\beta)\delta \big)$ \label{line:rec_es2_GCsv_timer}
\State \hskip2em \textbf{invoke} $\mathcal{CB}_{\ell}.\mathsf{input}(v_0, g_0)$ \label{line:rec_es2_EBl_propose}

\medskip
\State \textbf{upon} $\mathcal{CB}_{\ell}.\mathsf{output}(\mathsf{Value} \text{ } v')$: \label{line:rec_es2_EBl_output}
\State \hskip2em \textbf{invoke} $\mathcal{GC}_\ell.\mathsf{propose}(v')$ \label{line:rec_es2_GCl_propose_after_EBL}

\medskip
\State \textbf{upon} $\mathit{timer}_0$ expires: \label{line:rec_es2_timer0_expires}
\State \hskip2em \textbf{invoke} $\mathcal{GC}_\ell.\mathsf{propose}(v_i)$ \label{line:rec_es2_GCl_propose_after_timer}

\medskip
\State \textbf{upon} $\mathcal{GC}_\ell.\mathsf{decide}(\mathsf{Value} \text{ } v_\ell, \{0, 1\} \ni g_\ell)$: \label{line:rec_es2_GCl_decide}
\State \hskip2em  $v_i \gets v_\ell$ \label{line:rec_es2_GCl_update_opinion}
\State \hskip2em \textbf{invoke} $\mathcal{F}_{\ell}.\mathsf{input}(v_\ell, g_\ell)$ \label{line:rec_es2_fin1_to_finish}
\State \hskip2em \textbf{wait for} $3\delta$  \label{line:rec_es2_wait_1} \BlueComment{if $p_\ell$ is correct and $f<n/3$, decision is made at line~\ref{line:recursive_early_stopping2_decide_after_first_finisher}, without additional step}
\State \hskip2em \textbf{invoke} $\mathcal{GC}_1.\mathsf{propose}(v_\ell)$ \label{line:rec_es2_GC1_propose}

\medskip
\State \textbf{upon} $\mathcal{F}_{\ell}.\mathsf{output}(\mathsf{Value} \text{ } v)$: \label{line:rec_es2_fin1_finish}
\State \hskip2em \textbf{trigger} $\mathsf{decide}(v)$ \label{line:recursive_early_stopping2_decide_after_first_finisher}\BlueComment{will be triggered after $O(\delta)$ time if $p_\ell$ is correct and $f<n/3$}

\State \hskip2em  \textbf{wait for} $2\delta$, and \textbf{stop executing} \name \label{line:rec_es2_halt1}

\medskip
\State \textbf{upon} $\mathcal{GC}_1.\mathsf{decide}(\mathsf{Value} \text{ } v_1, \{0, 1\} \ni g_1)$: \label{line:rec_es2_GC1_decide}
\State \hskip2em  $v_i \gets v_1$: \label{line:rec_es2_GC1_update_opinion}

\State \hskip2em \textbf{invoke} $\mathit{timer}_1.\mathsf{measure}(\mathit{timeout}_1)$ \label{line:rec_es2_timeout1_measure}
\State \hskip2em \textbf{invoke} $\mathcal{CB}_{1}.\mathsf{input}(v_1, g_1)$ \label{line:rec_es2_EB1_propose}

\medskip

\State \textbf{upon} $\mathcal{CB}_{1}.\mathsf{output}(\mathsf{Value} \text{ } v')$: \label{line:rec_es2_EB1_output}
\State \hskip2em \textbf{invoke} $\mathcal{GC}_2.\mathsf{propose}(v')$ \label{line:rec_es2_GC2_propose_after_EB}

\medskip
\State \textbf{upon} $\mathit{timer}_1$ expires: \label{line:rec_es2_timer1_expires}
\State \hskip2em \textbf{invoke} $\mathcal{GC}_2.\mathsf{propose}(v_i)$ \label{line:rec_es2_GC2_propose_after_timer}

\medskip
\State \textbf{upon} $\mathcal{GC}_2.\mathsf{decide}(\mathsf{Value} \text{ } v_2, \{0, 1\} \ni g_2)$: \label{line:rec_es2_GC2_decide}
\State \hskip2em  $v_i \gets v_2$ \label{line:rec_es2_GC2_update_opinion}
\State \hskip2em \textbf{invoke} $\mathcal{F}_2.\mathsf{input}(v_2, g_2)$ \label{line:rec_es2_fin2_to_finish}
\State \hskip2em \textbf{wait for} $3\delta$  \label{line:rec_es2_wait2} \BlueComment{if $\mathcal{H}_1$ is healthy and $f<n/3$, decision is made at line~\ref{line:recursive_early_stopping2_decide_after_second_finisher}, without additional step}
\State \hskip2em \textbf{invoke} $\mathit{timer}_2.\mathsf{measure}(\mathit{timeout}_2)$ \label{line:rec_es2_timeout2_measure}
\State \hskip2em \textbf{invoke} $\mathcal{CB}_{2}.\mathsf{input}(v_2, g_2)$ \label{line:rec_es2_EB2_input}

\medskip
\State \textbf{upon} $\mathcal{F}_2.\mathsf{output}(v)$: \label{line:rec_es2_fin2_finish}
\State \hskip2em \textbf{trigger} $\mathsf{decide}(v)$ \label{line:recursive_early_stopping2_decide_after_second_finisher}\BlueComment{will be triggered after $O(f\delta )$ time if the first half is healthy, but not $p_1$}
\State \hskip2em  \textbf{wait for} $2\delta$, and \textbf{stop executing} \name \label{line:rec_es2_halt2}

\medskip
\State \textbf{upon} $\mathcal{CB}_{2}.\mathsf{output}(v')$: \label{line:rec_es2_EB2_output}
\State \hskip2em \textbf{trigger} $\mathsf{decide}(v')$, \textbf{wait for} $2\delta$, and \textbf{stop executing} \name \label{line:rec_es2_decide_after_EB2}

\medskip
\State \textbf{upon} $\mathit{timer}_2.\mathsf{expires}$: \BlueComment{after $O(t\delta)$ time} \label{line:rec_es2_timer2_expires} 
\State \hskip2em \textbf{trigger} $\mathsf{decide}(v_i)$, \textbf{wait for} $2\delta$, and \textbf{stop executing} \name \label{line:rec_es2_decide_after_timer2}\BlueComment{implies $f>n/3$}

\end{algorithmic}
\end{algorithm}

\smallskip
\noindent \textbf{Proof of correctness.}
As \name is a recursive distributed algorithm, we start by showing that \name executed among a single process is correct.

\begin{theorem} [\name is correct when $x = 1$]
If $x = 1$ (i.e., the system contains a single process), then $\name$ is a correct validated agreement algorithm.
\end{theorem}
\begin{proof}
Let $p_i$ be the only correct process. 
Strong unanimity follows from the fact that $p_i$ decides its proposal.
Due to the assumption that $p_i$'s proposal is valid, the external validity property is satisfied.
Finally, termination, integrity, and agreement trivially hold.
\end{proof}

The next lemma is crucial.
It proves that all correct processes behave according to assumptions specified in the specification of \name's building blocks.
Recall that we assume the following: (1) all correct processes propose to \name within $2\delta$ time of each other, and (2) all correct processes propose valid values to \name.
Moreover, the specification of \name needs to be ensured \emph{only if} there are less than $x / 3$ failures in $\mathsf{Entire}$; recall that $\mathsf{Entire}$ is the set of processes that operates in \name.

\begin{lemma} \label{lemma:big_lemma}
Let $\mathcal{O}$ (resp., $\mathcal{O}$') be any object in $\{ \mathcal{GC}_{su}, \mathcal{GC}_\ell, \mathcal{GC}_1, \mathcal{GC}_2, \mathcal{CB}_\ell, \mathcal{CB}_1, \mathcal{CB}_2, \mathcal{F}_{\ell}, \mathcal{F}_2  \}$ (resp., $\{ \mathcal{CB}_\ell, \mathcal{CB}_1, \mathcal{CB}_2, \mathcal{F}_{\ell}, \mathcal{F}_2 \}$).
The following holds:
\begin{compactitem}
    \item All correct processes input valid values only to $\mathcal{O}$ and they do so within $2\delta$ time of each other.

    \item If a correct process inputs a pair $(v, 1)$, for any value $v$, to any object $\mathcal{O}'$, no correct processes inputs a value different from $(v, \cdot)$ to $\mathcal{O}'$.

    \item All correct processes decide a valid value from \name and do so within $2\delta$ time of each other. 
\end{compactitem}
\end{lemma}
\begin{proof}
Let $p_i$ and $p_j$ denote any two correct processes.
To prove the lemma, we consider every considered object:
\begin{compactitem}
    \item $\mathcal{GC}_{su}$: Process $p_i$ proposes a valid value to $\mathcal{GC}_{su}$ (line~\ref{line:rec_es2_GCsv_propose}) as (1) $p_i$ forwards to $\mathcal{GC}_{su}$ its proposal to \name, and (2) $p_i$'s proposal to \name is valid.
    Moreover, all correct processes propose to $\mathcal{GC}_{su}$ within $2\delta$ time of each other as they all propose to \name within $2\delta$ time of each other.

    \item $\mathcal{CB}_{\ell}$:
    Process $p_i$ inputs a valid value $v$ to $\mathcal{CB}_{\ell}$ (line~\ref{line:rec_es2_EBl_propose}) as (1) $p_i$ decides $v$ from $\mathcal{GC}_{su}$, (2) the justification property of $\mathcal{GC}_{su}$ ensures that $v$ was proposed to $\mathcal{GC}_{su}$ by a correct process, and (3) only valid values are proposed to $\mathcal{GC}_{su}$ by correct processes.
    Moreover, all correct processes input to $\mathcal{CB}_{\ell}$ within $2\delta$ time of each other as all correct processes decide from $\mathcal{GC}_{su}$ within $2\delta$ time of each other (as $\mathcal{GC}_{su}$ is executed for a fixed time duration and all correct processes propose to $\mathcal{GC}_{su}$ within $2\delta$ time of each other).

    Suppose $p_i$ inputs a pair $(v, 1)$ to $\mathcal{CB}_\ell$, for some value $v$.
    That implies that $p_i$ has previously decided $(v, 1)$ from $\mathcal{GC}_{su}$.
    The consistency property of $\mathcal{GC}_{su}$ implies that $p_j$ decides $(v, \cdot)$ from $\mathcal{GC}_{su}$ and inputs the same pair to $\mathcal{CB}_{\ell}$.

    \item $\mathcal{GC}_{\ell}$:
    Let us consider two possibilities for $p_i$'s proposal $v$ to $\mathcal{GC}_{\ell}$:
    \begin{compactitem}
        \item Suppose $p_i$ proposes $v$ at line~\ref{line:rec_es2_GCl_propose_after_EBL}. 
        In this case, $v$ is a valid value due to the external validity property of $\mathcal{CB}_{\ell}$.

        \item Suppose $p_i$ proposes $v$ at line~\ref{line:rec_es2_GCl_propose_after_timer}.
        In this case, $v$ is $p_i$'s decision from $\mathcal{GC}_{su}$.
        Therefore, $v$ is valid as (1) it was proposed by a correct process to $\mathcal{GC}_{su}$ (due to the justification property of $\mathcal{GC}_{su}$), and (2) correct processes propose only valid values to $\mathcal{GC}_{su}$.    \end{compactitem}

    Let $p_k$ denote the first correct process to propose to $\mathcal{GC}_{\ell}$; let that happen at some time $\tau$.
    We distinguish two scenarios:
    \begin{compactitem}
        \item Let $p_k$ propose at line~\ref{line:rec_es2_GCl_propose_after_EBL}.
        In this case, the totality property of $\mathcal{CB}_{\ell}$ ensures that all correct processes propose to $\mathcal{CB}_\ell$ by time $\tau + 2\delta$.

        \item Let $p_k$ propose at line~\ref{line:rec_es2_GCl_propose_after_timer}.
        As all correct processes set their timer $\mathit{timer}_0$ within $2\delta$ time of each other, $\mathit{timer}_0$ at any correct process expires by time $\tau + 2\delta$.
    \end{compactitem}

    \item $\mathcal{F}_{\ell}$:
    Process $p_i$ inputs a valid value $v$ to $\mathcal{F}_{\ell}$ as (1) $p_i$ decides $v$ from $\mathcal{GC}_\ell$, (2) the justification property of $\mathcal{GC}_\ell$ ensures that $v$ was proposed by a correct process, and (3) only valid values are proposed to $\mathcal{GC}_\ell$ by correct processes.
    Moreover, all correct processes input to $\mathcal{F}_{\ell}$ within $2\delta$ time of each other, as all correct processes decide from $\mathcal{GC}_\ell$ within $2\delta$ time of each other.
    Finally, the consistency property of $\mathcal{GC}_\ell$ ensures that if a correct process inputs $(v, 1)$ to $\mathcal{F}_{\ell}$, for any value $v$, then all correct processes input $(v, \cdot)$.

    \item $\mathcal{GC}_1$: Process $p_i$ proposes a valid value $v$ as (1) $p_i$ decides $v$ from $\mathcal{GC}_\ell$, (2) the justification property of $\mathcal{GC}_\ell$ ensures that $v$ was proposed by a correct process, and (3) only valid values are proposed to $\mathcal{GC}_\ell$ by correct processes.
    Similarly, all correct processes propose to $\mathcal{GC}_1$ within $2\delta$ time of each other as all correct processes decide from $\mathcal{GC}_\ell$ within $2\delta$ time of each other and all correct processes wait for exactly $3\delta$ time before proposing to $\mathcal{GC}_1$.

    \item $\mathcal{CB}_1$: Process $p_i$ proposes a valid value $v$ to $\mathcal{CB}_1$ as (1) $p_i$ decides $v$ from $\mathcal{GC}_1$, (2) the justification property of $\mathcal{GC}_1$ ensures that $v$ was proposed to $\mathcal{GC}_1$ by a correct process, and (3) only valid values are proposed to $\mathcal{GC}_1$ by correct processes.
    Moreover, all correct processes input to $\mathcal{CB}_1$ within $2\delta$ time of each other, as all correct processes decide from $\mathcal{GC}_1$ within $2\delta$ time of each other.
    Lastly, the consistency property of $\mathcal{GC}_1$ ensures that if a correct process inputs $(v, 1)$ to $\mathcal{CB}_1$, for any value $v$, then all correct processes input $(v, \cdot)$.

    \item $\mathcal{GC}_2$: Let us consider two possibilities for $p_i$'s proposal $v$ to $\mathcal{GC}_{2}$:
    \begin{compactitem}
        \item Suppose $p_i$ proposes $v$ at line~\ref{line:rec_es2_GC2_propose_after_EB}. 
        In this case, $v$ is a valid value due to the external validity property of $\mathcal{CB}_{1}$.

        \item Suppose $p_i$ proposes $v$ at line~\ref{line:rec_es2_GC2_propose_after_timer}.
        In this case, $v$ is $p_i$'s decision from $\mathcal{GC}_1$.
        Thus, $v$ is valid as (1) it was proposed by a correct process to $\mathcal{GC}_1$ (due to the justification property of $\mathcal{GC}_1$), and (2) correct processes propose only valid values to $\mathcal{GC}_1$.
    \end{compactitem}

    Let $p_k$ denote the first correct process to propose to $\mathcal{GC}_{2}$; let that happen at some time $\tau$.
    We distinguish two scenarios:
    \begin{compactitem}
        \item Let $p_k$ propose at line~\ref{line:rec_es2_GC2_propose_after_EB}.
        In this case, the totality property of $\mathcal{CB}_{1}$ ensures that all correct processes propose to $\mathcal{GC}_2$ by time $\tau + 2\delta$.

        \item Let $p_k$ propose at line~\ref{line:rec_es2_GC2_propose_after_timer}.
        As all correct processes set their timer $\mathit{timer}_1$ within $2\delta$ time of each other, $\mathit{timer}_1$ at any correct process expires by time $\tau + 2\delta$.
    \end{compactitem}

    \item $\mathcal{F}_2$: Process $p_i$ inputs a valid value $v$ to $\mathcal{F}_2$ as (1) $p_i$ decides $v$ from $\mathcal{GC}_2$, (2) the justification property of $\mathcal{GC}_2$ ensures that $v$ was proposed by a correct process, and (3) only valid values are proposed to $\mathcal{GC}_2$ by correct processes.
    Moreover, all correct processes input to $\mathcal{F}_2$ within $2\delta$ time of each other, as all correct processes decide from $\mathcal{GC}_2$ within $2\delta$ time of each other.
    Finally, the consistency property of $\mathcal{GC}_2$ ensures that if a correct process inputs $(v, 1)$ to $\mathcal{F}_2$, for any value $v$, then all correct processes input $(v, \cdot)$.

    \item $\mathcal{CB}_2$: Process $p_i$ proposes a valid value $v$ to $\mathcal{CB}_2$ as (1) $p_i$ decides $v$ from $\mathcal{GC}_2$, (2) the justification property of $\mathcal{GC}_2$ ensures that $v$ was proposed to $\mathcal{GC}_2$ by a correct process, and (3) only valid values are proposed to $\mathcal{GC}_2$ by correct processes.
    Moreover, all correct processes input to $\mathcal{CB}_2$ within $2\delta$ time of each other, as all correct processes decide from $\mathcal{GC}_2$ within $2\delta$ time of each other.
    Lastly, the consistency property of $\mathcal{GC}_2$ ensures that if a correct process inputs $(v, 1)$ to $\mathcal{CB}_2$, for any value $v$, then all correct processes input $(v, \cdot)$.
\end{compactitem}

Lastly, we prove that all correct processes decide a valid value within $2\delta$ time of each other.
We first prove that all correct processes decide a valid value. Let us consider all possibilities for $p_i$'s decision $v$:
    \begin{compactitem}
        \item Suppose $p_i$ decides $v$ at line~\ref{line:recursive_early_stopping2_decide_after_first_finisher}.
        The preservation property of $\mathcal{F}_{\ell}$ ensures that $v$ was input to $\mathcal{F}_\ell$.
        As all correct processes input only valid values to $\mathcal{F}_\ell$, $v$ must be valid.
    
        \item Suppose $p_i$ decides $v$ at line~\ref{line:rec_es2_decide_after_EB2}. 
        In this case, $v$ is a valid value due to the external validity property of $\mathcal{CB}_{2}$.

        \item Suppose $p_i$ decides $v$ at line~\ref{line:rec_es2_decide_after_timer2}.
In this case, $v$ is the value decided from $\mathcal{GC}_2$.
Therefore, $v$ is a valid value due to (1) the justification property of $\mathcal{GC}_2$, and (2) the fact that correct processes only propose valid values to $\mathcal{GC}_2$.
    \end{compactitem}

    We now prove all correct processes decide within $2\delta$ time of each other. Let $p_k$ denote the first correct process to decide; let that happen at some time $\tau$.
    We distinguish three scenarios:
    \begin{compactitem}
        \item Let $p_k$ decide at line~\ref{line:recursive_early_stopping2_decide_after_first_finisher}.
        In this case, the totality property of $\mathcal{F}_{\ell}$ ensures that all correct processes decide by time $\tau + 2\delta$.
    
        \item Let $p_k$ propose at line~\ref{line:rec_es2_decide_after_EB2}.
        In this case, the totality property of $\mathcal{CB}_{2}$ ensures that all correct processes decide by time $\tau + 2\delta$.

        \item Let $p_k$ propose at line~\ref{line:rec_es2_decide_after_timer2}.
        As all correct processes set their timer $\mathit{timer}_2$ within $2\delta$ time of each other, $\mathit{timer}_2$ at any correct process expires by time $\tau + 2\delta$.
    \end{compactitem}
    Let us remark that no process stops before every correct process decides, because of the $2\delta$ waiting step, at lines \ref{line:rec_es2_halt1}, \ref{line:rec_es2_halt2}, \ref{line:rec_es2_decide_after_EB2} and \ref{line:rec_es2_decide_after_timer2}.
Thus, the lemma holds.
\end{proof}

The following theorem proves that \name satisfies \strongVal.

\begin{theorem} [Strong unanimity] \label{theorem:recursive_early_stopping_strong_validity}
\name (\Cref{algorithm:recursive_early_stopping_king_2}) satisfies \strongVal. 
\end{theorem}
\begin{proof}
Suppose all correct processes propose the same value $v$; recall that all correct processes do so within $2\delta$ time of each other.
Hence, all correct processes decide $(v, 1)$ from $\mathcal{GC}_{su}$, which implies that all correct processes input $(v, 1)$ to $\mathcal{CB}_\ell$.
The stability property of $\mathcal{CB}_\ell$ ensures that all correct processes propose $v$ to $\mathcal{GC}_{\ell}$, which implies that every correct process decides $(v, 1)$ from $\mathcal{GC}_{\ell}$ (due to the strong unanimity property of $\mathcal{GC}_\ell$).
Thus, all correct processes propose $(v, 1)$ to $\mathcal{F}_{\ell}$.
Let $\tau$ denote the first time a correct process proposes to $\mathcal{F}_{\ell}$.
Lemma~\ref{lemma:big_lemma} proves that all correct processes propose to $\mathcal{F}_{\ell}$ by time $\tau + 2\delta$.
Moreover, the liveness property of $\mathcal{F}_{\ell}$ ensures that all correct processes output $v$ from $\mathcal{F}_{\ell}$ (and decide $v$ from \name) by time $\tau + 3\delta$, which is before any correct process finishes its $3\delta$ waiting (line~\ref{line:rec_es2_wait_1}).
Therefore, \strongVal is ensured.
\end{proof}

Next, we prove the external validity property.

\begin{theorem} [External validity]
\name (\Cref{algorithm:recursive_early_stopping_king_2}) satisfies external validity.
\end{theorem} \label{theorem:recursive_early_stopping_external_validity}
\begin{proof}
The property is ensured due to Lemma~\ref{lemma:big_lemma}.
\end{proof}

We now prove that if process $p_\ell$ is correct, then all correct processes agree on a common value after $\mathcal{F}_{\ell}$ (without executing $\mathcal{GC}_1$).

\begin{lemma} \label{lemma:deciding_immediately_with_a_correct_leader}
 If process $p_\ell$ is correct, correct processes agree on a common value after the first finisher $\mathcal{F}_{\ell}$ at line~\ref{line:recursive_early_stopping2_decide_after_first_finisher} in $O(\delta)$ time (and without executing $\mathcal{GC}_1$).
\end{lemma}
\begin{proof}
    Suppose $p_\ell$ is correct. Then, every correct process outputs the same value $v$ from $\mathcal{CB}_\ell$ (line~\ref{line:rec_es2_EBl_output}) before $\mathit{timer}_0$'s expiration (line~\ref{line:rec_es2_timer0_expires}) due to $\mathcal{CB}_\ell$'s optimistic consensus property. Hence, all correct processes propose $v$ to $\mathcal{GC}_\ell$ within $2\delta$ time of each other (line~\ref{line:rec_es2_GC1_propose}). Due to $\mathcal{GC}_\ell$'s strong unanimity, every correct process outputs $(v,1)$ from $\mathcal{GC}_\ell$ (line~\ref{line:rec_es2_GC1_decide}). This means that all correct processes propose $(v,1)$ to $\mathcal{F}_{\ell}$ within $2\delta$ time of each other (line~\ref{line:rec_es2_fin1_to_finish}). Finally, due to $\mathcal{F}_{\ell}$'s liveness, every correct process outputs $v$ from $\mathcal{F}_{\ell}$ (line~\ref{line:rec_es2_fin1_finish}), decides $v$ and stops at line~\ref{line:recursive_early_stopping2_decide_after_first_finisher} before the end of the $3\delta$-waiting step (line~\ref{line:rec_es2_wait_1}).
\end{proof}

The next lemma proves a connection between (1) $\mathcal{F}_\ell$ and $\mathcal{GC}_\ell$, and (2) $\mathcal{F}_2$ and $\mathcal{GC}_2$. 

\begin{lemma}\label{lemma:deciding_via_finisher_implies_unanimity}
    If a correct process decides some value $v$ after  $\mathcal{F}_{\ell}$ at line~\ref{line:recursive_early_stopping2_decide_after_first_finisher} (resp., $\mathcal{F}_2$ at line~\ref{line:recursive_early_stopping2_decide_after_second_finisher}), then every correct process that decides from $\mathcal{GC}_\ell$ (resp., $\mathcal{GC}_2$) does so with a pair $(v, \cdot)$.
\end{lemma}
\begin{proof}
    Assume a correct process decides $v$ at line~\ref{line:recursive_early_stopping2_decide_after_first_finisher} (resp., line~\ref{line:recursive_early_stopping2_decide_after_second_finisher})  upon getting an output at line~\ref{line:rec_es2_fin1_finish} (resp., line~\ref{line:rec_es2_fin2_finish}) from $\mathcal{F}_{\ell}$ (resp., $\mathcal{F}_2$). By finisher's justification and preservation, some correct process inputs $(v,1)$ to $\mathcal{F}_{\ell}$ (resp., $\mathcal{F}_2$) at line~\ref{line:rec_es2_fin1_to_finish} (resp., line~\ref{line:rec_es2_fin2_to_finish}). By the consistency property of $\mathcal{GC}_\ell$ (resp., $\mathcal{GC}_2$), every correct process that outputs a pair from $\mathcal{GC}_\ell$ (resp., $\mathcal{GC}_2$) at line~\ref{line:rec_es2_GCl_decide} (resp., line~\ref{line:rec_es2_GC2_decide}) does so with a pair $(v, \cdot)$.
\end{proof}

Next, we prove that correct processes decide the same value within $2\delta$ time of each other if a correct process decides after $\mathcal{F}_\ell$.

\begin{lemma}\label{lemma:first_finisher_preserves_agreement}
    If a correct process decides some value $v$ at time $\tau$ after $\mathcal{F}_{\ell}$ at line~\ref{line:recursive_early_stopping2_decide_after_first_finisher}, then every correct process decides $v$ by time $\tau + 2\delta$.
\end{lemma}
\begin{proof}
    If a correct process decides some value $v$ at time $\tau$ after the first finisher $\mathcal{F}_{\ell}$ at line~\ref{line:recursive_early_stopping2_decide_after_first_finisher}, then Lemma~\ref{lemma:deciding_via_finisher_implies_unanimity} implies that every correct process that outputs from $\mathcal{GC}_\ell$ does so with a pair $(v, \cdot)$. Thus, due to $\mathcal{GC}_1$'s strong unanimity, no correct process outputs a pair different from $(v,1)$ from $\mathcal{GC}_1$. Hence, no correct process can decide a value different from $v$, due to strong unanimity of $\mathcal{GC}_2$ and stability of $\mathcal{CB}_{2}$. Moreover, by $\mathcal{F}_{\ell}$'s totality, every correct process decides $v$ by time $\tau + 2\delta$. 
\end{proof}

The following lemma proves that if the first half is ``healthy'', then all correct processes agree on a value.
In the following, we define $\mathcal{L}(x,f_x)$ as the latency of \name, executed among $x$ processes with $f_x$ actual failures.

\begin{lemma}\label{lemma:early_stopping_with_an_healthy_first_half}
Let $f_1$ denote the actual number of failures in $\mathcal{H}_1$ and let $f_1 < |\mathcal{H}_1| / 3$.
Then, all correct processes decide the same value $v$ from \name and they so in $O(\delta) + \mathcal{L}(|\mathcal{H}_1|, f_1)$ time since starting \name (and without executing $\mathcal{CB}_2$).
\end{lemma}
\begin{proof}
If a correct process decides some value $v$ at time $\tau$ after $\mathcal{F}_{\ell}$ at line~\ref{line:recursive_early_stopping2_decide_after_first_finisher}, then every correct process decides $v$ by time $\tau + 2\delta$ (by Lemma~\ref{lemma:first_finisher_preserves_agreement}).
Suppose no correct process decides after $\mathcal{F}_{\ell}$ at line~\ref{line:recursive_early_stopping2_decide_after_first_finisher} by time $\tau \in O(\delta) + \mathcal{L}(|\mathcal{H}_{1}|, f_1)$.
Hence, all correct processes participate in $\mathcal{CB}_{1}$ and they start doing so within $2\delta$ of each other; let all correct processes start $\mathcal{CB}_1$ by some time $\tau'$.
Due to the optimistic consensus property of $\mathcal{CB}_{1}$, all correct processes output (line~\ref{line:rec_es2_EB1_output}) the same value $v'$ from $\mathcal{CB}_{1}$ by time $\tau' + \mathcal{L}(|\mathcal{H}_{1}|, f_1) + 7\delta$. Thus, all correct processes propose $v'$ to $\mathcal{GC}_2$ (line~\ref{line:rec_es2_GC2_propose_after_EB}), which means, due to $\mathcal{GC}_2$'s strong 
unanimity and $\mathcal{F}_2$'s liveness that all correct processes decide (line~\ref{line:recursive_early_stopping2_decide_after_second_finisher}) a common value by time $O(\delta) + \mathcal{L}(|\mathcal{H}_{1}|, f_1)$ before the end of the $3\delta$-waiting step (line~\ref{line:rec_es2_wait2}).

In all cases, correct processes decide a common value by time $O(\delta) + \mathcal{L}(|\mathcal{H}_{1}|, f_1)$ without executing $\mathcal{CB}_{2}$.
Therefore, the lemma holds.
\end{proof}

Finally, we are ready to prove the agreement property of \name.

\begin{theorem} [Agreement] \label{lemma:rec_es_agreement}
\name satisfies agreement.
\end{theorem}
\begin{proof}
If $p_\ell$ is correct or $\mathcal{H}_{1}$ is ``healthy'' (containing less than one-third of faulty processes), the lemma holds due to Lemma~\ref{lemma:deciding_immediately_with_a_correct_leader} or Lemma~\ref{lemma:early_stopping_with_an_healthy_first_half}, respectively. 
Assume $p_\ell$ is faulty and $\mathcal{H}_{1}$ is not healthy. Hence, $\mathcal{H}_{2}$ must be healthy.
If a correct process decides some value $v$ at time $\tau$ after $\mathcal{F}_{\ell}$ at line~\ref{line:recursive_early_stopping2_decide_after_first_finisher}, then every correct process decides $v$ by time $\tau + 2\delta$ (by Lemma~\ref{lemma:first_finisher_preserves_agreement}).
Assume no correct process decides after $\mathcal{F}_{\ell}$ at line~\ref{line:recursive_early_stopping2_decide_after_first_finisher}.

If some correct process decides after $\mathcal{F}_2$ at line~\ref{line:recursive_early_stopping2_decide_after_second_finisher} by time $\tau_2$, then, by Lemma~\ref{lemma:deciding_via_finisher_implies_unanimity}, every correct process outputs $(v,\cdot)$ from $\mathcal{GC}_{2}$ (line~\ref{line:rec_es2_GC2_decide}). Then, due to $\mathcal{CB}_{2}$'s optimistic consensus property, every correct process either decides $v$ at line~\ref{line:rec_es2_decide_after_EB2} (obtained from  $\mathcal{CB}_{2}$ at line~\ref{line:rec_es2_EB2_output}) or first decides $v$ at line~\ref{line:recursive_early_stopping2_decide_after_second_finisher} after the second finisher $\mathcal{F}_2$.

Assume no correct process decides after the second finisher $\mathcal{F}_2$ at line~\ref{line:recursive_early_stopping2_decide_after_second_finisher}. Then, by $\mathcal{CB}_{2}$'s optimistic consensus, all correct processes agree at line~\ref{line:rec_es2_decide_after_EB2} on some value.
\end{proof}

We next prove that \name satisfies integrity.

\begin{theorem} [Integrity]
\name (\Cref{algorithm:recursive_early_stopping_king_2}) satisfies integrity.
\end{theorem}
\begin{proof}
The integrity property follows directly from the algorithm, as no correct process decides more than once.
\end{proof}

Finally, we prove the termination property of \name.

\begin{theorem} [Termination]
\name (\Cref{algorithm:recursive_early_stopping_king_2}) satisfies termination.
\end{theorem}
\begin{proof}
The termination property follows directly from Lemma~\ref{lemma:big_lemma}.
\end{proof}
    
\smallskip
\noindent \textbf{Proof of complexity.}
We start by proving \name's bit complexity.

\begin{theorem} [Bit complexity] \label{theorem:nameES_complexity}
The bit complexity of \name (\Cref{algorithm:recursive_early_stopping_king_2}) is $O\big( (nL + n^2) \log n \big)$.
\end{theorem}
\begin{proof}
To analyze the (total) bit complexity of \name, we focus on the per-process bit complexity of \name with $f$ failures.
(We underline that $f$ is the actual number of failures and it is not known to processes.)
Let us denote by $\mathcal{B}(n, f)$ the per-process bit complexity of \name executed among $n$ processes out of which exactly $f$ processes are faulty.
We express $\mathcal{B}(n, f)$ through three specific cases: (1) $n = 1$, (2) the predetermined leader $p_\ell$ is a correct process, and (3) the predetermined leader $p_\ell$ is faulty.
Moreover, let $f_1$ (resp., $f_2$) denote the actual number of faulty processes in $\mathcal{H}_1$ (resp., $\mathcal{H}_2$).
Note that $f = f_1 + f_2$ if $p_\ell$ is correct and $f = f_1 + f_2 + 1$ if $p_\ell$ is faulty.

The per-process bit complexity verifies the following inequality, where $a$ and $b$ are some positive constants related to the different modules (with $aL + b\log (n) \geq 4\mathcal{B}_{\mathcal{GC}}(n) + 3\mathcal{B}_{\mathcal{CD}}(n) + 2\mathcal{B}_{\mathcal{F}}(n)$, where $\mathcal{B}_{\mathcal{M}}(n)$ denotes the worst-case per-process bit complexity of module $\mathcal{M}$ among $n$ processes):

\begin{equation*}
\mathcal{B}(n,f) \leq
\begin{cases}
0, & \text{if $n=1$,} \\
a L + b  \log n, & \text{if $p_\ell$ is correct,} \\
a L + b  \log n + \max \big( \mathcal{B}(\frac{n}{2}, f_1), \mathcal{B}(\frac{n}{2}, f_2) \big), & \text{otherwise (with } f = 1 + f_1 + f_2 \text{)}.
\end{cases}
\end{equation*}

By induction, we show that, for every integer $n$, $\mathcal{B}(n,f) \leq c (L+n) \log n$ for some sufficiently large positive constant $c \geq \max(2b, a / \log 2)$. The base case trivially holds since $\mathcal{B}(1,0) = 0$. Let $n$ be a strictly positive integer. We assume that the result holds for any integer strictly lower than $n$.
By the induction hypothesis,
$\max \big(\mathcal{B}(\frac{n}{2}, f_1), \mathcal{B}(\frac{n}{2}, f_2)\big) \leq c(L+\frac{n}{2})\log \frac{n}{2} $.
Hence, $\mathcal{B}(n,f) \leq (a-c\log  2+c\log n)L + (b+\frac{c}{2})n\log n$.
Finally, $\mathcal{B}(n,f) \leq c (L+n) \log n$, since $c \geq \mathsf{max}(2b, a/ \log 2)$. The induction holds, which concludes the proof for the per-process bit complexity.
\end{proof}

Lastly, we prove \name's latency complexity, i.e., the fact that \name is early-stopping.

\begin{theorem} [Early-stopping]
\name (\Cref{algorithm:recursive_early_stopping_king_2}) is early-stopping.
\end{theorem}
\begin{proof}
Let us denote by $\mathcal{L}(n, f)$ the latency complexity of \name executed among $n$ processes out of which exactly $f$ processes are faulty.
(We underline that $f$ is the actual number of failures and is unknown to processes.)
Let $f_1$ (resp., $f_2$) denote the actual number of faulty processes in $\mathcal{H}_1$ (resp., $\mathcal{H}_2$).
Note that $f = f_1 + f_2$ if $p_\ell$ is correct and $f = f_1 + f_2 + 1$ if $p_\ell$ is faulty.
Let $\mathcal{L}_{\mathcal{M}}$denote the worst-case latency of module $\mathcal{M}$ given exactly $f$ failures (independent of $n$); this module $\mathcal{M}$ can be graded consensus, committee dissemination and finisher.
Let $a'$, $b'$ and $c'$ be some positive constants verifying:
\begin{equation*}
\begin{cases}
a' = (1/\delta) \cdot \big(2 \cdot 3 \cdot \mathcal{L}_{\mathcal{GC}} + \mathcal{L}_{\mathcal{CD}} + \mathcal{L}_{\mathcal{F}} + 2\big) = (2 \cdot 3 \cdot 7 + 7 + 1 + 2) = 52 \\

b' = (1/\delta) \cdot \big( 2 \cdot 3 \cdot \mathcal{L}_{\mathcal{GC}} + \mathcal{L}_{\mathcal{CD}} + \mathsf{max}(\mathcal{L}_{\mathcal{CD}}, \mathcal{L}_{\mathcal{F}}) \big) = (2 \cdot 3 \cdot 7 + 7 + 7)  =  56  \\
c' = 2a' + b' = 160.
\end{cases}
\end{equation*}

The latency of \name can be expressed as follows:
\begin{equation*}
\mathcal{L}(n,f) \leq
\begin{cases}
a' \delta, & \text{if $p_\ell$ is correct,} 
\\ (a'+b') \delta + \mathcal{L}(\frac{n}{2},f_1) + \mathcal{L}(\frac{n}{2},f_2), & \text{otherwise.}
\end{cases}
\end{equation*}
 
By induction, we show that $\mathcal{L}(n,f) \leq (a'+c'f)\delta$. 
Since  $\mathcal{L}(n,0) \leq a'\delta$, for any $n$,  the base case trivially holds. 
Let $n, f$ be two strictly positive integers.
We assume $\mathcal{L}(n',f') \leq (a'+c'f')\delta$, for any pair $(n',f')$ such that either (1) $n' < n$ and $f' \leq f$, or (2) $n' \leq n$ and $f' < f$, and aim to show it implies the result for the pair $(n,f)$. 

\begin{compactenum}
    \item Let $p_\ell$ be correct.
    In this case, the result immediately holds.

    \item Let $p_\ell$ be faulty. 
    In this case, $f=1+f_1+f_2$. Hence, by the induction hypothesis, the following holds: (1) $\mathcal{L}(\frac{n}{2},f_1) \leq (a'+c'f_1)$, and (2) $\mathcal{L}(\frac{n}{2},f_2) \leq (a'+c'f_2)$.
    Thus, $\mathcal{L}(\frac{n}{2},f_1) + \mathcal{L}(\frac{n}{2},f_2) \leq (2a'-c'+c'f)$.
    This further implies $\mathcal{L}(n,f) \leq (a'+b') \delta + \mathcal{L}(\frac{n}{2},f_1) + \mathcal{L}(\frac{n}{2},f_2) \leq a' \delta + (b'+2a'-c'+c'f)\delta \leq a' \delta + c'f\delta$.
\end{compactenum}
The theorem holds.
\end{proof}

\newpage
\setcounter{tocdepth}{2}
\tableofcontents

\end{document}